\title{Polynomial-time computation of homotopy groups and Postnikov
systems in fixed dimension\thanks{This research 
was supported by the  ERC Advanced Grant No.~267165.
The research of M.~\v{C}.~was
supported by the project CZ.1.07/2.3.00/20.0003 of the
Operational Programme Education for Competitiveness of the Ministry of
Education, Youth and Sports of the Czech Republic.
The research by M.\,K.\ and J.\,M.\ was
supported by the Center of Excellence -- Inst.\ for Theor.\
        Comput.\ Sci., Prague (project P202/12/G061 of GA~\v{C}R).
The research of L.~V.~was supported by the Center of Excellence --
Eduard \v{C}ech Institute (project P201/12/G028 of GA~\v{C}R).
The research by U.\,W.\ was supported by the Swiss National Science Foundation
(grants SNSF-200020-138230 and SNSF-PP00P2-138948).
 }
}

\newif\ifcmts
\cmtstrue

\documentclass[11pt,a4paper]{article}

\usepackage{ifthen}
\newboolean{MareksFormat}
\InputIfFileExists{MareksSwitch}{}{}
\ifthenelse{\boolean{MareksFormat}}{
 \usepackage{geometry}

\geometry{papersize={135mm,180mm},top=2.5mm,bottom=2.5mm,left=1.5mm,right=1.5mm}
}{\usepackage{a4wide}}

\usepackage{graphicx}
\usepackage{amsmath}
\usepackage{amsthm}
\usepackage{amssymb}
\usepackage{bbm}
\usepackage{url}
\usepackage[all,cmtip]{xy}
\usepackage{mathtools}

\newcommand{\ra}{\rightarrow}
\newlength{\hlp}
\newcommand{\leftbox}[2]{\settowidth{\hlp}{$#1$}\makebox[\hlp][l]{${#1}{#2}$}}

\SelectTips{cm}{10}

\theoremstyle{plain}
\newtheorem{theorem}{Theorem}[section]
\newtheorem{lemma}[theorem]{Lemma}
\newtheorem{corol}[theorem]{Corollary}

\newtheorem{prop}[theorem]{Proposition}

\newtheorem{defn}[theorem]{Definition}
\theoremstyle{definition}

\DeclareMathOperator{\im}{im} 

\DeclareMathOperator{\SM}{SMap}

\DeclareMathOperator{\sset}{SSet}
\DeclareMathOperator{\rank}{rank}
\DeclareMathOperator{\Hom}{Hom}
\DeclareMathOperator{\Barr}{Bar}
\DeclareMathOperator{\Bas}{Bas}
\DeclareMathOperator{\ev}{ev}
\DeclareMathOperator{\fib}{fib}
\DeclareMathOperator{\ind}{ind}

\renewcommand\:{\colon}
\newcommand{\alterdef}[1]{\!\left\{\!\!\begin{array}{ll}
                                   #1 \end{array}  \right. }

\newcommand{\Z}{\mathbbm{Z}}
\newcommand{\N}{\mathbbm{N}}
\newcommand{\Q}{\mathbbm{Q}}
\newcommand{\R}{\mathbbm{R}}

\newcommand{\II}{\mathcal{I}}
\newcommand{\JJ}{\mathcal{J}}
\newcommand{\MM}{\mathcal{M}}

\newcommand{\FSC}{\mathcal{FSC}}
\newcommand{\FSS}{\mathcal{FSS}}

\newcommand{\EC}{\mathit{EC}}
\newcommand{\EtC}{\widetilde{\mathit{EC}}}
\newcommand{\EZ}{\mathit{EZ}^M}
\newcommand{\EF}{\mathit{EF}}
\newcommand{\ET}{\mathit{ET}}
\newcommand{\EB}{\mathit{EB}}

\newcommand{\EM}{\mathit{EM}}
\newcommand{\tEM}{\widetilde{\EC}^M}

\newcommand{\AB}{\mathrm{Ab}}

\newcommand\makevec[1]{{\bf #1}}

\def \mm {\makevec{m}}

\newcommand\kkk{{\mathbf{k}}}

\newcommand\thedim{k}
\newcommand\thedimm{k} 
\newcommand\otherdim{\ell}
\newcommand\diff{d}
\newcommand\cobo{\delta}
\newcommand\tC{\tilde{C}}

\newcommand\hC{\hat{C}}
\newcommand\hT{\hat{T}}
\newcommand\ef{\mathrm{ef}}

\newcommand{\dt}[1]{%
{\tilde{\raisebox{-.14ex}{\phantom{#1}}}}\mathllap{%
{\tilde{\raisebox{.16ex}{\phantom{#1}}}}\mathllap{#1}}
}

\newcommand\dtC{\dt{C}}

\newcommand\redu{\Rightarrow\!\!\!\!\Rightarrow}
\newcommand\reduP{\stackrel{{\scriptscriptstyle \rm P~\,}}{\redu}}
\newcommand\lredu{\Leftarrow\!\!\!\!\Leftarrow}
\newcommand\lreduP{\stackrel{{\scriptscriptstyle \rm ~~~P\,}}{\lredu}}
\newcommand{\steq}{\Leftarrow\!\!\!\!\Leftarrow\!\!\!\!\Rightarrow\!\!\!\!\Rightarrow}
\newcommand{\steqP}{\stackrel{{\scriptscriptstyle \rm \,P}}{\steq}}

\def\EML{\mathrm{EML}}
\def\AW{\mathrm{AW}}
\def\SHI{\mathrm{SHI}}
\newcommand\intg{\iota}

\DeclareMathOperator{\id}{id}   
\newcommand\ndg{{\rm ndg}}
\DeclareMathOperator{\size}{{\sf size}}
\DeclareMathOperator{\enc}{{\sf enc}}

\DeclareMathOperator\MCyl{Cyl}
\DeclareMathOperator\MCone{Cone}
\newcommand\MCo{M} 

\long\def\onefigure#1#2{
\begin{figure*}[tbp]
\begin{center}
#1
\end{center}
\caption{#2}
\end{figure*}
}
\def\immediateFigure#1{%
\smallskip\begin{center}#1\end{center}\smallskip }
\newcommand{\labfig}[2]  
{\onefigure{\mbox{\includegraphics{#1}}}{\label{f:#1} #2} }

\newcommand{\labfigw}[3]  
{\onefigure{\mbox{\includegraphics[width=#2]{#1}}}{\label{f:#1} #3}}

\newcommand{\immfig}[1]  
{\immediateFigure{\mbox{\includegraphics{#1}}}}

\newcommand{\immfigw}[2] 
{\immediateFigure{\mbox{\includegraphics[width=#2]{#1}}}}

\def\indef#1{\emph{#1}}

\newcommand{\heading}[1]{\vspace{1ex}\par\noindent{\bf\boldmath #1}}

\ifcmts
\newcommand{\marrow}{\marginpar{\boldmath$\longleftarrow$}}
\newcommand{\jirka}[1]{\ifhmode\newline\fi\marrow \textsf{*** (JIRKA: ) #1\newline}}
\newcommand{\marek}[1]{\ifhmode\newline\fi\marrow \textsf{*** (MAREK: ) #1\newline}}
\newcommand{\lukas}[1]{\ifhmode\newline\fi\textsf{*** (LUKAS: ) \marrow #1\newline}}
\newcommand{\martin}[1]{\ifhmode\newline\fi\textsf{*** (MARTIN: ) \marrow #1\newline}}

\else
\newcommand{\marrow}{}
\newcommand{\jirka}[1]{}
\newcommand{\marek}[1]{}
\newcommand{\lukas}[1]{}
\newcommand{\martin}[1]{}

\fi

\def\kamsymb{{\rm b}}
\def\ethsymb{{\rm c}}

\def\massymb{{\rm a}}
\def\epflsymb{{\rm d}}

\author{
Martin \v{C}adek$^\massymb$
\and
Marek Kr\v{c}\'al$^{\kamsymb}$ \and
Ji\v{r}\'{\i} Matou\v{s}ek$^{\kamsymb, \ethsymb}$
\and Luk\'a\v{s} Vok\v{r}\'{\i}nek$^\massymb$
 \and Uli Wagner$^{\epflsymb}$
}

\begin{document}

\maketitle

{\renewcommand\thefootnote{\massymb}
\footnotetext{Department of Mathematics and Statistics,
 Masaryk University, Kotl\'a\v{r}sk\'a~2, 611~37~Brno,
 Czech Republic}
}

{\renewcommand\thefootnote{\kamsymb}
\footnotetext{Department of Applied Mathematics,
Charles University, Malostransk\'{e} n\'{a}m.~25,
118~00~~Praha~1,  Czech Republic}
}
{\renewcommand\thefootnote{\ethsymb}
\footnotetext{Department of Computer Science,
ETH Zurich, 8092~Zurich, Switzerland}
}


{\renewcommand\thefootnote{\epflsymb}
\footnotetext{Institut de Math\'{e}matiques de G\'{e}om\'{e}trie et Applications, {\'E}cole Polytechnique F\'{e}d\'{e}rale de Lausanne, EPFL SB MATHGEOM, MA C1 553, Station 8, 1015 Lausanne, Switzerland}}


\begin{abstract}
For several computational
problems in homotopy theory,
we obtain algorithms  with running time  polynomial in
the input size.
In particular, for
every \emph{fixed} $\thedim\ge 2$, there is a polynomial-time algorithm that,
for a $1$-connected topological space $X$ given as a finite simplicial
complex, or more generally, as a simplicial set with polynomial-time
homology, computes the $\thedim$th homotopy group $\pi_\thedim(X)$,
as well as the first $\thedim$ stages of a Postnikov system of $X$.
 Combined with results of an earlier paper,
this yields a polynomial-time computation
of $[X,Y]$, i.e., all homotopy classes of continuous mappings $X\to Y$,
under the assumption that
$Y$ is $(\thedim{-}1)$-connected and $\dim X\le 2\thedim-2$.
We also obtain a polynomial-time solution of the
\emph{extension problem}, where the input consists
of finite simplicial complexes $X,Y$, where $Y$ is $(\thedim{-}1)$-connected
and $\dim X\le 2\thedim-1$,
plus a subspace $A\subseteq X$ and a (simplicial) map $f\:A\to Y$,
and the question is the extendability of $f$ to all of~$X$.

The algorithms are based on the notion of a
\emph{simplicial set with polynomial-time homology},
which is an enhancement of the notion of a simplicial set with
effective homology developed earlier by Sergeraert and his co-workers.
Our polynomial-time algorithms are obtained by showing that
simplicial sets with polynomial-time homology are closed under
various operations, most notably, Cartesian products, twisted
Cartesian products, and classifying space.
One of the key components is also polynomial-time homology for
the Eilenberg--MacLane  space $K(\Z,1)$, provided in another recent paper
by Kr\v{c}\'al, Matou\v{s}ek, and Sergeraert.
\end{abstract}
\clearpage

\section{Introduction}

 One of the central themes in algebraic
topology is understanding the structure of all continuous
maps $X\to Y$, for given topological spaces $X$ and $Y$
(all maps between topological spaces in this paper are assumed to
be continuous). Here two maps $f,g\:X\to Y$ are usually
considered equivalent if they are
\emph{homotopic}\footnote{Homotopy means a continuous deformation
of one map into another. More precisely,
$f$ and $g$ are defined to be homotopic, in symbols $f\sim g$,
if there is a continuous $F\colon X\times [0,1]\to Y$
such that $F(\cdot,0)=f$ and $F(\cdot,1)=g$.
With this notation, $[X,Y]=\{[f]:f\colon X\to Y\}$,
where $[f]=\{g: g\sim f\}$ is the \emph{homotopy class} of~$f$.
};
thus, the object of interest is $[X,Y]$,
the set of all homotopy classes of maps $X\to Y$.

\heading{Computing higher homotopy groups. }
Many of the celebrated results throughout the history of topology
can be cast as information about $[X,Y]$
for certain spaces $X$ and~$Y$. In particular,
one of the important challenges propelling the research in
algebraic topology has been the computation of the
\emph{homotopy groups of spheres}\footnote{The $\thedim$th homotopy group
$\pi_\thedim(Y)$ of a space $Y$ is defined as the set of all homotopy classes
of \emph{pointed} maps $f\:S^\thedim\to Y$, i.e.,
maps $f$ that send a distinguished
point $s_0\in S^\thedim$ to a distinguished point $y_0\in Y$ (and the homotopies $F$ also satisfy $F(s_0,t)=y_0$ for all $t\in[0,1]$). Strictly speaking, one should
really write $\pi_\thedim(Y,y_0)$ but for a path-connected $Y$, the choice of
$y_0$ does not matter. Moreover, for $1$-connected $Y$
the pointedness of the maps
does not matter either and one can identify $\pi_\thedim(Y)$
with $[S^\thedim,Y]$. Each $\pi_\thedim(Y)$ is a group, which for $\thedim\ge 2$
is Abelian, but the definition of the group operation is not important
for us at the moment.
 }
$\pi_\thedim(S^n)$, where only partial results have been obtained in spite
of an enormous effort (see, e.g., \cite{Ravenel,Kochman}).

Our concern here is the (theoretical) \emph{complexity} of
computing homotopy groups $\pi_\thedim(Y)$ for an arbitrary~$Y$.
It is well known that the fundamental group $\pi_1(Y)$ is
uncomputable, as follows from undecidability of the word
problem in groups \cite{Novikov:UndecidabilityWordProblem-1955}.%
\footnote{
The undecidability  of the word problem holds even for the fundamental groups of \(2\)-complexes or \(4\)-manifolds. On the other
hand, the problem is decidable for certain classes of manifolds \cite{Manning:mani-with-word-problem, Dehn:comput-fund-group-2-mani}.}
On the other hand, given a $1$-connected space $Y$
(i.e., one with $\pi_1(Y)$ trivial), say represented as a finite simplicial
complex, there are algorithms that compute the higher homotopy
group $\pi_\thedim(Y)$, for every given~$\thedim\ge 2$.
The first such algorithm is due to Brown  
\cite{Brown},
and newer ones have been obtained as a part of general
computational frameworks in algebraic topology
due to Sch\"on \cite{Schoen-effectivetop}, Smith
\cite{smith-mstructures}, and Sergeraert and his
co-workers (e.g.,
 \cite{Sergeraert:ComputabilityProblemAlgebraicTopology-1994,RubioSergeraert:ConstructiveAlgebraicTopology-2002,RomeroRubioSergeraert,SergRub-homtypes}).
In particular, an algorithm based on the methods
of Sergeraert et al.\ can be found in Real \cite{Real96}.
We also refer to Romero and Sergeraert \cite{SergRomEffHmtp}
for a new approach to homotopy computations.

The computation of the higher homotopy groups is generally considered
very hard. The running time for the algorithms
mentioned above has apparently never been analyzed.
It is clear, however, that Brown's algorithm, which for a long time
 had been the only explicitly published
algorithm for computing $\pi_\thedim(Y)$, is heavily superexponential and
totally unsuitable for actual computations.

Moreover, Anick \cite{Anick-homotopyhard} proved that
computing $\pi_\thedim(Y)$ is \#P-hard,\footnote{Somewhat
informally, the class of  \#P-hard problems consists of computational
problems that should return a natural number
(as opposed to YES/NO problems) and are at least as hard as
counting the number of all Hamiltonian cycles in a given graph,
or counting the number of subsets with zero sum for a given set of integers,
etc. These problems are clearly at least as hard as NP-complete
problems, and most likely even less tractable.}
where $Y$ can even be assumed to be a $4$-dimensional space,
but, crucially, $\thedim$ is regarded as a part of the input.
Actually, the hardness already applies to the potentially easier
problem of computing the \emph{rational homotopy groups}
$\pi_\thedim(Y)\otimes \Q$; practically speaking, one asks
only for the rank of $\pi_\thedim(Y)$,
i.e., the number of direct summands isomorphic to $\Z$.

Anick's \#P-hardness result has a caveat: it assumes $Y$
to be given as a cell complex with a certain very compact representation.
However, recently it was shown by the present authors \cite{ext-hard}
that the computation of $\pi_\thedim(Y)$ remains \#P-hard
even for a $4$-dimensional simplicial complex $Y$, still
with $\thedim$  a part of the input.  

Recently the computation of $\pi_\thedim(Y)$, with $\thedim$ as the parameter,
has been shown W[1]-hard \cite{Mat-homotopyW1}. This means that this computational problem is very unlikely to admit an algorithm with time complexity
bounded by $f(k)n^{C}$, where $n$ is the input size, $C$ is a constant independent of $k$, and $f$ is an arbitrary function. 

Since, as was mentioned above, higher homotopy groups have the reputation
of being very difficult to compute, we consider the following result
surprising.

\begin{theorem}\label{t:pi_k}
For every fixed $\thedim\ge 2$, there is a polynomial-time algorithm
that, given a $1$-connected space $Y$
 represented as a finite simplicial complex,
or more generally, as a \emph{simplicial set with polynomial-time homology}
(this notion will be
 defined in Section~\ref{s:ph-intro}), computes
 (the isomorphism type of) the $\thedim$th homotopy group~$\pi_\thedim(Y)$.
\end{theorem}

Here and in the sequel, the size of a simplicial complex
is the number of simplices.

Since, under the conditions of the theorem,
$\pi_\thedim(Y)$ is a finitely generated Abelian group,
it can be represented as a direct sum of finitely many cyclic groups,
and the algorithm returns such a representation.

Let us remark that the algorithm does not need any certificate
of the $1$-connectedness of $Y$, but if $Y$ is not $1$-connected,
the result may be wrong.

We should also mention that, although the theorem asserts the existence
of an algorithm for every $\thedim\ge 2$, we will actually present a single
algorithm that accepts $Y$ and $\thedim$ as input and outputs
$\pi_{\thedim}(Y)$, and for every $\thedim$
the running time is bounded by a polynomial in the size of $Y$,
where the polynomial generally depends on $\thedim$.
However, for this setting, a single algorithm accepting $Y$ and $\thedim$,
some of the formulations in the sequel would become more cumbersome,
and so in the interest of simpler presentation, we stick to the
setting as in Theorem~\ref{t:pi_k}. A similar remark applies
to all of the other results below.

\heading{Remark: simple spaces. } It can be checked that
Theorem~\ref{t:pi_k}, as well as Theorem~\ref{t:ipost} below, hold,
without any significant change in the proofs,
under the assumtion that $Y$ is a 
\emph{simple} space
(instead of $1$-connected). This, somewhat technical, notion means
that the fundamental group $\pi_1(Y)$ is possibly nontrivial
but Abelian, and  its action on each $\pi_{\thedim}(Y)$, $\thedim\ge 2$,
is trivial. Here the action basically means ``pulling the basepoint
in $Y$ along a loop''---see \cite[pp.~341--342]{Hatcher} for discussion.
A natural example of simple spaces are \emph{H-spaces},
which are a generalization of topological groups.
In the interest of easier presentation we stick to the
1-connectedness assumption, though.

\heading{Computing Postnikov systems. }
The algorithm for computing $\pi_\thedim(Y)$ in
Theorem~\ref{t:pi_k} is a by-product of a polynomial-time
algorithm for computing the first $\thedim$ stages
of a (standard) \emph{Postnikov system} for a given space~$Y$.
In this respect it is similar to the algorithm
of Brown \cite{Brown} and some
others, while, e.g., the algorithm in Real~\cite{Real96} is,
in a sense, dual, building a \emph{Whitehead tower} of~$Y$.
We note that with the tools used in the present paper,
the Whitehead tower algorithm, too, could serve to prove
Theorem~\ref{t:pi_k}.

A Postnikov system of a space $Y$ is, roughly speaking, a
way of building $Y$ from
``canonical pieces'', called \emph{Eilenberg--MacLane
spaces}, whose homotopy structure is the simplest possible.
A Postnikov system has countably many \emph{stages} $P_0,P_1,\ldots$,
where $P_\thedim$ reflects the homotopy properties of $Y$ up to
dimension $\thedim$, and in particular, $\pi_i(P_\thedim)\cong\pi_i(Y)$
for all $i\le\thedim$, while $\pi_i(P_\thedim)=0$ for $i>\thedim$. The
 isomorphisms of the homotopy groups for $i\le\thedim$
are induced by maps $\varphi_i\:Y\to P_\thedim$, which are also a part
of the Postnikov system.
Moreover, there is a mapping $\kkk_{i}$
defined on $P_i$, called the $i$th \emph{Postnikov class}; together with
the group $\pi_{i+1}(Y)$ it describes how $P_{i+1}$ is obtained
from $P_i$, and it is of fundamental importance for dealing with
maps from a space $X$ into~$Y$.
We will say more about Postnikov systems
later on; now we state the result somewhat informally.

\begin{theorem}[informal]\label{t:ipost}
For every fixed $\thedim\ge 2$,
given a $1$-connected space $Y$
represented as a finite simplicial complex,
or more generally, as a simplicial set with polynomial-time homology,
a suitable representation of
the first $\thedim$ stages of a Postnikov system for $Y$
can be constructed, in such a way that
each of the mappings $\varphi_i$ and $\kkk_i$, $i\le \thedim$,
can be evaluated in polynomial time.
\end{theorem}

A precise statement will be given as Theorem~\ref{t:restat14}.

\heading{Computing the structure of all maps. }
In the earlier paper \cite{CKMSVW11} we provided an algorithm
that, given finite simplicial complexes $X$ and $Y$,
computes the structure of $[X,Y]$ under the assumption that,
for some $\thedim\ge 2$, we have $\dim X\le 2\thedim-2$ and
$Y$ is $(\thedim-1)$-connected.\footnote{This means that
$\pi_i(Y)=0$ for all $i=0,1,\ldots,\thedim-1$; a basic example
is $Y=S^\thedim$.} More precisely, under these assumptions,
$[X,Y]$ has a canonical structure of a finitely generated Abelian group,
and the algorithm determines its isomorphism type.

In the algorithm, the stage $P_{2\thedim-2}$ of the Postnikov system
of $Y$ is used as an approximation to $Y$, since for every
$1$-connected $Y$ and every $X$ of dimension at most $2\thedim-2$,
there is an isomorphism $[X,Y]\cong[X,P_{2\thedim-2}]$, induced
by the composition with the mapping
$\varphi_{2\thedim-2}\:Y\to P_{2\thedim-2}$. At the same time,
the continuous maps $X\to P_{2\thedim-2}$ are easier to handle
than the maps $X\to Y$: each of them is homotopic to a simplicial,
and thus combinatorially described, map, and
it is possible to define (and implement) a binary operation on
$P_{2\thedim-2}$ which induces the group structure on $[X,P_{2\thedim-2}]$.
This, in a nutshell, explains the usefulness of the Postnikov
system for dealing with maps into~$Y$.

It is easy to check that the algorithm in  \cite{CKMSVW11}
works in polynomial time
in the size (number of simplices) of $X$ and $Y$ for every \emph{fixed}
$\thedim$, provided that the first $2\thedim-2$ stages of a Postnikov system
for $Y$ can be computed in polynomial time, as in
Theorem~\ref{t:ipost} (the precise requirements on what
should be computed can be found in  \cite{CKMSVW11}).
We thus obtain the following result, anticipated in  \cite{CKMSVW11}.

\begin{corol}[based on \cite{CKMSVW11}]\label{c:[XY]}
For every fixed $\thedim\ge 2$, there is a polynomial-time algorithm that,
given finite simplicial complexes $X$, $Y$,
where $\dim(X)\le 2k-2$ and $Y$ is $(k-1)$-connected,
computes the isomorphism type of $[X,Y]$ as an Abelian group.
More generally, $X$ can be a finite simplicial set
 and $Y$ a simplicial set with polynomial-time homology.
\end{corol}

We will not dwell on the proof here, since it follows immediately by
plugging the Postnikov system algorithm of Theorem~\ref{t:ipost}
into the algorithm of \cite{CKMSVW11} as a subroutine.
We only remark that while the result of \cite{CKMSVW11} is
formulated for $Y$ a finite simplicial complex,
 $Y$ actually enters the computation solely through its Postnikov
system, and so any $Y$ can be handled for which the appropriate
stages of the Postnikov system are available.

\heading{Computing extensions of maps. }
Related to the problem of determining $[X,Y]$ is the
\emph{extension problem}: given spaces $A$, $X$, $Y$, where $A \subseteq X$,
and a map $f\:A\to Y$, can $f$ be extended to a map $X\to Y$?
This is one of the most basic questions in algebraic topology, and
a number of topological concepts, which may look quite advanced
and esoteric to a newcomer, such as
\emph{Steenrod squares}, have a natural motivation in an attempt
at a stepwise solution of the extension problem;
see, e.g., Steenrod  \cite{Steenrod:CohomologyOperationsObstructionsExtendingContinuousFunctions-1972}.

For $A\subseteq X$ and $f\:A\to Y$ as above, let
$[X,Y]_{f}\subseteq [X,Y]$
denote the set of all homotopy classes of maps $X\to Y$ that
contain a map extending~$f$.

One may also  want to study the set of all
extensions $\bar f$ of $f$ with a finer equivalence relation
than the ordinary homotopy of maps $X\to Y$, namely, homotopy
fixing the map on $A$ (i.e., $\bar f_1,\bar f_2\:X\to Y$
are equivalent if they are connected by a homotopy $F\:X\times[0,1]\to Y$
with $F(x,t)=f(x)$ for all $x\in A$ and $t\in[0,1]$).
In order to distinguish these two notions, we refer to
determining the structure of all extensions modulo homotopy
fixing $f$ on $A$ as the \emph{fine classification of the extensions of~$f$},
and to determining $[X,Y]_f$ as the \emph{coarse classification of
the extensions of~$f$}.

As a simple consequence of Theorem~\ref{t:ipost}
and the methods of \cite{CKMSVW11}, we obtain the following.

\begin{theorem}[Extendability of maps]\label{t:ext}
Let $\thedim\ge 2$ be fixed. Then there is a polynomial-time algorithm that,
given finite simplicial complexes $X$, $Y$, a subcomplex
$A\subseteq X$, and a simplicial map $f\:A\to Y$,
where $\dim(X)\le 2k-1$ and $Y$ is $(k-1)$-connected,
decides whether $f$ admits
an extension to a (not neccessarily simplicial) map $X\to Y$.

Moreover, if the extension exists and $\dim X\le 2k-2$,
the algorithm computes the structure of $[X,Y]_{f}$
as a coset in the Abelian group $[X,Y]$.

More generally, $X$ can be a finite simplicial set
and $Y$ a simplicial set with polynomial-time homology.
\end{theorem}

The proof, assuming some of the material from  \cite{CKMSVW11},
is presented in Section~\ref{s:ext} below. We stress that,
while $f$ is given as a simplicial map (so that it can be
specified by finite means), the extensions are considered
as \emph{arbitrary continuous maps}, and in particular, they
are not assumed to be simplicial maps $X\to Y$.

Theorem~\ref{t:ext} provides a coarse classification of
all extension assuming $\dim X\le2\thedim-2$.
There is also an  algorithm that, under the same
conditions, provides a fine classification
of all extensions. It appears as a part of a more general
result in  
 \v{C}adek, Kr\v{c}\'al, and Vok\v{r}\'inek~\cite{aslep}.

For the next higher dimension $\dim X=2\thedim-1$, although
the existence of an extension can be decided, we can
no longer produce the coarse classification of all
extensions, and we suspect that this
problem should be intractable in a suitable sense.

\heading{Hardness results.} The assumption on $X$ and $Y$ in Corollary~\ref{c:[XY]} may perhaps
look artificial at first sight. However, it is needed for $[X,Y]$ to have
a canonical structure of an Abelian group.
Moreover, the similar assumption in Theorem~\ref{t:ext}
(with $\dim X$ one higher) turns out to be sharp,
in the following sense: In \cite{ext-hard} we show that the extendability
problem is \emph{algorithmically undecidable} for finite
simplicial complexes $A\subseteq X$ and $Y$ and a simplicial map
$f\:A\to Y$ with $\dim X= 2\thedim$ and $(\thedim{-}1)$-connected $Y$.
Moreover, for every $\thedim\ge 2$, there
is a \emph{fixed} $(\thedim-1)$-connected $Y=Y_\thedim$ such that the extension
problem for maps into $Y_\thedim$, with $A,X,f$ as the input,
$\dim X\le 2\thedim$,
is undecidable. In particular, for every \emph{even} $\thedim \ge 2$,
the extension problem is undecidable for $X$ of dimension
$2\thedim$ and $Y=S^\thedim$, the sphere.  (Interestingly,
for \emph{odd} $\thedim$, it was recently shown \cite{vokrinek:odd-spheres}
that the extension problem is decidable for $Y=S^\thedim$,
without any restriction on the dimension of~$X$.)

 In a similar sense, $X=X_\thedim$ and $A=A_\thedim$
can be fixed, so that the input consists only of $Y$ and $f$,
and undecidability still holds. See  \cite{ext-hard}
for more details. The undecidability is obtained by reduction from quadratic Diophantine equations. A very similar argument shows that deciding the existence of a nontrivial map $X\to Y$ is as hard as deciding the existence of a nontrivial solution of quadratic \emph{homogeneous} Diophantine equations \cite{Krcal-thesis}.  

We have already mentioned known hardness results
for computing the homotopy group $\pi_\thedim(Y)$: the \#P-hardness
if $k$ is a part of input and W[1]-hardness if $k$ is regarded
as a parameter. The latter shows that, modulo a widely believed complexity
assumption, for every polynomial-time algorithm that computes $\pi_k(Y)$,
the  degree of the polynomial in the running time bound has to grow with~$k$
(and of course, the same applies to algorithms for computing the
Postnikov stages of $Y$). Still, it may be interesting to analyze
the running time in more detail. 

On the other hand, this kind of finer theoretical analysis
may not be very relevant for the practical performance of the algorithm
on manageable instances. For example, one of the main
ingredients of our polynomial-time algorithm, is an algorithm
of \cite{pKZ1} dealing with the Eilenberg--MacLane space $K(\Z,1)$
(discussed later). That algorithm is not quite simple and
its analysis is demanding; however, as for practical performance,
it seems to be inferior to a simple, classical, but worst-case
exponential algorithm due to Eilenberg and Mac~Lane, at least in simple
tests (as we were informed by Francis Sergeraert).


\heading{Methods. } The results of this paper rely on a number of known
methods and techniques.  We see the main contributions
in selecting suitable methods  among various available alternatives and
adapting them for our purposes, assembling everything together,
and setting up a framework for dealing with
polynomial-time algorithms of a somewhat unusual kind.

This framework, with somewhat modified terminology, has been used
in several subsequent papers 
\cite{aslep,Vokrinek,VokriFil-homotopic,Vokrinek-algoheaps}, 
which provide polynomial-time
algorithms for a number of other homotopy-theoretic problems.

We have also made a significant effort to present the results
in an accessible manner. The required techniques
involve a large amount of material, and methods from two traditionally
separated areas, algebraic topology and algorithm design,
need to be brought together. We expect
the number of potential readers moving with ease in both of these
areas to be rather small, and thus we try to make the exposition
as self-contained as reasonably possible, sometimes
covering things which may be considered well known
in one of  the areas.

The Postnikov system algorithm, on the top level, essentially
follows the approach of Brown \cite{Brown} (we have examined proofs of existence of
a Postnikov system in standard
textbooks, such as \cite{Hatcher,Spanier}, and none seemed quite
suitable for our purposes). But Brown's algorithm in the original
form uses a straightforward representation of
simplicial Eilenberg--MacLane spaces, and thus it
works only for input spaces with all the relevant homotopy groups finite.
In the case of infinite homotopy groups, the corresponding
Eilenberg--MacLane spaces
are simplicial sets with \emph{infinitely many} nondegenerate
simplices in the relevant dimensions.
For dealing with these, and with other infinitary objects
derived from them in the course of the algorithm,
we follow the paradigm of \emph{objects with effective homology}
developed  by Sergeraert, Rubio, Real, Dousson, and Romero
(see, e.g.,
\cite{Sergeraert:ComputabilityProblemAlgebraicTopology-1994,RubioSergeraert:ConstructiveAlgebraicTopology-2002,RomeroRubioSergeraert,SergRub-homtypes};
the lecture notes \cite{SergerGenova} provide the most detailed
exposition available so far).
Some of their results have never appeared
in peer-reviewed journals; for example, for
some of the operations needed in the present paper,
we use methods described in some detail,
as far as we know, only in Real's PhD. thesis \cite{RealThesis}
written in Spanish.

For the purpose of polynomial-time computations, we replace effective
homology with \emph{po\-ly\-nomi\-al-ti\-me homology}, as introduced in
\cite{pKZ1}. Thus, we need polynomial-time versions of
all the required operations in effective homology.

There is one case, namely, polynomial-time
homology for the simplicial Eilenberg--MacLane space $K(\Z,1)$,
 where we had to develop  a new algorithm, since the classical one
is not polynomial in general. This part is not provided here,
but rather in the companion paper \cite{pKZ1}; the methods used in that
paper have flavor somewhat distinct from those employed here, and
we feel that a combined paper would be too extensive and cumbersome.

In all other cases, we could rely on known algorithms. Verifying
their polynomiality sometimes still requires nontrivial analysis
and assumptions. Moreover, since the intermediate
objects used in the algorithms are of somewhat unusual kind
from the computer science point of view, we need to set up
a suitable formal framework in order to make claims about
polynomial running time.


%

\heading{Applications.} We consider the fundamental nature of the algorithmic problem considered here a sufficient
motivation of our research (e.g., because $[X,Y]$ is indeed one of the most basic objects of study in algebraic topology).
However, we also believe that work in this area will bring
various  connections and applications, also in other fields,
possibly including practically usable software, e.g., for aiding
research in topology.

A nice concrete application comes from the  so-called ROB-SAT problem---robust satisfiability of systems of equations. The problem is given by a rational value $\alpha>0$ and a PL function $f\:K\to \R^k$ defined by rational values on the vertices of a simplicial complex $K$.  The question is whether an arbitrary continuous $g\:K\to\R^k$ that is at most $\alpha$-far from $f$ (i.e., \(\|f-g\|_\infty\leq \alpha\)) has a root. Franek and  Kr\v{c}\'{a}l \cite{FranekKrcal:RobustSatisfiability} exhibit a computational equivalence of ROB-SAT and the extension problem for maps into the sphere \(S^{k-1}\). Our Theorem~\ref{t:ext} then yields an algorithmic solution when $\dim K\leq 2k-3$.

One important motivation for starting this project was the computation
of the  \emph{$\Z_2$-index}  (or \emph{genus}) $\ind(X)$ of a
$\Z_2$-space $X$,\footnote{A \emph{$\Z_2$-space} is a topological space $X$ with an action of the group $\Z_2$; the action is described by a homeomorphism $\nu\colon X\to X$ with $\nu\circ\nu=\id_X$. A primary example is a sphere $S^d$ with the antipodal action $x\mapsto-x$. An \emph{equivariant  map} between $\Z_2$-spaces is a continuous map that commutes with the $\Z_2$ actions.} i.e., the smallest $d$ such that $X$ can be equivariantly mapped into~$S^d$. For example, the classical Borsuk--Ulam theorem asserts that there is \emph{no} equivariant map $S^{d}\to S^{d-1}$, i.e., that $\ind(S^d)= d$.

Generalizing the results in the present paper, it is shown in~\cite{aslep} that there is an algorithm that decides whether $\ind(X)\le d$, provided that $d\geq 2$ and $\dim(X)\le 2d-1$; for fixed $d$ the running time is polynomial in the size of $X$.

The problem of computing $\ind(X)$ arises, among others, in the problem of \emph{embeddability}
of topological spaces, which is a classical and much studied area
(see, e.g., the survey by Skopenkov \cite{skopenkov-survey}).
One of the basic questions here is, given a $k$-dimensional
finite simplicial complex $K$, can it be (topologically) embedded
in $\R^d$? The celebrated \emph{Haefliger--Weber theorem}
from the 1960s asserts that, in the \emph{metastable range
of dimensions}, i.e., for $k\le \frac23d-1$, embeddability
is equivalent to $\ind(K_\Delta^2)\le d-1$,
where $K_\Delta^2$ is a certain $\Z_2$-space constructed from $K$
(the \emph{deleted product}). Thus, in this range, the embedding problem is, computationally,
a special case of $\Z_2$-index computation; see \cite{MatousekTancerWagner:HardnessEmbeddings-2011} for a study of algorithmic aspects of the embedding problem, where the metastable range was left as one of the main open
problems.

The $\Z_2$-index also appears as
a fundamental quantity in combinatorial applications
of topology. For example, the celebrated result of Lov\'asz
on Kneser's conjecture can nowadays be re-stated as
$\chi(G)\ge \ind(B(G))+2$, where $\chi(G)$ is the chromatic
number of a graph $G$, and $B(G)$ is a certain simplicial complex
constructed from $G$ (see, e.g., \cite{Matousek:BorsukUlam-2003}).
We find it striking that prior to \cite{aslep}, \emph{nothing} seems to 
have been known about the computability of such an interesting quantity
as $\ind(B(G))$.
Indeed, some authors (e.g., Kozlov \cite{kozl-surv})
worked with a weaker, cohomologically defined index, in part because 
of suspicions that the $\Z_2$-index might be intractable.

\heading{Implementation.} As indicated  above, another appealing research direction is the development  of a practical software for the problems  considered in  this  paper.  A particular solution, the  Kenzo  program written in Common Lisp  by Francis Sergeraert and Xavier  Dousson, maintained and extended with other collaborators, is freely available at
\url{http://www-fourier.ujf-grenoble.fr/~sergerar/Kenzo/}.

The  program implements  the  concepts of  effective  homology, and
currently it  enables  the construction of the Postnikov  stages as long as the homotopy  groups involved are isomorphic to  direct summands of copies of  \(\Z \) and  \(\Z_2\). For instance, for spheres \(S^{d}, d\ge 2,\) the Postnikov stages $P_0,P_1,\ldots, P_{d+2}$ can be constructed (as well as homotopy groups \(\pi_0(S^d),\ldots,\pi_{d+2}(S^d)\) can be  computed).   The program Kenzo cannot compete with the current state-of-the-art computations of homotopy  groups  of spheres,  where  many  special  properties of  the spheres are employed.  Rather in an orthogonal fashion,  it provides a general solution  for essentially arbitrary spaces  in low dimensions.

A different piece of software is a package called HAP written by Graham Ellis extending the computer algebraic system GAP; see \cite{ellis2008HAP}. Among others, it provides homological computations related to Eilenberg--MacLane spaces.


\section{Simplicial sets and chain complexes
with polynomial-time homology}\label{s:ph-intro}

\subsection{Preliminaries on simplicial sets and chain complexes}
\label{s:ssetchain}

\heading{Simplicial sets. } A simplicial set is a way of specifying
a topological space in purely combinatorial terms;
 we can think of it as an instruction manual telling us how
the considered space should be assembled from simple building
blocks.
All topological spaces in the considered algorithms
are going to be represented in this way.
Simplicial sets can be regarded as a generalization of
simplicial complexes; they are formally more complicated,
but more powerful and flexible. We refer to \cite{Friedm08,SergerTrieste}
for an introduction,
to \cite{Curtis:SimplicialHomotopyTheory-1971,May:SimplicialObjects-1992}
as compact more comprehensive sources, and to \cite{GoerssJardine}
for a more modern treatment.

Similar to a simplicial complex,
a simplicial set is a space built of vertices, edges, triangles,
and higher-dimensional simplices, but simplices are allowed to be glued
to each other and to themselves in more general ways. For example,
one may have several 1-dimensional simplices connecting the same
pair of vertices, a 1-simplex forming a loop,
two edges of a  2-simplex identified
to create a cone, or the boundary of a 2-simplex all contracted
to a single vertex, forming an $S^2$.
\immfig{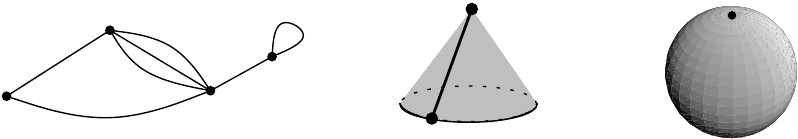}

Another new feature of a simplicial set, in comparison with a simplicial
complex, is the presence of \emph{degenerate simplices}. For example,
the edges of the triangle with a contracted boundary (in the last
example above) do not disappear, but each of them becomes
a degenerate 1-simplex.

A simplicial set $X$ is represented as a sequence
$(X_0,X_1,X_2,\ldots)$ of mutually disjoint sets, where the
elements of $X_\thedim$ are called the \emph{$\thedim
$-simplices of $X$} (we note that, unlike for simplicial
complexes, a simplex in a simplicial set need not be
determined by the set of its vertices; indeed, there can be
many simplices with the same vertex set). For every $\thedim
\ge 1$, there are $\thedim +1$ mappings
$\partial_0,\ldots,\partial_\thedim\:X_\thedim\to
X_{\thedim-1}$ called \indef{face operators}; the intuitive
meaning is that for a simplex $\sigma\in X_\thedim$,
$\partial_i\sigma$ is the face of $\sigma$ opposite to the
$i$th vertex. Moreover, there are $\thedim +1$ mappings
$s_0,\ldots,s_\thedim\:X_\thedim\to X_{\thedim+1}$ (opposite
direction) called the \emph{degeneracy operators}; the
approximate meaning of $s_i\sigma$ is the degenerate simplex
which is geometrically identical to $\sigma$, but with the
$i$th vertex duplicated.

A simplex is called
\indef{degenerate} if it lies in the image of some $s_i$;
otherwise, it is \indef{nondegenerate}. We write $X^\ndg$ for
the set of all nondegenerate simplices of~$X$.
We call $X$ \emph{finite} if $X^\ndg$ is finite (every nonempty simplicial
set has infinitely many degenerate simplices).

There are natural axioms that the $\partial_i$
and the $s_i$ have to satisfy, but we will not list them here,
since we will not really use them. Moreover, the usual definition
of simplicial sets uses the language of category theory
and is very elegant and concise; see, e.g.,
\cite[Section I.1]{GoerssJardine}.

Every simplicial set $X$ specifies a topological space $|X|$,
the \emph{geometric realization} of $X$. It is obtained by assigning
a geometric $\thedim $-dimensional simplex to each nondegenerate
$\thedim $-simplex of $X$,
and then gluing these simplices together according to the
face operators; we refer to the literature for the precise definition.

\heading{Simplicial maps.} For simplicial sets $X,Y$,
a \emph{simplicial map} $f\:X\to Y$ is a sequence
$(f_\thedim)_{\thedim=0}^\infty$ of maps
$f_\thedim\:X_\thedim\to Y_\thedim$ (every $\thedim$-simplex is
mapped to a $\thedim$-simplex) that commute with the face and
degeneracy operators, i.e., $\partial_i f_\thedim =f_{\thedim-1}\partial_i$
and $s_i f_\thedim=f_{\thedim+1}s_i$.  We let $\SM(X,Y)$ stand
for the set of all simplicial maps $X\to Y$.

It is useful to observe that it suffices to specify
a simplicial map $f\: X\to Y$ on the \emph{nondegenerate} simplices
of $X$; the values on the degenerate simplices are then determined
uniquely. In particular, if $X$ is finite,
then such an $f$ can be specified as a finite object.

Every simplicial map $f\:X\to Y$ defines a continuous map
$\varphi\:|X|\to|Y|$ of the geometric realizations. There is a very
important class of simplicial sets, called \emph{Kan simplicial sets},
with the following crucial property: if $Y$ is a Kan simplicial
set and $X$ is any simplicial set, then every continuous
map $\varphi:|X|\to|Y|$ is homotopic to (the geometric realization of)
some simplicial map $f\:X\to Y$.
This is essential in the algorithmic treatment of continuous maps.
Here we omit the definition of a Kan simplicial set, since we will
not directly use it.

\heading{Chain complexes. } Together with a simplicial set $X$,
we will consider its associated \emph{normalized chain
complex} $C_*(X)$, but sometimes the algorithms will
also need other types of chain complexes.

For our purposes, it is sufficient to use the kind of chain complexes
usually considered in introductory textbooks when defining homology
and cohomology groups. Thus, in the sequel, a \emph{chain complex} $C_*$
is a sequence
$(C_\thedim)_{\thedim\in\Z}$ of \emph{free} Abelian groups
(in other words, free $\Z$-modules), together with a
sequence $(\diff_\thedim\:C_\thedim\to
C_{\thedim-1})_{\thedim\in\Z}^\infty$ of group
homomorphisms that satisfy the condition
$\diff_{\thedim-1} \diff_\thedim=0$.\footnote{More generally,
instead of $\Z$-modules, one might consider modules over a commutative ring $R$,
and they need not be free. Moreover, in the literature, the operations
considered in Section~\ref{s:ops} below are sometimes presented
in a still more general algebraic setting, with \emph{differential modules}
replacing chain complexes. Here we prefer the more concrete setting of chain
complexes, mainly in order to avoid burdening the presentation
with additional notions.}
The $C_\thedim$ are the \emph{chain groups},
their elements are called \emph{$\thedim$-chains}, and the
$\diff_\thedim$ the \emph{differentials}. If $c$ is a $\thedim$-chain,
we sometimes say that the \emph{degree} of $c$ equals~$\thedim$.
We will work only with chain complexes $C_*$ with $C_\thedim=0$
for all $\thedim<0$.

We also recall that $Z_\thedim=Z_\thedim(C_*):=\ker \diff_\thedim\subseteq
C_\thedim$ is the group of \emph{cycles},
$B_\thedim=B_\thedim(C_*):=\im\diff_{\thedim+1}\subseteq Z_\thedim$ is
the group of \emph{boundaries}, and
the quotient group $H_\thedim(C_*):=Z_\thedim/B_\thedim$
is the $\thedim $th \emph{homology group}
of the chain complex $C_*$.

For the normalized chain complex $C_*(X)$ of a simplicial set $X$
mentioned above, the $\thedim$th chain group
$C_\thedim(X)$ is the free Abelian group over
$X_\thedim^\ndg$, the set of all $\thedim$-dimensional
\emph{nondegenerate} simplices (in particular, $C_\thedim(X)=0$
for $\thedim <0$).\footnote{In the literature, the notation
$C_*(X)$ is sometimes used for another
chain complex
associated with $X$, where the degenerate simplices also
appear as generators (it yields the same homology
as the normalized chain complex). But since we will work exclusively with
the normalized chain complex, we reserve the simple notation
$C_*(X)$ for these.}
This means that a $\thedim $-chain is a formal sum
$$
c=\sum_{\sigma\in X_\thedim^\ndg} \alpha_\sigma \cdot\sigma,
$$
where the $\alpha_\sigma$ are integers, only finitely many of them
nonzero.
The differentials are defined in a standard way
using the face operators: for $\thedim $-chains of the form $1\cdot\sigma$,
which constitute a basis of $C_\thedim(X)$, we set
$\diff_\thedim (1\cdot\sigma):=\sum_{i=0}^\thedim(-1)^i\cdot\partial_i\sigma$
(some of the $\partial_i\sigma$ may be degenerate simplices; then
they are ignored in the sum),
and this extends to a homomorphism in a unique way (``linearly'').

Let $C_*$ and $\tC_*$ be two chain complexes.
We recall that a \emph{chain map} $f\:C_*\to\tC_*$
is a sequence $(f_\thedim)_{\thedim\in\Z}$ of homomorphisms,
$f_\thedim\:C_\thedim\to\tC_\thedim$, compatible with the differentials,
i.e., $f_{\thedim-1} \diff_\thedim=\tilde\diff_{\thedim} f_\thedim$.
A simplicial map $f\:X\to Y$ of simplicial sets induces a chain map
$f_*\:C_*(X)\to C_*(Y)$ in the obvious way.

\heading{Mapping cylinder and mapping cone. }
We recall two standard constructions for topological spaces,
and then we mention their algebraic counterparts.
Let $f\:X\to Y$
be a map of topological spaces. Then the \emph{mapping cylinder}
$\MCyl(f)$ is obtained by gluing the product (``cylinder'')
$X\times [0,1]$ to  $Y$ via the identification of $(x,0)$
with $f(x)\in Y$, for all $x\in X$, as the next picture indicates.

\immfig{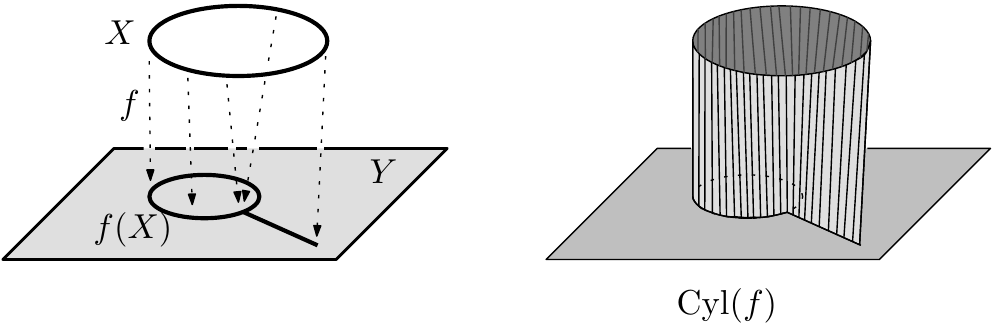}

The mapping cone $\MCone(f)$ is obtained from the mapping cylinder
$\MCyl(f)$ by contracting the ``top copy''
of $X$, i.e., the subspace $X\times\{1\}$, to a single point.

We will not use these geometric constructions directly.
In one of the proofs, we will need a simplicial version of the mapping
cylinder,  in a setting where $X,Y$ are simplicial sets and
$f$ is a simplicial map, and we will introduce it at the appropriate
moment. Otherwise, we will work exclusively with \emph{algebraic
analogs} of these constructions. Conceptually, they are obtained
by considering how the chain complexes of $\MCyl(f)$ and $\MCone(f)$
are related to the chain complexes of $X$ and $Y$ and to the chain
map $f_*$ induced by $f$, and then generalizing to arbitrary chain complexes
and chain map.

The resulting definitions are as follows. Let
$C_*,\tC_*$ be chain complexes and let $\varphi\:C_*\to\tC_*$ be
a chain map. Then the (algebraic) \emph{mapping cylinder}
$\MCyl_*(\varphi)$ has chain groups $\MCyl_\thedim:=C_{\thedim-1}\oplus
\tC_\thedim\oplus C_\thedim$ (a direct sum), and the differential given by
\[
\diff^{\MCyl_*(\varphi)}_\thedim(a,\tilde c,b):=
(-\diff_{\thedim-1}(a), \varphi_\thedim(a)+\tilde\diff_\thedim(\tilde c),
-a+\diff_\thedim(b)),
\]
where $\diff$ is the differential of $C_*$
and $\tilde\diff$ is the differential of $\tC_*$.

In a similar spirit,
the (algebraic) \emph{mapping cone} $\MCone_*(\varphi)$ of $\varphi$
is the chain complex whose $\thedim$th
chain group is the direct sum $C_{\thedim-1}\oplus\tC_\thedim$,
and with the differential given by
\begin{equation}\label{e:conediff}
\diff^{\MCone_*(\varphi)}_\thedim(a,\tilde b):=
 (-\diff_{\thedim-1}(a),\varphi_\thedim(a)+\tilde\diff_\thedim(\tilde b)),
\ \ \ (a,\tilde b)\in C_{\thedim-1}\oplus\tC_\thedim.
\end{equation}

For later use, we note that the canonical inclusion $i\:\tC_*\to
\MCone_*(\varphi)$, given by $i(\tilde b)=(0,\tilde b)$, is a chain map,
as can be seen from (\ref{e:conediff}); on the other hand,
the other canonical inclusion $j\:C_*\to \MCone_*(\varphi)$
is not a chain map (it does not respect degrees, and it does not
commute with the face operators, unless $\varphi=0$).

\subsection{The meaning of ``computing $\pi_{17}(Y)$ in polynomial time''}
\label{s:encsize}

In computational complexity theory, which is a branch of computer
science that focuses on classifying computational problems according
to their inherent difficulty, algorithms are usually represented
as Turing machines, or some other models of a general computing
machine. Such an algorithm accepts an input $u\in \Sigma^*$,
where $\Sigma$ is some fixed finite alphabet (for our purposes,
we may assume w.l.o.g.~that $\Sigma=\{0,1\}$ is the binary alphabet),
and where $\Sigma^*$ denotes the set of all strings (finite sequences)
of symbols of $\Sigma$. Given $u\in \Sigma^*$, the algorithm
computes some output $v\in\Sigma^*$.

We say that a mapping $f\:\Sigma^*\to\Sigma^*$ is a \emph{polynomial-time
mapping} if there is an algorithm $A$ and a polynomial $p(x)$
such that, for every input $u\in\Sigma^*$, the algorithm $A$
outputs $f(u)$ after at most $p(|u|)$ steps, where $|u|$
denotes the length (number of symbols) of~$u$.

It is easy to see that
the composition of two polynomial-time mappings is again
a polynomial-time mapping. (Here we use that
if the computation of $f(u)$ takes at most $p(|u|)$ steps,
then $|f(u)|\le p(|u|)$, for otherwise, the algorithm
for evaluating $f$ would not have enough time to write
$f(u)$ down.)
We will frequently use this fact, often without mentioning
it explicitly.

\heading{Encoding size. }
Thus, the notion of polynomial time is very straightforward,
although not easy to study, for mappings assigning strings to strings.
However, if we consider ``real-life'' computational tasks,
such as testing whether a given natural number $n$ is a prime,
or computing $\pi_{17}(Y)$ for a simplicial complex $Y$,
then neither the input nor the output are \emph{a~priori} strings.
In order to talk about the computational complexity of such
tasks, we first need to specify some \emph{encoding}
of the input and output objects by strings.

For testing primality, we thus need to specify an (injective) function
$\enc\:\N\to\Sigma^*$ assigning a string to every natural number
(here we consider the encoding of the two possible outputs, YES and NO,
as too trivial to discuss). The most usual choice is representing
$n$ by the standard binary notation; e.g., $\enc(17)=10001$.
In this paper we assume binary encoding of all integers (unless
stated otherwise).
With this encoding, the possibility of primality testing
in polynomial time is a celebrated recent result. However, if
we chose a different, \emph{unary} encoding $\enc'$,
which represents $n$ by a string of $n$ ones, e.g.,
$\enc'(17)=1111111111111111$, then testing primality in polynomial
time becomes very easy---we can afford to test all possible divisors
from $2$ to $n-1$. This example illustrates that sometimes
the choice of encoding may be very significant.\footnote{Here is another
example, closer to our topics, of how the encoding may matter:
If a simplicial complex $K$ is given by a list of \emph{all}
of its simplices, as we are going to assume here, then
computing the Euler characteristic $\chi(K)$ is a trivial matter
and can obviously be done in polynomial time.
However, if $K$ is specified by listing only the \emph{maximal}
simplices, and if we do \emph{not} assume $\dim K$ fixed,
then computing $\chi(K)$ is \#P-hard \cite{RouneSaensdeCabezon},
and thus extremely unlikely to be polynomial-time solvable!}

For discussing polynomial-time algorithms, we often do not need
to specify the encoding function $\enc$ completely. Usually we suffice
with the \emph{encoding size}, where the size of an object $a$
is $\size(a)=|\enc(a)|$, the number of bits in its encoding.
In the above example with primality,
we had $\size(n)=\lfloor 1+\log_2 n\rfloor$ for the binary encoding and
$\size'(n)=n$ for the unary encoding.

We note that changes in the encoding
that transform the size by at most a fixed polynomial,
e.g., replacing $\size(a)$ with $\size'(a)=(37\size(a)+100)^{26}$,
leave the notion of a polynomial-time mapping unchanged.
Thus, for the purpose of developing polynomial-time algorithms,
we usually need not describe the encoding in much detail.

\heading{The encoding size of simplicial complexes and of Abelian
groups. } We recall that a finite simplicial complex $Y$
can be regarded as a hereditary system of subsets of a finite
vertex set $V$ (hereditary meaning that if $\sigma\in Y$
and $\sigma'\subseteq\sigma$, then $\sigma'\in Y$ as well).
For encoding such an $Y$, we can identify $V$ with
the set $\{1,2,\ldots,n\}$, and then represent $Y$ as a list
of all simplices, where each simplex is given by the list
of vertices. Thus, if the dimension of $Y$ is bounded by a
constant (as we may assume in all of our results),
$\size(Y)$ is bounded by a polynomial
function of the number of simplices of $Y$, and so for the
purpose of discussing polynomial-time algorithms, we may
assume that $\size(Y)$ equals the number of simplices.

The elements of the homotopy group $\pi_{17}(Y)$ are, by definition,
equivalence classes of pointed maps $S^{17}\to Y$, and it is
far from obvious how even a single such element could be
represented by a string. However, our algorithm computes only
the \emph{isomorphism type} of the homotopy groups. (Computing
a reasonable representation for the mappings corresponding
to the generators of the homotopy group is currently an
interesting open problem.)

It is known that, for a finite
simplicial complex $Y$ and $\thedim\ge 2$,
$\pi_{\thedim}(Y)$ is a finitely
generated Abelian group; this actually also follows from the analysis of our
algorithm. A well-known structure theorem asserts that
each finitely generated Abelian group $\pi$ can be represented as
a direct sum $\Z^r\oplus
(\Z/m_1)\oplus (\Z/m_2)\oplus\cdots\oplus (\Z/m_s)$
of cyclic groups.\footnote{Moreover, we may assume that
the $m_i$ satisfy the divisibility
condition $m_1|m_2|\cdots|m_s$, in which case these orders are determined uniquely from $\pi$ and thus describe its isomorphism type.} We are going to encode it
by the $(r+s)$-tuple
$\mm=(\underbrace{0,0,\ldots,0}_{r~\mathrm{zeros}},m_1,\ldots,m_s)$,
where $m_1,\ldots,m_s$ are encoded in binary.
Thus, we may take
\[
\size(\pi)=r+\sum_{i=1}^s\size(m_i).
\]

The reader may wonder why $r$ is not encoded in binary as well.
The reason is pragmatic; we will also be using finitely generated
Abelian groups as inputs to certain auxiliary algorithms,
and we would not be able to make these auxiliary algorithms
polynomial with $r$ encoded in binary. A heuristic explanation
for this is that an \emph{element} of $\Z^r$ is an $r$-tuple
of integers, and thus an encoding of such an element has
size at lest proportional to $r$. If the encoding size
of $\Z^r$ were of order $\log_2 r$, then a polynomial-time
algorithm working with $\Z^r$ would not be able even to
read or write any single group element.

This specification of encoding sizes gives a precise meaning
to the polynomiality claim in Theorem~\ref{t:pi_k}.
We note that the polynomiality of the algorithm also implies
the (non-obvious) claim
that, for $\thedim$ fixed, $\size(\pi_\thedim(Y))$ is bounded
by a polynomial function of $\size(Y)$.

\subsection{Locally polynomial-time simplicial sets and
chain complexes}

\heading{In what sense do we construct a Postnikov system?}
As was mentioned after Theorem~\ref{t:ipost},
the stages $P_\thedim=P_\thedim(Y)$ of a Postnikov system of $Y$ can be
regarded as approximations of $Y$,
which are in some sense easier to work with than $Y$ itself.
The price to pay is that even if $Y$ is a finite simplicial
complex, the $P_\thedim$ are simplicial sets that
usually have infinitely many nondegenerate simplices
in each dimension.

In many areas where computer scientists seek efficient
algorithms, the algorithms work with finite objects,
such as finite graphs or matrices, and there is no problem
with explicitly representing such objects in the computer memory.
This contrasts with the situation for the $P_\thedim$,
where we cannot produce the infinite list of all simplices
of a given dimension explicitly. Thus, the question
arises, in what sense we construct $P_\thedim$
and how we can work with it.

A complete answer is that we want to equip $P_\thedim$
with polynomial-time homology, which is a notion defined
later. For now, we give at least a partial answer:
 We certainly want to be able to inspect \emph{locally} every given
piece of $P_\thedim$. For example, for every fixed $\thedim$ and
$\otherdim$, given any $\otherdim$-dimensional
simplex $\sigma$ of $P_\thedim$, and an integer $i\in\{0,1,\ldots\otherdim\}$,
we should be able to compute the $i$th face $\partial_i\sigma$,
the $i$th degeneracy $s_i\sigma$, and
also the value $\kkk_\thedim(\sigma)$
of the Postnikov class at~$\sigma$.
Because of the infinite domains, the mappings $\partial_i$, $s_i$,
and $\kkk_\thedim$ cannot be given by a finite table (somewhat
exceptionally, the mapping $\varphi_\thedim\:Y\to P_\thedim$
could be represented by a table if $Y$ is finite).
Instead, each of them is going to be given as an \emph{algorithm}.

Thus, we are going to represent  stage of the Postnikov system
by a collection of algorithms, and similarly for
various other infinite simplicial sets,
chain complexes, and some other kinds of objects.
In computer science, this is sometimes
called a \emph{black box} or \emph{oracle}
representation.\footnote{Professor Sergeraert has suggested an alternative
framework, inspired by functional programming,
 for dealing with computational complexity of algorithms
similar to those considered in the present paper.
It should be presented in \cite{RS-funcform}. }

\heading{Polynomiality. } Since we want to use the stages of the
Postnikov system in polynomial-time algorithms, such as the one
in Corollary~\ref{c:[XY]} (the computation of $[X,Y]$),
we obviously want that the black boxes representing $P_\thedim=P_\thedim(Y)$
work in polynomial time. But some care is needed in formulating
such a requirement.

For example, let us consider the Postnikov class $\kkk_{17}$,
which is a simplicial map from $P_{17}$ into another simplicial
set, namely, the Eilenberg--MacLane space $K(\pi_{18}(Y),19)$,
to be introduced in Section~\ref{s:EML1}. The simplices
of $P_{17}$, as well as those of $K(\pi_{18}(Y),19)$,
are canonically represented by certain ordered collections
of integers (or sometimes elements of some $\Z/m$), and
it might happen that while $\size(\sigma)$ is a constant,
$\size(\kkk_{17}(\sigma))$ also depends on the input simplicial
complex $Y$ and becomes arbitrarily large for some choices
of~$Y$.\footnote{Here is an example
of a similar phenomenon in a simpler and perhaps more familiar
setting. Suppose that we want to represent the elements of
the cyclic group $\Z/m$ by the integers $0,1,\ldots,m-1$,
and we want an algorithm for computing the inverse element
$-i$ for a given $i$. Then we cannot require the algorithm
to run in time polynomial in $\size(i)$, because for $i=1$
the output must be $m-1$, and its encoding size depends on~$m$---at
least if we use the standard binary encoding of the integers.
A reasonable requirement is to bound the running time
polynomially in $\size(m)$.} Then $\kkk_{17}(\sigma)$ cannot
be evaluated in time polynomial in~$\size(\sigma)$.

Even if, for every input $Y$, we could compute $\kkk_{17}(\sigma)$
in time polynomial in $\size(\sigma)$, it might happen that
the polynomial depended on $Y$. For example, we might encounter
a sequence $Y^{(1)},Y^{(2)},\ldots$ of inputs such that
$\size(Y^{(j)})\le j$, say, and the time for evaluating
$\kkk_{17}(\sigma)$ is $\size(\sigma)^j$.
Then we would not be able to use such a Postnikov stage
an algorithm such as the one for computing $[X,Y]$ (Corollary~\ref{c:[XY]}),
where the running time should depend polynomially on $\size(Y)$.

Thus, we cannot simply require $\kkk_{17}(\sigma)$ to be computed
in time polynomial in $\size(\sigma)$.
Instead, we are going to require the running time to be
bounded by a polynomial in $\size(Y)+\size(\sigma)$
(where the polynomial depends on $\dim\sigma$ and on $\thedim$,
the index of the Postnikov stage).

To get $Y$ in the picture,
we introduce \emph{parameterized simplicial sets}; these are
families of simplicial sets, typically with infinitely many members,
where each member of the family is described by some value
of a parameter.
We assume some agreed-upon encoding of the parameters
by strings, and the length of the encoding strings is taken as the
size of the corresponding simplicial set in the family.
Then we assume that the black boxes, such as the one
for evaluating $\kkk_{17}$,  take both the parameter value
and $\sigma$ as input, and that they run in time polynomial
in the size of this combined input.

\heading{Locally polynomial-time simplicial sets. }
At this moment we postpone further discussion of the
Postnikov stages $P_\thedim$ and the Postnikov classes $\kkk_\thedim$
until Section~\ref{s:post}, and we introduce a general notion of
a simplicial set represented ``locally'' by polynomial-time black boxes.

\begin{defn}[Locally polynomial-time simplicial set]\label{d:lptss}
Let $\II$ be a set, on which an injective mapping
$\enc\:\II\to\Sigma^*$ is defined, specifying an encoding
of each element of $\II$ by a string; we will refer to $\II$
as a \emph{parameter set}.
We define a \emph{parameterized simplicial set} as a mapping $X$
that, for some parameter set $\II$, assigns to each $I\in\II$
a simplicial set  $X(I)$. Sometimes we will write
such a parameterized simplicial set as $(X(I):I\in\II)$.
We also assume that an encoding of simplices by strings has been
fixed for each of the simplicial sets~$X(I)$.

We say that such an
$X$ is a \emph{locally polynomial-time simplicial set} if,
for each $\thedim$, there is an algorithm that, given
$I\in\II$, a $\thedim$-dimensional simplex $\sigma\in X(I)_\thedim$,
and $i\in\{0,1,\ldots,\thedim\}$, computes $\partial_i\sigma$
in time polynomial in $\size(I)+\size(\sigma)$ (where the polynomial
may depend on $\thedim$), and there is a similar algorithm for evaluating
the degeneracy operators~$s_i\sigma$.

Let $(X(I):I\in\II)$ and $(Y(I):I\in\II)$ be parameterized simplicial
sets with the same parameter set, and for each $I\in\II$,
let $f_I$ be a simplicial map  $X(I)\to Y(I)$.
We say that $f=(f_I\:I\in\II)$ is a \emph{polynomial-time simplicial map}
$X\to Y$ if
for each $\thedim\ge 0$, there is an algorithm that, given
$I\in\II$ and $\sigma\in X(I)_\thedim$, computes $f_I(\sigma)$
in time polynomial in $\size(I)+\size(\sigma)$.\footnote{
More generally, we might want to consider a simplicial map $f_I$ that
goes from $X(F(I))$ to $Y(G(I))$,  for some polynomial-time
computable maps $F,G$. By composing algorithms we may think of
$X(F(I))$ and $Y(G(I))$ as simplicial sets parameterized by $I$
and thus this seemingly more general notion can be interpreted
as a polynomial-time simplicial map in our sense.}
\end{defn}

As was explained above,
the main purpose of the parameterized setting is to make
the polynomial bounds on the running time of the black boxes
uniform in the input of the considered algorithms.

We will see numerous examples  of locally polynomial-time simplicial
sets later on. Of course, the Postnikov stage $P_\thedim=P_\thedim(Y)$,
parameterized by the set of all finite 1-connected simplicial complexes,
is going to be one such example. (However, $P_\thedim$ also
has an additional structure besides being a locally polynomial-time
simplicial set.)

Another, rather simple, example is made of all finite simplicial sets,
as will be discussed at the end
of the present section. Others can be built from this one
by applying various operations, such as products or twisted products,
which will be considered later.

\heading{Locally polynomial-time chain complexes. }
First, let $(X(I):I\in\II)$ be a locally polynomial-time
simplicial set, and let $C_*(X(I))$ be the normalized
chain complex of $X(I)$. This gives us a chain complex parameterized
by $\II$. The $\thedim$-chains
of $C_*(X(I))$ are finite sums of the
form $c=\sum_{\sigma:\alpha_\sigma\ne 0}\alpha_\sigma\cdot\sigma$,
where the $\sigma$ are nondegenerate simplices of $X(I)_\thedim$,
and we can represent such a $c$ by a list of simplices and
of the corresponding nonzero coefficients.
Thus we naturally put
$\size(c):= \sum_{\sigma:\alpha_\sigma\ne 0} (\size(\sigma)+
\size(\alpha_\sigma))$.

For this representation, it is easy to check that the addition
 and subtraction of $\thedim$-chains, as well as the differentials,
can be computed in time polynomial in $\size(I)$ plus the size of
the chains involved. (For this, we need to observe that,
given a simplex $\sigma\in X(I)$, we can test whether it
is degenerate, since every degenerate $\sigma$ satisfies
$\sigma=s_i\partial_i\sigma$ for some~$i$.)

We will also need to work with chain complexes that are not
necessarily normalized chain complexes of simplicial sets.
We will need that the chain groups are ``effectively free,''
meaning that the chains are represented by coefficients
w.r.t.\ some fixed basis. The following definition is a direct
analog of the definition of a locally polynomial-time simplicial set,
and it includes the normalized chain complex of a locally
polynomial-time simplicial set as a special case.

\begin{defn}[Locally polynomial-time chain complex]\label{d:lpcc}
Let $\II$ be a parameter set as in Definition~\ref{d:lptss},
and let $(C(I)_*:I\in\II)$ be a parameterized chain complex,
i.e., a mapping assigning a chain complex to each $I\in\II$.
We say that such a parameterized chain complex
is a \emph{locally polynomial-time chain complex} if the following
hold.
\begin{enumerate}
\setlength{\itemsep}{1pt}
\setlength{\parskip}{0pt}
\setlength{\parsep}{0pt}
\item[\rm(i)]
For each $C(I)_*$ and each $\thedim$, there is a basis
$\Bas_\thedim=\Bas(I)_\thedim$ of $C(I)_\thedim$ (possibly infinite),
which we call the \emph{distinguished basis}\footnote{Chain
complexes with a distinguished basis for each chain group
are sometimes called \emph{cellular}.}
of $C(I)_\thedim$, and whose elements have some agreed-upon
encoding by strings. An arbitrary
$\thedim$-chain $c\in C(I)_\thedim$ is (uniquely) represented
as an integer linear combination of elements of $\Bas(I)_\thedim$,
i.e.,  by a finite list of elements of $\Bas(I)_\thedim$ and the
corresponding nonzero coefficients. (This also defines
the encoding size for chains.)
\item[\rm(ii)]
For every fixed $\thedim$,  there is an algorithm
for evaluating the differential $\diff_\thedim$ of $C(I)_*$, which
computes $\diff_\thedim(c)$ in time polynomial in $\size(I)+\size(c)$.
\end{enumerate}
\end{defn}

We note that in the representation of $\thedim$-chains as in (i),
the chains $c+c'$ and $c-c'$ can be computed in time polynomial
in $\size(c)+\size(c')$, even without including $\size(I)$.

If $(C(I)_*:I\in\II)$ and $(\tC(I)_*:I\in\II)$ are parameterized
chain complexes, then, in complete analogy with polynomial-time
simplicial maps in Definition~\ref{d:lptss}, we define a
\emph{polynomial-time chain map} $\varphi=(\varphi_I)_{I\in\II}\:
C_*\to\tC_*$, where each $\varphi_I$ is a chain map
$C(I)_*\to \tC(I)_*$,
such that for each fixed $\thedim$, $(\varphi_I)_\thedim(c)$
can be computed in time polynomial in $\size(I)+\size(c)$.

\heading{Changing the parameter or: preprocessing. }
Let $(X(J):J\in\JJ)$ be a parameterized simplicial set,
and let $F\:\II\to\JJ$ be a polynomial-time mapping
of another parameter set $\II$ into $\JJ$.
Then we can define a new parameterized simplicial set $(\tilde X(I):I\in\II)$
by $\tilde X(I):=X(F(I))$; if $X$ is locally polynomial-time, then so
is~$\tilde X$.

In our algorithms, $X$ can often be regarded as a version of $\tilde X$
``with preprocessing''. For this, the parameter $J$ will typically
be of the form $(I,G(I))$, where $I$ is the original parameter
and $G$ is some polynomial-time map. Here $G(I)$ represents
some auxiliary data computed from~$I$.

For example, if we regard the
Postnikov stage $P_\thedim(Y)$ as parameterized by the finite
simplicial complex $Y$, then by Definition~\ref{d:lptss}, the
algorithm for evaluating $\partial_i\sigma$ receives $Y$
and $\sigma$ as input. Thus, each time we want to
know the $i$th face of some simplex, all of the
computations are done from scratch.

In the algorithm from Theorem~\ref{t:ipost}
for constructing a Postnikov system, we will proceed differently:
given $Y$, we first compute, once and for all, some data based on
$Y$, such as the first $\thedim$ homotopy groups of $Y$.
Then we will represent $P_\thedim$ using these data (concretely,
as a twisted product of suitable Eilenberg--MacLane spaces),
instead of the ``raw'' representation by $Y$,
so that these computations can be reused in all subsequent
computations of face operators in $P_\thedim$. This will make
the computation of the face operators and other operations
with the Postnikov system much more efficient, although if
we care only about the distinction polynomial/non-polynomial,
both ways are equivalent.

\heading{Keeping the parameters implicit. }
Although a locally polynomial-time simplicial set $(X(I):I\in\II)$
is defined as a mapping assigning a simplicial set $X(I)$ to every value
of $I$, in most cases we can think of it as a single simplicial set
$X$. The exact nature of the parameter $I$ usually does not matter;
it may be useful to keep in mind that $X$ is actually parameterized,
but in most of the subsequent discussion, we will suppress the parameter.

This is in agreement with the common
practice in the literature on polynomial-time algorithms, where
phrases like ``the resulting graph has a polynomial size'' are used,
which are also formally imprecise but easily understood.

\heading{Converting finite simplicial complexes into simplicial sets. }
Here we make a slight digression and describe how a finite simplicial
complex, which is one of the possible kinds of inputs for
our algorithms, is (canonically)
converted into a simplicial set.

Given a finite simplicial complex $K$,
the corresponding simplicial set $\sset(K)$,
which in particular has the same geometric realization as $K$
and thus specifies the same topological space, is defined
as follows.
The $\thedim$-dimensional \emph{nondegenerate} simplices of $\sset(K)$
are just the $k$-simplices of $K$, with the face operators defined
in the obvious way. It remains to specify the degenerate simplices
and the face and degeneracy operators on them.
For this,
we can use a standard fact about simplicial sets:
every degenerate simplex $\tau$ can be expressed as
$s_{i_t}s_{i_{t-1}}\cdots s_{i_1}\sigma$, where $\sigma$ is a
uniquely determined nondegenerate simplex of $X$ and $i_1<i_2<\cdots<i_t$
is a uniquely determined sequence of integers. Thus, we can
represent $\tau$ by $\sigma$ and $i_1,\ldots,i_t$. With
this representation, the face and degeneracy operators can be
evaluated by simple rules; see, e.g.,
 \cite{Friedm08,May:SimplicialObjects-1992}.
(Also see \cite[Section~3]{Friedm08} for another,
simpler way of adding degenerate simplices to a simplicial
complex.)

Then $(\sset(K): K\in\FSC)$
forms a locally polynomial-time simplicial set, whose parameter
set $\FSC$ consists of all finite simplicial complexes.

More generally, we can consider the family of all finite simplicial
sets, which are given
by  lists of nondegenerate simplices for each
of the relevant dimensions and tables specifying the face
operators, and where the degenerate
simplices and degeneracy operators are represented as above.
Then the identity map on $\FSS$ forms a locally polynomial-time
simplicial set.

\subsection{Reductions, strong equivalences, and polynomial-time homology}

It turns out that the notion of locally polynomial-time simplicial
set is too weak for most computational purposes. We can inspect such
a simplicial set locally, but it is in general impossible to
compute useful global information about it, such as homology groups
or homotopy groups.

Here we introduce a stronger notion of simplicial set
with polynomial-time homology, modeled after simplicial
sets with effective homology due to Sergeraert et al.
This is a (parameterized) locally polynomial-time simplicial
set $X$ whose normalized chain complex $C_*(X)$ is, moreover, associated
with another, typically much smaller chain complex $\EC_*$,
which we can think of as a finitary approximation of $C_*(X)$.
(The notation $\EC_*$ follows \cite{SergerGenova}, and it should
suggest that $\EC_*$ is an ``effective version'' of $C_*$.)
The chain groups $\EC_\thedim$ have polynomially many
generators for every fixed $\thedim$, and thus we can compute
each homology group $H_\thedim(\EC_*)$ in polynomial time.
The association of
$\EC_*$ with $C_*(X)$ is such that these homology computations
in $\EC_*$ can be ``pulled back'' to $C_*(X)$.
We will now define the properties of $\EC_*$ and the
way it is associated with $C_*(X)$ in detail.

\begin{defn}[Globally polynomial-time chain complexes]\label{d:glob-ptime}
A \emph{globally polynomial-time chain complex}
is a locally polynomial-time chain complex $(\EC(I)_*:I\in\II)$
such that,  for each fixed $\thedim$, the chain group
$\EC(I)_\thedim$ is finitely generated, and there is an
algorithm that, given $I\in \II$, outputs
the list of elements of the distinguished basis $\Bas(I)_\thedim$ of
$\EC(I)_\thedim$, in time bounded by a polynomial in $\size(I)$
(and in particular, the rank of $\EC(I)_\thedim$ is bounded
by a polynomial in $\size(I)$).
\end{defn}

We note that, for a globally polynomial-time $\EC_*$
and each fixed $\thedim$, we can compute the matrix
of the differential $\diff_\thedim\:\EC_{\thedim}\to\EC_{\thedim-1}$
w.r.t.\ the distinguished bases in polynomial
time---we just evaluate $\diff_\thedim$
on each element of the distinguished basis $\Bas_\thedim$.
Then the homology groups $H_\thedim(\EC_*)$
is computed using a Smith normal form algorithm
applied to the matrices of $\diff_\thedim$ and
$\diff_{\thedim+1}$, as is explained in standard textbooks
(such as \cite{Munkres}). Polynomial-time algorithms for
the Smith normal form are nontrivial but known
\cite{KannanBachem}; also see
\cite{Storjohann:NearOptimalAlgorithmsSmithNormalForm-1996}
for apparently the asymptotically fastest deterministic algorithm.

\heading{Globally polynomial-time Abelian groups. }
By the above, we can compute $H_\thedim(\EC_*)$ in polynomial
time. We represent its isomorphism type%
\footnote{To get a bijective correspondence with isomorphism types,
we should ask for divisibility $m_1|\cdots|m_s$.
We do not care about uniqueness, however, and thus we will not require this.}
in the usual way, as a direct sum
$\Z^r\oplus (\Z/m_1)\oplus (\Z/m_2)\oplus\cdots\oplus (\Z/m_s)$.
But in our algorithms, we are not interested just in knowing this
description of the homology group; we will also need to work
with its elements, with homomorphisms into it, etc.
Moreover, since the chain complex $\EC_*$ is parameterized,
the homology group $H_\thedim(\EC_*)$ should be regarded
as parameterized as well (and similarly for homotopy groups
of parameterized simplicial sets). We thus define a globally
polynomial-time Abelian group in analogy with a globally
polynomial-time chain complex.

First, let $\MM$ be the set of all $(r+s)$-tuples
$\mm=(0,0,\ldots,0,m_1,\ldots,m_s)$ specifying isomorphism
types of finitely generated Abelian groups in the way introduced
in Section~\ref{s:encsize}. For $\mm\in\MM$, let $\AB(\mm)$
be the group $\Z^r\oplus (\Z/m_1)\oplus\cdots\oplus (\Z/m_s)$,
with elements represented by $(r+s)$-tuples $(\alpha_1,\ldots,\alpha_{r+s})$,
$\alpha_1,\ldots,\alpha_r\in\Z$, $\alpha_{r+i}\in \Z/m_i$.
Here $\AB(\mm)$ can be regarded as a canonical representation
of an Abelian group with the isomorphism type~$\mm$.

Now we define a \emph{parameterized Abelian group} and
\emph{locally polynomial-time Abelian group} in an obvious
analogy with the corresponding notions for simplicial sets
and chain complexes.
A \emph{globally polynomial-time Abelian
group} $(\pi(I):I\in\II)$ is a locally polynomial-time Abelian
group equipped with a polynomial-time algorithm that,
given $I\in\II$, returns an $\mm\in\MM$  specifying the isomorphism
type of $\pi(I)$, and with a polynomial-time isomorphism
of $\pi(I)$ with $\AB(\mm)$. In more detail, in time polynomial
in $\size(I)$ we can compute a basis $(b_1,b_2,\ldots,b_{r+s})$
of $\pi(I)$ such that $b_i$ generates the $i$th cyclic summand
isomorphic to $\Z$ (for $i\le r$) or $\Z/m_{i-r}$ (for $i>r$)
in an expression of $\pi(I)$ as a direct sum. Moreover, given
an arbitrary element $a\in\pi(I)$, in time polynomial
in $\size(I)+\size(a)$ we can compute the coefficients
$\alpha_1,\ldots,\alpha_{r+s}$  such that $a=\sum_{i=1}^{r+s}\alpha_i b_i$.
This provides the isomorphism $\pi(I)\to \AB(\mm)$, and the
inverse mapping is also obviously polynomial-time computable.

We now consider the globally polynomial-time chain complex
$\EC_*$ parameterized by $\II$.  We want to regard $H_\thedim(\EC_*)$
as a globally polynomial-time Abelian group parameterized
by $\II$. To this end, we need that the computation of
$H_\thedim(\EC(I)_*)$ returns its isomorphism type $\mm$,
and also fixes an isomorphism of $H_\thedim(\EC(I)_*)$
with $\AB(\mm)$. Such an isomorphism is naturally
obtained from the Smith normal form algorithm.\footnote{Formally,
for this we need the Smith normal form algorithm to be
\emph{deterministic}, so that it always returns the same
isomorphism for a given $I$ (which need not be true for
a randomized algorithm, for example). However, in an actual
implementation, this issue does not arise, since anyway we
want to store in memory the Smith normal form once computed for a
given $I$, in order to avoid repeated computations.}
In this way, $H_\thedim(\EC_*)$ becomes a globally
polynomial-time Abelian group parameterized by $\II$.

Moreover, given a chain $z\in Z_\thedim(\EC(I)_*)$, we can
compute in polynomial time the corresponding homology class $[z]
\in H_\thedim(\EC(I)_*)$. This defines a polynomial-time homomorphism
$Z_\thedim(\EC_*)\to H_\thedim(\EC_*)$, also parameterized by~$\II$.
Slightly more generally,
given a chain $c\in \EC_\thedim$, we can decide
whether $c$ is a cycle, and if yes, compute $[c]$.
Moreover, if $[c]$ is zero, that is, if $c$ is a boundary,
we can also compute a ``witness,'' i.e., a $(\thedim+1)$-chain $b$
 with $c=\diff_{\thedim+1}b$. Conversely, given
 $h\in H_\thedim(\EC_*)$,
we can compute a \emph{representing cycle}, i.e.,
$z\in Z_\thedim(\EC_*)$ with $[z]=h$.
All of these calculations are easily done in polynomial time using
the Smith normal form of the matrices of the differentials.

\heading{Reductions. } Now we start discussing the way of associating
a ``small'' chain complex $\EC_*$ with a ``big'' chain complex $C_*$.
First we deal with the usual setting of homological algebra,
where we consider individual chain complexes, rather than
parameterized ones, and then we add some remarks on transferring
the notions to the setting of parameterized chain complexes and maps.

The most common way in algebraic topology of making two chain complexes
$C_*$ and $\tC_*$ ``equivalent'' is \emph{chain homotopy equivalence},
but for effective homology and polynomial-time homology, it
is more convenient to use two special cases
of chain homotopy equivalences, namely, \emph{reduction} and \emph{strong equivalence}.

If $f,g\:C_*\to\tC_*$ are two chain maps, then a
\emph{chain homotopy} of $f$ and $g$ is a sequence $(h_\thedim)_{\thedim\in\Z}
$ of homomorphisms, where $h_\thedim\: C_\thedim\to\tC_{\thedim+1}$
(raising the dimension by one),
such that $g_\thedim-f_\thedim=\tilde\diff_{\thedim+1} h_\thedim +
h_{\thedim-1}\diff_\thedim$. Chain maps and chain homotopies
can be regarded as algebraic counterparts of continuous maps
of spaces and their homotopies, respectively. In particular,
two chain-homotopic chain maps induce the same map in homology.


\begin{defn}[Reduction\footnote{In
a part of the literature, other notions such as  \emph{chain
contraction} or \emph{strong deformation retraction} are used
instead of the word \emph{reduction}. For instance  Eilenberg
and Mac~Lane
\cite[Section~12]{EilenbergMacLane:GroupsHPin1-1953} use the
word \emph{contraction}, while \emph{reduction}
has a different meaning there.}]\label{d:redu}
Let $C_*$ and $\tC_*$ be chain complexes.
A \emph{reduction} $\rho$ from $C_*$ to $\tilde C_*$ consists
of three maps $f=(f_\thedim)_{\thedim\in\Z},g=
(g_\thedim)_{\thedim\in\Z},h=(h_\thedim)_{\thedim\in\Z}$, such that
\begin{enumerate}
\setlength{\itemsep}{1pt}
\setlength{\parskip}{0pt}
\setlength{\parsep}{0pt}
\item[\rm(i)] $f\:C_*\to \tilde C_*$ and $g\:\tilde C_*\to C_*$
are chain maps;
\item[\rm(ii)] the composition $f  g\:\tilde C_*\to\tilde C_*$
is equal to the identity $\id_{\tilde C_*}$, while the composition
$g  f\:C_*\to C_*$ is chain-homotopic to $\id_{C_*}$,
with $h\:C_*\to C_*$ providing the chain homotopy, i.e.~$\id_{C_*}-gf=\diff h+h\diff$; and
\item[\rm(iii)] $f  h=0$, $h  g=0$, and $h  h=0$.
\end{enumerate}
We write
\[
C_*\redu \tC_*
\]
if there is a reduction from
$C_*$ to $\tC_*$.
\end{defn}

A reduction $C_*\redu \tC_*$ can be depicted by the following diagram:
\[
\xymatrix{
C_* \POS\save[]+L+<5pt,0pt>*++[o]{}\ar@(ul,dl){}_h\restore
\ar@/^/[rr]^f && \tC_* \ar@/^/[ll]^g
}
\]
Intuitively, such a reduction is a tool that allows us to
reduce questions about homology of a ``big'' chain complex $C_*$
to questions about homology of a ``smaller'' chain complex $\tC_*$.
In particular, the existence of a reduction $C_*\redu \tC_*$
implies that $H_\thedim(C_*)\cong H_{\thedim}(\tC_*)$ for all~$\thedim$.

It is easily checked that $(f,g,h)$ is a reduction $C_*\redu \tC_*$ and
$(f',g',h')$ is a reduction $\tC_*\redu \dtC_*$,
then there is a reduction $C_*\redu  \dtC_*$,
namely,
$
(f' f,g g', h+gh'f)
$ 
(see, e.g., \cite[Proposition~59]{SergerGenova}).
We will also need a (straightforward) extension to composing
a larger number of reductions (the proof is omitted).

\begin{lemma}\label{l:comp-r0}
Let $C_*^{(1)},\ldots,C_*^{(n)}$ be chain complexes, and let
$\rho^{(i)}=(f^{(i)},g^{(i)},h^{(i)})$ be a reduction
$C_*^{(i)}\redu C_*^{(i+1)}$, $i=1,2,\ldots,n-1$. Then
the reduction $(f,g,h)\:C_*^{(1)}\redu C_*^{(n)}$ obtained by composing
these reductions is given by
$f=f^{(n-1)}f^{(n-2)}\cdots f^{(1)}$, $g=g^{(1)}g^{(2)}\cdots g^{(n-1)}$,
and
$$
h=h^{(1)}+g^{(1)}h^{(2)}f^{(1)}+\cdots+
g^{(1)}g^{(2)}\cdots g^{(n-2)}h^{(n-1)}f^{(n-2)}\cdots f^{(1)}.
$$
\end{lemma}

\heading{Strong equivalences. }
While reductions $C_*\redu \tC_*\redu \dtC_*$
compose to a reduction $C_*\redu\dtC_*$,
in some constructions one naturally arrives at a different kind
of situation:
\begin{equation}\label{e:steq}
C_*\lredu \dtC_*\redu \tC_*.
\end{equation}
Here we have no natural way of composing the reductions to obtain
a reduction between $C_*$ and $\tC_*$.
For algorithmic purposes, we regard
the situation (\ref{e:steq}) as a primitive notion,
called \emph{strong chain homotopy equivalence} or just
strong equivalence.

\begin{defn}[Strong equivalence] A \emph{strong equivalence}
of chain complexes $C_*$ and $\tC_*$, in symbols
$C_*\steq \tC_*$, means that there
exists another chain complex $\dtC_*$ and reductions
$C_*\lredu \dtC_*\redu \tC_*$.
\end{defn}

\begin{lemma}\label{l:tr-steq0}
Strong equivalence is transitive:
if $C_*\steq \tC_*$ and $\tC_*\steq \dtC_*$,
then $C_*\steq \dtC_*$.
\end{lemma}

\begin{proof} There are several proofs available.
One of them follows \cite[Proposition~124]{SergerGenova}
(using the algebraic mapping cylinder). Another possibility
is to regard reductions as special cases of chain homotopy
equivalences, which are closed under composition, and then
show that a chain homotopy equivalence can be converted into
a strong equivalence, also using a suitable mapping cylinder---see,
e.g., \cite{BarnesLambe}, \cite[Sec.~3]{RomeroEllisRubio}.

Here we offer yet another short proof.
Let us consider strong equivalences $C_*\lredu A_*\redu \tC_*$
and $\tC_*\lredu A'_*\redu \dtC_*$. In view of Lemma~\ref{l:comp-r0}
it is suffices to exhibit a strong equivalence
$A_*\steq A'_*$.

Let the  reduction $A_*\redu \tC_*$ be $(f,g,h)$ and
let the reduction $A'_*\redu \tC$ be  $(f',g',h')$.
We construct a new chain complex $D_*$, the
\emph{double mapping cylinder}
of the pair of maps $A_*\xleftarrow{g}\tC_*\xrightarrow{g'}A'_*$
(this construction is analogous to the mapping cylinder
introduced earlier).
Its chain groups are
\[D_\thedim:=A_\thedim\oplus \tC_{\thedim-1}\oplus A'_\thedim\]
and the differential is given by $\diff^{D}(a,c,a')
:=(\diff(a)-g(c),-\tilde\diff(c),\diff'(a')+g'(c))$
(where $\diff,\tilde\diff,\diff'$ are differentials
in $A_*$, $\tC_*$, and $A'_*$, respectively).
It is easily checked that $D_*$ indeed forms a chain complex.

We now describe a reduction $(F,G,H)\:D_*\redu A_*$; we set
\[F(a,c,a'):=a+gf'(a'),\ G(a)=(a,0,0),\ H(a,c,a')=(0,f'(a'),h'(a')).\]
The reduction $(F',G',H')\:D_*\redu A'_*$ is obtained almost symmetrically as
\[F'(a,c,a'):=a'+g'f(a),\ G'(a')=(0,0,a'),\ H'(a,c,a')=(h(a),-f(a),0).\]
Checking that both $(F,G,H)$ and $(F',G',H')$ are indeed reductions is routine and we omit it.
\end{proof}

\heading{Polynomial-time reductions and strong equivalences. }
Let $(C(I)_*:I\in\II)$ and $(\tC(I)_*:I\in\II)$
be two locally polynomial-time chain complexes with the \emph{same} parameter
set. A \emph{polynomial-time reduction} of $C_*$ to $\tC_*$,
in symbols
\[C_*\reduP\tC_*,
\]
 is a triple $\rho=(f,g,h)$.
Here $f=(f_I)_{I\in\II}$ is a polynomial-time chain map
$C_*\to\tC_*$,
$g=(g_I)_{I\in\II}$ is a polynomial-time chain map $\tC_*\to C_*$,
and $h=(h_I)_{I\in\II}$ is a polynomial-time chain homotopy
$C_*\to C_*$, defined in obvious analogy with a
polynomial-time chain map.
For each $I$,
$(f_I,g_I,h_I)$ form a reduction $C(I)_*\redu \tC(I)_*$
according to Definition~\ref{d:redu}.

Similarly, we define a polynomial-time strong equivalence
of two locally polynomial-time chain complexes, $C_*\steqP\tC_*$,
with the middle chain complex also locally polynomial-time and
with the same parameterization as $C_*$ and $\tC_*$.

By the fact that a composition of any constant number of
polynomial-time maps is polynomial-time, it is easy to check
that the proof of Lemma~\ref{l:tr-steq0}
yields the following.

\begin{corol}\label{c:steq}
Polynomial-time strong equivalence of locally polynomial-time
chain complexes is transitive:
$C_*\steqP \tC_*$ and $\tC_*\steqP \dtC_*$
implies $C_*\steq \dtC_*$.
\end{corol}

\heading{Polynomial-time homology.} With the notions
of polynomial-time strong equivalence and globally polynomial-time
chain complex, the definition of polynomial-time homology
is now straightforward.

\begin{defn}[Chain complexes and simplicial sets with polynomial-time homology]
We say that a parameterized chain complex $C_*$ is equipped with
\emph{polynomial-time homology} if $C_*$ is locally polynomial-time
and there are a globally polynomial-time chain complex $\EC_*$
and a polynomial-time strong equivalence $C_*\steqP\EC_*$.

A parameterized simplicial set $X$ is equipped with polynomial-time homology
if $X$ is locally polynomial-time and its normalized chain
complex $C_*(X)$ is equipped with polynomial-time homology.
\end{defn}

We should perhaps stress that equipping a parameterized simplicial set $X$
with polynomial-time homology does \emph{not} mean only
the ability of computing the homology groups of $X(I)$ in time polynomial
in $\size(I)$ (for every fixed dimension); this ability is a consequence
of polynomial-time homology, but in itself it would not be sufficient.

For one thing, if $X$ is equipped with polynomial-time homology,
then for $C_*(X)$
we can do all of the computations mentioned
after Definition~\ref{d:glob-ptime}: finding a representative
of a given homology class, the
homology class of a given chain,
and a witness for being a boundary.

Moreover, the definition of polynomial-time homology, following
the earlier notion of effective homology by Sergeraert et al., is designed
so that it has the following meta-property: if $X^{(1)},\ldots,X^{(t)}$
are simplicial sets equipped with polynomial-time homology
and $\Phi$ is a ``reasonable'' way of constructing a new
simplicial set from $t$ old ones, then the simplicial set
$\Phi(X^{(1)},\ldots,X^{(t)})$
can also be equipped with polynomial-time homology
(some of the constructions also involve polynomial-time simplicial
maps, polynomial-time chain maps, etc.).
Of course, this is only a guiding principle, and for every specific
construction $\Phi$ used in our algorithm, we need a corresponding
result about preserving polynomial-time homology by~$\Phi$.
The next section is devoted to such results.

The reader may also wonder what are homology computations good for
in algorithms for computing homotopy groups and Postnikov systems.
The connection is via the \emph{Hurewicz isomorphism}, which in
its simplest form asserts that, for a 1-connected space $Y$,
the first nonzero homotopy group of $Y$ occurs in the same dimension
as the first nonzero homology group, and these two groups
are isomorphic. Thus, roughly speaking,
to find $\pi_\thedim(Y)$, the Postnikov system
algorithm ``kills'' the first $\thedim-1$ homotopy groups
of $Y$ by constructing the mapping cone of $\varphi_{\thedim-1}\:Y\to
P_{\thedim-1}$ with polynomial-time homology,
and then it computes the appropriate homology group of this cone.

Let us remark that in \cite{pKZ1}, polynomial-time homology was defined
using only reductions, rather than strong equivalences (since strong
equivalences were not needed there). Of course, a reduction is
a special case of strong equivalence, so the definition here
is more permissive.

\section{A toolbox of
operations for polynomial-time homology}\label{s:ops}

In this longish section we will build a repertoire of algorithmic
operations on simplicial sets and chain complexes, in such a way that
if the input objects come with polynomial-time homology, the output
object is also equipped with polynomial-time homology.

As was mentioned in the introduction, we mostly review known
methods, developed for effective homology and based on much older
work by algebraic topologists. We try to make the presentation
streamlined and mostly self-contained, and in particular,
we describe the algorithms in full, sometimes referring
to the literature for details of proofs. Moreover, there are places
where polynomiality requires extra analysis or assumptions;
most notably, Section~\ref{s:prod} (products of many factors)
and Section~\ref{s:EML2} (polynomial-time homology for $K(\Z/m,1)$)
contain some new material.

For the rest of the paper, we will use only three specific results
of this section: Proposition~\ref{p:mcone} (mapping cone),
Corollary~\ref{c:pullback} (a certain pullback operation), and
Theorem~\ref{t:EML} (polynomial-time homology for Eilenberg--MacLane
spaces). But we will also need some of the notions and simple
facts introduced here.

Let us remark that some of the operations can be implemented
in several different ways. For example, polynomial-time homology
for $K(\Z/m,1)$ can most likely be obtained directly
by modifying the method of \cite{pKZ1} used for $K(\Z,1)$,
and for the passage from $K(\pi,\thedim)$ to $K(\pi,\thedim+1)$,
one could also use the method in \cite[Chap.~4]{RealThesis}
(also see \cite{AlvarezArmarioFrauReal}).
Our main criterion for selecting
among the various possibilities was simplicity of presentation
and general applicability of the tools. 

Moreover, the chain complexes that appear naturally during our construction of Postnikov systems can often be equipped with an additional algebraic structure. For instance the chain complex $C_*(K(\pi,\thedim))$ has a structure of the so-called Hopf-algebra; that is, $C_*(K(\pi,\thedim))$ is endowed with an algebra and a coalgebra structures that are compatible in some strong sense.
The structure is often ``transferred'' through the chain equivalences to the globally polynomial-time counterparts. 

As was suggested by a referee, it is possible that using this additional structure might lead to an algorithm more efficient in practice. 
The polynomial running-time bounds might also improve,  and so
investigating the algorithmic use of these additional
structures is a worthwhile research direction. On the other hand,
in view of the W[1]-hardness result \cite{Mat-homotopyW1} mentioned above,
such improvements cannot remove the dependence of the degree of the
polynomial on $\thedim$.
Thus, since our goal at this stage is to get polynomial-time
algorithms, and in order to keep the presentation simple, we do not
discuss these additional algebraic structures in this paper.

\subsection{Products}\label{s:prod}

We recall that the \emph{product} $X\times Y$ of simplicial
sets $X$ and $Y$ is the simplicial set whose
$\thedim$-simplices are ordered pairs $(\sigma,\tau)$,
where $\sigma\in X_\thedim$ and $\tau\in Y_\thedim$. The face and
degeneracy operators are applied to such pairs componentwise.
We have $|X\times Y|\cong|X|\times |Y|$ for geometric
realizations.\footnote{To be precise,
the product of topological spaces on the right-hand side
should be taken in the category of $k$-spaces; but for the
spaces we encounter, it is the same as the usual product
of topological spaces.} The definition of the product is deceptively
simple, but actually it hides a sophisticated way of triangulating
the product (and degenerate simplices play a crucial role)---see
\cite{SergerGenova} or~\cite{Friedm08} for an explanation.

As shown by Sergeraert et al.\ as
one of the first steps in the theory of effective homology,
if effective homology is available for  $X$ and $Y$, then
it can also be obtained for $X\times Y$. The core of this
result is the Eilenberg--Zilber theorem (see, e.g.,
\cite[Theorem~123]{SergerGenova}), which provides a reduction
of $C_*(X\times Y)$ to the tensor product $C_*(X)\otimes C_*(Y)$,
and which goes back to Eilenberg and Mac~Lane
\cite{EilenbergMacLane:GroupsHPin1-1953,EilenbergMacLane:GroupsHPin2-1954}.
The proof immediately shows that polynomial-time homology
for $X,Y$ yields polynomial-time homology for $X\times Y$.

However, this works directly only
for products of two, or constantly many, factors,
while we need to deal with products
$X^{(1)}\times\cdots \times X^{(n)}$ of arbitrarily many factors.
There the situation with polynomiality
is somewhat more subtle, and we will
actually need an additional condition on the $X^{(i)}$'s
in order to obtain polynomial-time homology.
We begin with defining the notion needed for the extra condition.

\begin{defn}[$\thedim$-reduced] A simplicial set $X$ is
\emph{$\thedim$-reduced}, where $\thedim\ge 0$ is an integer,
if $X$ has a single $0$-simplex (vertex) and no nondegenerate
simplices of dimensions $1$ through~$k$.
We call a chain complex $C_*$ \emph{$\thedim$-reduced} if
$C_0\cong \Z$ and $C_i=0$ for $1\le i\le\thedim$.
\end{defn}

We remark that $\thedim$-reducedness is a very useful property
of simplicial sets, which has no analog for simplicial complexes.
For example, being $\thedim$-reduced is an easily checkable certificate
for $\thedim$-connectedness.

\begin{prop}[Product with many factors]\label{p:bigprod}
Let $(X(I):I\in\II)$ be a simplicial set with polynomial-time
homology. Let us form a new parameter set
$\JJ=\bigcup_{n=1}^\infty \II^n$, where $\II^n$ is the $n$-fold
Cartesian product, and let $(W(J):J\in\JJ)$ be the parameterized
simplicial set of products, with $W(I_1,I_2,\ldots,I_n):=
X(I_1)\times \cdots\times X(I_n)$. For $J=(I_1,\ldots,I_n)\in\JJ$,
let $\size(J)=\sum_{i=1}^n\size(I_i)$, and for a simplex
$\sigma=(\sigma_1,\ldots,\sigma_n)\in W(J)$, let
$\size(\sigma)=\sum_{i=1}^n\size(\sigma_i)$.
Let us also assume that
all the $X(I)$ and all the chain complexes witnessing
polynomial-time homology for $X$ are $0$-reduced.
Then $W$ can be equipped with polynomial-time homology.
\end{prop}

For reasons of ``uniform polynomiality'', we needed to assume
that the factors in the considered products are all instances
of a single parameterized simplicial set. However, as we
remarked above, the product of a \emph{constant} number
of arbitrary, possibly different, simplicial sets with polynomial-time
homology can be equipped with polynomial-time homology.
This allows us to obtain polynomial-time
homology for products where all but a constant
number of factors are $0$-reduced and come from the same
parameterized simplicial set, while the remaining factors
are arbitrary.

In the forthcoming proof, for brevity, we are going to write
$X^{(i)}$ instead of $X(I_i)$, and use similar abbreviations
for chain complexes.

\heading{Tensor products. } Before discussing the proof,
we need some preparations concerning tensor products.
Let $C_*^{(1)}$ and $C_*^{(2)}$ be chain complexes, and suppose,
as we do for locally polynomial-time chain complexes, that
each chain group $C_\thedim^{(i)}$ has a distinguished basis
$\Bas_\thedim^{(i)}$. Then the \emph{tensor product}
$T_*:=C^{(1)}_*\otimes C^{(2)}_*$ can be defined as the
chain complex in which $T_\thedim$ is the free Abelian
group over the distinguished basis
\[
\Bas_\thedim:=\{b_1\otimes b_2:
b_1\in \Bas_{\thedim_1}^{(1)},
b_2\in \Bas_{\thedim_1}^{(2)}, \thedim_1+\thedim_2=\thedim\}.
\]

Here we may regard $b_1\otimes b_2$ just as a formal symbol.
For arbitrary chains $c_1\in C_{\thedim_1}^{(1)}$,
$c_2\in C_{\thedim_2}^{(2)}$, $\thedim_1+\thedim_2=\thedim$,
the $\thedim$-chain
$c_1\otimes c_2$ is then defined using linearity of $\otimes$
in both operands, as the appropriate linear combination
of the elements of~$\Bas_{\thedim}$.

The differential in $T_*$ is given on the elements of $\Bas_\thedim$ by
\begin{equation}\label{e:difffmla}
\diff_\thedim (b_1\otimes b_2):=\diff_{\thedim_1}^{(1)}(b_1)\otimes
b_2+(-1)^{\thedim_1}b_1\otimes\diff_{\thedim_2}^{(2)}(b_2),
\end{equation}
where as above, $\thedim_i=\deg(b_i)$.

Next, let us consider the tensor product
$T_*:=C^{(1)}_*\otimes\cdots\otimes C^{(n)}_*$  of many factors.
The distinguished basis $\Bas_\thedim$ now consists of
elements $b_1\otimes\cdots\otimes b_n$, with each $b_i$ an element
of a distinguished basis in $C_*^{(i)}$,  $\sum_{i=1}^n\deg(b_i)=\thedim$.
Hence the rank of $T_\thedim$ equals
\begin{equation}\label{e:rank-tens}
\rank (T_\thedim)=\sum_{\thedim_1+\cdots+\thedim_n=\thedim}\ \  \prod_{i=1}^n
\rank (C^{(i)}_{\thedim_i}).
\end{equation}
Thus, if many of the $C^{(i)}$ are not $0$-reduced,
already $\rank(T_0)$ is exponentially large; for example, if each
$C^{(i)}_0$ is $\Z\oplus\Z$, then $\rank(T_0)=2^n$.
This is the basic reason
why we need the 0-reducedness conditions in Proposition~\ref{p:bigprod}.
If, on the other hand, all the $C_*^{(i)}$'s are 0-reduced, then so is~$T_*$.

The key to the polynomial-time bounds we need is the following lemma.

\begin{lemma} \label{l:tenspr}
Let $(C(I)_*:I\in\II)$ be a locally polynomial-time chain complex,
with all the $C(I)_*$ $0$-reduced,
let $\JJ$ be the parameter set as in Proposition~\ref{p:bigprod},
and let $(T(J)_*:J\in\JJ)$ be the parameterized set of tensor products,
with $T(I_1,\ldots,I_n)_*=C_*^{(1)}\otimes\cdots \otimes C_*^{(n)}$
(where $C_*^{(i)}$ abbreviates $C(I_i)_*$),
and with the same definitions of encoding sizes as in
Proposition~\ref{p:bigprod}.
Then $T_*$ is also $0$-reduced and locally polynomial-time,
and given chains $c_i\in C^{(i)}_{\thedim_i}$ with
$\sum_{i=1}^n\thedim_i=\thedim$, the $\thedim$-chain
$c_1\otimes\cdots\otimes c_n$ can be computed (i.e., expressed
in the distinguished basis of $T_\thedim(J)$) in time polynomial
in $\size(J)+\sum_{i=1}^n\size(c_i)$,
assuming $\thedim$ fixed.
\end{lemma}

\begin{proof} To show that the differential $\diff_\thedim$ of $T_*$ is
a polynomial-time map, it is enough to consider
computing it on elements $b_1\otimes\cdots\otimes b_n$
of the standard basis. By iterating the differential formula
(\ref{e:difffmla}), we can express $\diff_\thedim (b_1\otimes\cdots\otimes b_n)$
as a sum of $n$ terms of the form $\pm c_1\otimes\cdots\otimes c_n$,
where each $c_i$ is either $b_i$ or $\diff_{\thedim_i}(b_i)$.
For evaluating this sum it is thus sufficient to be able
to evaluate $c_1\otimes\cdots\otimes c_n$ in polynomial time,
as in the second claim of the lemma.

As for this second claim, we use the observation that
if $\deg(c_1\otimes\cdots\otimes c_n)=\thedim$, then all but
at most $\thedim$ of the $c_i$'s have degree $0$. Suppose that
only $c_1,\ldots,c_\thedim$ have nonzero degrees. Then we can
compute $c_1\otimes\cdots\otimes c_\thedim$ in a straightforward
way (at most $\prod_{i=1}^\thedim \size(c_i)$
basis elements are involved, which is polynomially bounded
for fixed $\thedim$). Then the tensor product of the result with
$c_{\thedim+1}\otimes\cdots\otimes c_n$ amounts just to
multiplying all coefficients by a number (since
$C_0^{(\thedim+1)}\cong\cdots\cong C_0^{(n)}\cong \Z$
by the $0$-reducedness assumption) and renaming the basis
elements appropriately.
\end{proof}

\heading{Proof of Proposition~\ref{p:bigprod}. }
We basically follow a proof for the
case of effective homology (where it is enough to
deal with two factors). There are two main steps,
encapsulated in the following two lemmas, which together imply
the proposition via Corollary~\ref{c:steq} (composing strong
equivalences).

\begin{lemma}[Tensor product of strong equivalences]\label{l:tens-steq}
Let $(C(I)_*:I\in\II)$ and $(\hC(I)_*:I\in\II)$ be a locally
polynomial-time chain complexes, let $(\EC(I)_*:I\in\II)$
be a globally polynomial-time chain complex, and suppose that
a strong equivalence $C_*\lreduP \hC_*\reduP \EC_*$ is given,
with all the chain complexes involved $0$-reduced.
As in Lemma~\ref{l:tenspr}, let $T_*$, $\hT_*$, $\ET_*$ be the
parameterized chain complexes of tensor products with factors
from $C_*$, $\hC_*$, and $\EC_*$, respectively. Then
$\ET_*$ is globally polynomial-time and there is
a strong equivalence $T_*\lreduP \hT_*\reduP \ET_*$.
\end{lemma}

\begin{lemma}[Eilenberg--Zilber for many factors]
\label{l:EZ} Let $(X(I):I\in\II)$ be a $0$-reduced
locally polynomial-time simplicial set, let $(W(J):J\in\JJ)$
be the parameterized set of products as in Proposition~\ref{p:bigprod},
and let $(T(J)_*:J\in\JJ)$ be the parameterized chain complex
of the tensor products $C_*(X^{(1)})\otimes\cdots\otimes
C_*(X^{(n)})$ as in Lemma~\ref{l:tenspr}.
Then there is a polynomial-time reduction $C_*(W)\reduP T_*$.
\end{lemma}

\begin{proof}[Proof of Lemma~\ref{l:tens-steq}]
We know from Lemma~\ref{l:tenspr}
that $T_*$, $\hT_*$, and $\ET_*$ are locally polynomial-time.
To check that $\ET_*$ is globally polynomial-time,
let us consider the chain group $\ET(J)_\thedim$,
$J=(I_1,\ldots,I_n)$. Since $\EC_*$ is globally polynomial-time,
there is a polynomial $p$ such that  $\rank(\EC(I_i)_j)
\le p(\size(I_i))\le p(\size(J))$ for all $J$ and all $j\le\thedim$.
Setting $N:=p(\size(J))$, by the $0$-reducedness assumption and the rank
formula (\ref{e:rank-tens}) we get
$\rank(\ET(J)_\thedim)\le {n+\thedim-1\choose\thedim}N^\thedim$,
which is bounded by a polynomial in $\size(J)\ge n$.
Generating the distinguished basis of $\ET(J)_\thedim$ in
polynomial time is done by a straightforward combinatorial enumeration
algorithm. We conclude that $\ET_*$ is globally polynomial-time.

It remains to provide a polynomial-time reduction $\hT_*\reduP T_*$
(then $\hT_*\reduP \ET_*$ is obtained in the same way).
We consider $\hT_*(J)=\hC^{(1)}\otimes\cdots\otimes
\hC^{(n)}$, $J=(I_1,\ldots,I_n)$,
$\hC_*^{(i)}=\hC_*(I_i)$, and let
$\rho^{(i)}=(F^{(i)},G^{(i)},H^{(i)})$ be the reduction
$\hC_*^{(i)}\redu C_*^{(i)}$ obtained from the assumption
$\hC_*\reduP C_*$  (we use capital
letters to avoid conflict with the notation of Lemma~\ref{l:comp-r0}).
The desired reduction $\hT_*(J)\redu T_*(J)$ goes through the intermediate
chain complexes
$$
\hC^{(1)}_*\otimes\cdots\otimes \hC^{(i-1)}\otimes
C_*^{(i)}\otimes \cdots\otimes C_*^{(n)}, \ \ i=1,\ldots,n,
$$
and the $i$th of these chain complexes is reduced to the $(i+1)$st one
with the reduction that is the tensor product with $\rho_i$ as the
$i$th factor and the identities in all the other factors.

Specializing the formulas from Lemma~\ref{l:comp-r0}
for composing reductions, we obtain the reduction $(F_J,G_J,H_J):
\hC_*(J)\redu C_*(J)$ with
$F_J=F^{(1)}\otimes\cdots\otimes F^{(n)}$, $G_J=G^{(1)}\otimes\cdots\otimes G^{(n)}$, and
\begin{eqnarray*}
H_J&=&H^{(1)}\otimes \id \otimes\cdots
\otimes\id + G^{(1)}F^{(1)}\otimes H^{(2)}\otimes\id\otimes\cdots\otimes\id+
\cdots \\
&& \ \ \ {}+ G^{(1)}F^{(1)}\otimes \cdots\otimes G^{(n-1)}F^{(n-1)}
\otimes H^{(n)}.
\end{eqnarray*}
(Tensor products of chain maps are defined as expected,
via $(f\otimes g)(a\otimes b)=f(a)\otimes g(b)$; for
chain homotopies there is a sign convention
involved, with the signs obviously polynomial-time computable---see,
e.g., \cite[Definition~57]{SergerGenova}.)

These formulas define the desired reduction $(F_J,G_J,H_J)_{J\in\JJ}:
\hT_*\reduP T_*$; polynomial-time computability
of these maps follows from Lemma~\ref{l:tenspr}.
\end{proof}

\begin{proof}[Proof of Lemma~\ref{l:EZ}]
For the binary case, with simplicial sets $Y$ and $Z$,
there is the classical Eilenberg--Zilber
reduction $C_*(Y\times Z)\redu C_*(Y)\otimes C_*(Z)$,
which is denoted by $(\AW,\EML,\SHI)$ (these are acronyms
for Alexander--Whitney, Eilenberg--MacLane, and Shih\footnote{The
explicit formula for the operator $\SHI$ was found
by Rubio \cite{Rubio-thesis} and proved by Morace---see the appendix in
\cite{real-assoc}.}).
Explicit formulas for these maps
are available; see \cite[pp.~1212--1213]{GonRea05} (for
$\AW$ and $\EML$ we also provide the formulas below).
In particular, it is clear from these formulas that
the maps $\AW$, $\EML$, $\SHI$ are polynomial-time
for locally polynomial-time $Y$ and~$Z$.

To build the reduction $C_*(W(J))\redu T_*(J)$,
where as usual $J=(I_1,\ldots,I_n)$, $W(J)=X^{(1)}\times
\cdots\times X^{(n)}$, and $T_*(J)=C_*(X^{(1)})\otimes\cdots
\otimes C_*(X^{(n)})$,
we go through the intermediate chain complexes
$$
D^{(i)}_*:=C_*(X^{(1)})\otimes\cdots\otimes C_*(X^{(i-1)})\otimes
C_*(X^{(i)}\times\cdots\times X^{(n)}).
$$
Let $(f^{(i)},g^{(i)},h^{(i)})$ be the reduction $D_*^{(i)}\redu D_*^{(i+1)}$.
We have
$f^{(i)}=\id\otimes\cdots\otimes\id\otimes \AW^{(i)}$,
$g^{(i)}=\id\otimes\cdots\otimes\id\otimes \EML^{(i)}$, and
$h^{(i)}=\id\otimes\cdots\otimes\id\otimes \SHI^{(i)}$,
where $(\AW^{(i)},\EML^{(i)},\SHI^{(i)})$ is the Eilenberg--Zilber reduction
$C_*(X^{(i)}\times Z^{(i)})\redu C_*(X^{(i)})\otimes
C_*(Z^{(i)})$, with $Z^{(i)}:=X^{(i+1)}\times\cdots\times X^{(n)}$.

Now $f^{(i)},g^{(i)},h^{(i)}$ are polynomial-time by Lemma~\ref{l:tenspr},
and so in order to verify the polynomiality of the composed
reduction, using the formula in Lemma~\ref{l:comp-r0},
it suffices to check polynomiality of the compositions
$f^{(i)}f^{(i-1)}\cdots f^{(1)}$ and
$g^{(1)}g^{(2)}\cdots g^{(i)}$, $i=1,2,\ldots,n-1$.
For simpler notation, we will discuss only the case $i=n-1$, but
the case of arbitrary $i$ is the same.

Let $(\sigma_i,\tau_i)$ be a $\thedim$-dimensional simplex of
$X^{(i)}\times Z^{(i)}$, which we also consider as a generator
of $C_*(X^{(i)}\times Z^{(i)})$. According to
\cite{GonRea05}, we have
$$
\AW^{(i)}(\sigma_i,\tau_i)=\sum_{j=0}^\thedim \partial_{j+1}\cdots
\partial_\thedim \sigma_i\otimes \partial_0\cdots\partial_{j-1}\tau_i.
$$
Composing $f^{(2)}$ and $f^{(1)}$ thus yields
\begin{eqnarray*}
f^{(2)}f^{(1)}(\sigma_1,\sigma_2,\tau_3)&=&
\sum_{0\le j_1+j_2\le\thedim} \Big(
\partial_{j_1+1}\cdots\partial_{\thedim}\sigma_1
\otimes \partial_0\cdots\partial_{j_1-1}\partial_{j_2+1}\cdots
\partial_{\thedim-j_1}
\sigma_2\\
&&\ \ \ \ \ \ \ \ \ \ \ \ \ \ \ \
{}\otimes \partial_0\cdots\partial_{j_1-1}\partial_0\cdots\partial_{j_2-1}\tau_3\Big).
\end{eqnarray*}
Continuing in a similar manner, we obtain
$f^{(n-1)}\cdots f^{(1)}(\sigma_1,\ldots,\sigma_n)$ as the sum
$$
\sum_{0\le j_1+\cdots+ j_{n-1}\le\thedim}\sigma'_1\otimes\cdots\otimes\sigma'_n,$$
where each $\sigma'_i$ is the result of applying some
number (at most $\thedim$)
of face operators to $\sigma_i$. The number of terms in this sum is
$n+\thedim-1\choose\thedim$, which is polynomially bounded for $\thedim$
fixed, and each term is polynomial-time computable. Thus,
the compositions $f^{(i)}\cdots f^{(1)}$ are polynomial-time computable.

Concerning the $g^{(i)}$'s, for the mapping $\EML^{(i)}$ we have, again
following \cite{GonRea05}, for a $p$-simplex $\sigma$ and a $q$-simplex
$\tau$, $p+q=\thedim$,
$$
\EML^{(i)}(\sigma\otimes\tau)=\sum_{\alpha,\beta:\alpha\cup\beta=\{0,1,\ldots,
\thedim-1\}\atop |\alpha|=q,|\beta|=p,\alpha\cap\beta=\emptyset} \pm (s_\alpha\sigma,s_\beta\tau),
$$
where, writing $\alpha=\{j_1,j_2,\ldots,j_q\}$, $j_1<j_2<\cdots<j_q$,
$s_\alpha$  denotes the composition $s_{j_q}s_{j_{q-1}}\cdots s_{j_1}$
of degeneracy operators, and similarly for $s_\beta$.
The sign $\pm$ depends on $\alpha$ and~$\beta$
in a simple way, and we do not want to bother the reader with specifying it
(see \cite{GonRea05}).

By iterating this formula, we find that,
for a $\thedim$-simplex $\sigma_1\otimes\cdots\otimes\sigma_n$, where
$\dim\sigma_i=\thedim_i$, $\thedim_1+\cdots+\thedim_n=\thedim$,
$$
g^{(1)}g^{(2)}\cdots g^{(n-1)}(\sigma_1\otimes\cdots\otimes\sigma_n)=
\sum_{\alpha_1,\ldots,\alpha_n} \pm (s_{\alpha_1}\sigma_1,s_{\alpha_2}\sigma_2,
\ldots,s_{\alpha_n}\sigma_n),
$$
where the sum is over certain choices of index sets
$\alpha_1,\ldots,\alpha_n\subseteq \{0,1,\ldots,k-1\}$.
We need not specify these choices precisely here (it suffices
to know that there is a polynomial-time algorithm for
generating them); we just note that
$|\alpha_i|=\thedim-\thedim_i$, since each of the simplices
$s_{\alpha_i}\sigma_i$ must have dimension $\thedim$.
Therefore, the number of terms in the sum is bounded above by
$$
\prod_{i=1}^n {\thedim\choose\thedim-\thedim_i}=
\prod_{i=1}^n {\thedim\choose\thedim_i}< 2^{\thedim^2},
$$
since there are at most $\thedim$ nonzero $\thedim_i$'s,
and ${\thedim\choose\thedim_i}<2^\thedim$ always. (A more refined
estimate gives a better bound, but still exponential in $\thedim$.)
So the number of terms depends only on $\thedim$, and thus it is a constant
in our setting.

This concludes the proof.
\end{proof}

\subsection{The basic perturbation lemma}\label{s:perturb}

The following situation often occurs in the theory of effective
homology. Suppose that we have already managed to obtain
a reduction $C_*\redu \tC_*$ for some chain
complexes $C_*$ and $\tC_*$. Now we want a reduction from
$C'_*$ to some $\tC'_*$, where $C'_*$ is a chain complex
that is ``similar'' to $C_*$, in the following way:
the chain groups of $C_*$ and of $C'_*$ are the same,
i.e., $C_\thedim=C'_\thedim$ for all $\thedim$, and the differential
$\diff'$ of $C'_*$ is of the form $\diff'=\diff+\delta$,
where $\diff$ is the differential in $C_*$, and
$\delta$ is a map that is ``small'' in a sense to be specified
in Theorem~\ref{t:BPL}
below. Thus, we regard $\diff'$ as a \emph{perturbation}
of~$\diff$.

In this setting, we would like to modify the differential
$\tilde\diff$ in $\tC_*$ to a suitable $\tilde\diff'$, obtaining
a new chain complex $\tC'_*$ and a reduction $C'_*\redu \tC'_*$.
If, for example, $\tC_*$ was globally polynomial-time,
and the original reduction $C_*\redu \tC_*$ provided
polynomial-time homology for $C_*$, we would like the
new reduction $C'_*\redu\tC'_*$ to give polynomial-time
homology for $C'_*$.

A tool for that is the \emph{basic perturbation lemma},
originally discovered by Shih.\footnote{Let us remark that
there are many variants, extensions, and generalizations
of the basic perturbation lemma in the literature, whose
usefulness is by far not restricted to an algorithmic context.}
For our purposes, we formulate a version of the basic
perturbation lemma which yields polynomial-time reductions.

To state it, we need a definition. Let $f\:C_*\to C_*$ be
a chain map of a chain complex into itself. We say that $f$
is \emph{nilpotent} if for every $c\in C_\thedim$, $k\in\Z$,
there is some $n$ such that $(f_\thedim)^{n}(c)=0$, where $(f_\thedim)^{n}$
is the $n$-fold composition of $f_\thedim$ with itself.
Now if $C_*$ is a parameterized chain complex,
we say that $f$ has \emph{constant nilpotency bounds} if
for every $\thedim$ there exists $N=N_\thedim$,
depending on $\thedim$ but not on the value of the parameter, such that
$(f_\thedim)^{N}$ is the zero map.

\begin{theorem}[Basic perturbation lemma]\label{t:BPL}
Let $(f,g,h)$ be a reduction $C_*\redu \tC_*$,  let
$C'_*$ be a chain complex with $C'_\thedim=C_\thedim$
for all $\thedim$ and with differential $\diff'$,
and let us set $\delta:=\diff'-\diff$. If the composed
map $h\delta$ is nilpotent, then there is a chain
complex $\tC'_*$ with the same chain groups as
$\tC_*$ and with a modified differential $\tilde\diff'$,
and a reduction $C'_*\redu \tC'_*$.

If $C_*$ and $\tC_*$ are locally polynomial-time
chain complexes, $(f,g,h)$ is a polynomial-time reduction, $\delta$
is  a polynomial-time map, and the composition $h\delta$ has constant
nilpotency bounds, then $\tilde\diff'$ is
polynomial-time and $C'_*\reduP \tC'_*$.
\end{theorem}

\begin{proof} The proof of the existence statement, presented,
e.g., in \cite[Theorem~50]{SergerGenova}, provides
explicit formulas for $\tilde\diff'$ and for
the desired reduction $(f',g',h'):C'_*\redu \tC'_*$.
Namely, using auxiliary chain maps $\varphi$ and $\psi$
defined by
\[
\varphi:=\sum_{i=0}^\infty (-1)^i(h\delta)^i,\ \ \ \
\psi:=\sum_{i=0}^\infty (-1)^i(\delta h)^i,
\]
we have $\tilde\diff':=\tilde\diff+f\psi\delta g$, $f':=f\psi$,
$g'=\varphi g$, and $h':=\varphi h$. If $h\delta$
has constant nilpotency bounds, then so has $\delta h$,
and for each fixed $\thedim$, the number of nonzero
term in the sums defining $\varphi(c)$ and
$\psi(c)$, with $c\in C_\thedim$,
is bounded by a constant depending only on $\thedim$ but not on~$c$.
The claim about polynomiality follows.
\end{proof}

The basic perturbation lemma propagates the perturbation of the differential
in the direction of the reduction arrow. If we have a strong equivalence
$C_*\lredu \dtC_*\redu \tC_*$ and we want to perturb the differential
of $C_*$, we first need to propagate the perturbation to $\dtC_*$,
i.e., against the direction of the reduction. The next lemma
tells us that this can always be done; actually, only the differential
in $\dtC_*$ needs to be modified, the reduction stays the same.
We omit the easy proof---see \cite[Proposition~49]{SergerGenova}.

\begin{lemma}
\label{l:EPL}
Let $(f,g,h): C_*\redu \tC_*$ be a reduction,
and let $\tC'_*$ be obtained from $\tC_*$
by perturbing the original differential $\tilde\diff$
to $\tilde\diff'=\tilde\diff+\tilde\delta$.
Then $(f,g,h)$ is a reduction $C'_*\redu \tC'_*$,
where $C'_*$ is obtained from $C_*$ by perturbing
the original differential $\diff$ to
$\diff':=\diff+g\tilde\delta f$.
\end{lemma}

Thus, under favorable circumstances, if a parameterized chain complex
$C_*$ is equipped with polynomial-time homology, the
combination of the basic perturbation lemma and Lemma~\ref{l:EPL}
allows us to obtain polynomial-time homology for the
perturbed chain complex~$C'_*$.

%

\subsection{Mapping cone}\label{s:mcon}

Here we consider the mapping cone operation
for chain complexes, as introduced in Section~\ref{s:ssetchain}.

\begin{prop}[Algebraic mapping cone]\label{p:mcone}
If $C_*,\tC_*$ are (parameterized) chain complexes with polynomial-time homology
and $\varphi\:C_*\to \tC_*$ is a polynomial-time chain map,
then the cone  $\MCone_*(\varphi)$
can be equipped with polynomial-time homology.
\end{prop}

\begin{proof}[Proof (sketch)]
This is essentially \cite[Theorem~79]{SergerGenova}.
We sketch the proof since it is a simple and
instructive use of the perturbation lemma.

Given strong equivalences
$C_*\steqP \EC_*$ and $\tC_*\steqP \EtC_*$,
we want to construct a polynomial-time strong
equivalence of $\MCone_*(\varphi)$ with a suitable
globally polynomial-time chain complex~$\EM_*$.

We observe that, by definition, the chain groups
of $\MCone_*(\varphi)$ depend only on $C_*,\tC_*$ but not on $\varphi$
(only the differential depends on $\varphi$). We thus first
consider $\MCone_*(0_{C_*\to \tC_*})$, where $0_{C_*\to \tC_*}$
is the zero chain map of the indicated chain complexes.
Given the strong equivalences for $C_*$ and $\tC_*$ as
above, it is straightforward to construct a strong equivalence
\[
\MCone_*(0_{C_*\to \tC_*})\steqP \MCone_*(0_{\EC_*\to\EtC_*});
\]
this is just a direct sum construction.

Next, we regard $\MCone_*(\varphi)$
as a perturbation of $\MCone_*(0_{C_*\to \tC_*})$.
Then we propagate the perturbation through the strong equivalence;
in the application of the basic perturbation lemma,
it turns out that the nilpotency of the relevant maps
is bounded by~$2$ (independent of~$\thedim$). We refer to
\cite[Theorems~61,79]{SergerGenova} for details.
\end{proof}

We remark that the strong equivalence $\MCone_*(\varphi)\steqP\EM_*$
produced in the proposition restricts to the original strong equivalence
$\tC_*\steqP \EtC_*$. This follows at once from the explicit formulas in
the basic perturbation lemma and Lemma~\ref{l:EPL} and the fact that the involved perturbation is zero on~$\tC_*$.

\subsection{Twisted product}\label{s:twist}

\heading{On fiber bundles. } Our main goal is the computation
of a Postnikov system for a given space $Y$. As we have mentioned,
the $\thedim$th stage of a Postnikov system can be thought of as
an approximation of $Y$, in a homotopy-theoretic sense, made of
simple building blocks, which are called Eilenberg--MacLane spaces.
These building blocks will be discussed in Section~\ref{s:EML1} below,
but here we will consider the operation used to paste the building
blocks together.

To convey some intuition, we begin with the topological notion of
\emph{fiber bundle}\footnote{In the literature on simplicial sets,
effective homology and such, one usually speaks about a \emph{fibration},
which is a notion more general than a fiber bundle; roughly
speaking, a fibration can be regarded as a ``fiber bundle up to homotopy.''
}
(a \emph{vector bundle} is a special
case of a fiber bundle). Let $B$, the \emph{base space}, and
$F$, the \emph{fiber space}, be two spaces. The Cartesian product $F\times B$
can be thought of as a copy of $F$ sitting above each point of $B$;
for $B$ the unit circle $S^1$ and $F$ a segment this is indicated
in the left picture:
\immfig{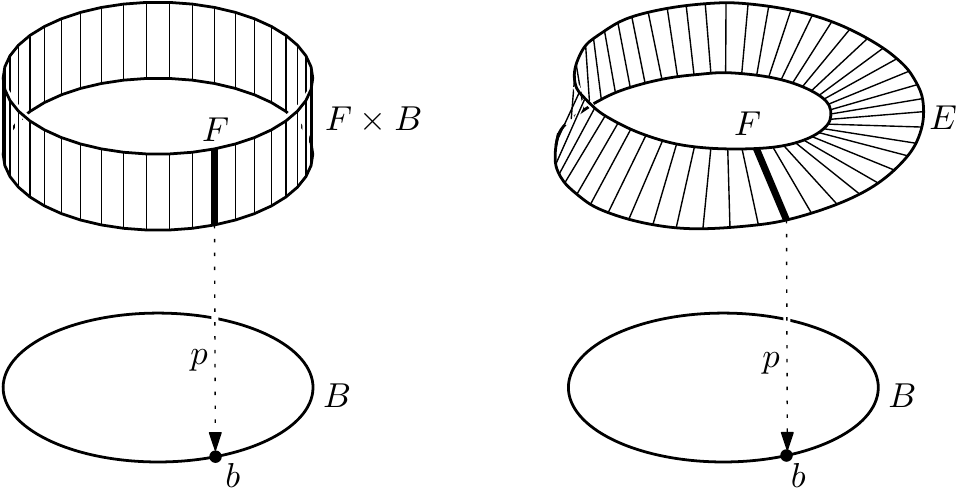} 
The product $F\times B$ is a \emph{trivial} fiber bundle, while
the right picture shows a nontrivial fiber bundle (a M\"obius band
in this case). Above every point $b\in B$, we still have a copy
of $F$, and moreover, each such $b$ has a small neighborhood
$U$ such that the union of all fibers sitting above $U$
is homeomorphic to the product $F\times U$, a rectangle in the picture.
However, globally, the union of the fibers above all of $B$
forms a space $E$, the \emph{total space} of the fiber bundle,
that is in general different from $F\times B$.

More precisely, a fiber bundle is given as $p\: E\to B$,
where $E,B$ are spaces and $p$ is a surjective map,
such that for every $b\in B$ there are a neighborhood $U$ of $b$ and
a homeomorphism  $h\:p^{-1}(U)\to F\times U$ fixing
the second component, i.e.,  with $h(x)_2=p(x)$ for every $x\in E$.
(Other famous examples of nontrivial fiber bundles involve the
the Klein bottle with $B=F=S^1$ or the Hopf fibration $S^3\to S^2$.)

For our purposes, we will deal with fiber bundles where the fiber $F$
has ``enough symmetries,'' meaning that there is a group $G$
acting on the fiber $F$, and this helps in specifying the
total space $E$ in terms of $B$, $F$, and some additional data
which, informally speaking, tell us how $E$ is ``twisted'' compared
to the product $F\times B$.

\heading{Simplicial groups. } In order to define the appropriate
simplicial notions, we first need to recall that a \emph{simplicial
group} is a simplicial set $G$ such that, for each $\thedim\ge 0$,
the set $G_\thedim$ of $\thedim$-dimensional simplices forms
a group, and moreover, the face and degeneracy operators are
group homomorphisms.

A \emph{parameterized simplicial group} and a \emph{locally polynomial-time
simplicial group} are defined in an obvious analogy with the
corresponding notions for simplicial sets and chain complexes.

A basic example of a simplicial group is the standard simplicial
model of an Eilenberg--MacLane space; see Section~\ref{s:EML1} below.
Actually, it is known that every Abelian simplicial group
is homotopy equivalent to a product of  Eilenberg--MacLane spaces
(see \cite[Chap.~V]{May:SimplicialObjects-1992}), and we will be interested only
in the Abelian case. Every simplicial group $G$ is a Kan simplicial
set \cite[Theorem~17.1]{May:SimplicialObjects-1992}, and so continuous maps into $|G|$
have a simplicial representation up to homotopy.

\heading{A simplicial setting: twisted products. }
For our purposes, we will deal with fiber bundles where $F$, $B$, and $E$
are simplicial sets, and a simplicial group $G$ acts (simplicially) on~$F$.
The corresponding simplicial notion is called a
\emph{twisted Cartesian product} (a more general simplicial
notion, a counterpart of a fibration, is a \emph{Kan fibration};
see, e.g., \cite[Chap.~I,II]{May:SimplicialObjects-1992}).

\begin{defn}[Twisted Cartesian product]\label{d:twistprod}
Let $B$ and $F$ be simplicial sets, and let an action of a simplicial
group $G$ on $F$ be given, i.e., a simplicial map $F\times G\to F$
satisfying the usual conditions for a (right) action of a group on a set;
that is, $\phi(\gamma\gamma')=(\phi\gamma)\gamma'$
and $\phi e_\thedim=\phi$ ($\phi\in F_\thedim$, $\gamma,\gamma'\in G_\thedim$,
$e_\thedim$ the unit element of $G_\thedim$). Moreover, let
$\tau=(\tau_\thedim)_{\thedim=1}^\infty$ be a \emph{twisting operator},
where $\tau_\thedim\:B_\thedim\to G_{\thedim-1}$ are mappings satisfying
the following conditions (we omit the dimension indices for simplicity):
\begin{enumerate}
\setlength{\itemsep}{1pt}
\setlength{\parskip}{0pt}
\setlength{\parsep}{0pt}
\item[\rm(i)] $\partial_0\tau(\beta)=\tau(\partial_1\beta)\tau(\partial_0\beta)^{-1}$;
\item[\rm(ii)] $\partial_i\tau(\beta)=\tau(\partial_{i+1}\beta)$
for $i\ge 1$;
\item[\rm(iii)] $s_i\tau(\beta)=\tau(s_{i+1}\beta)$ for all $i$; and
\item[\rm(iv)] $\tau(s_0\beta)=e_{\thedim}$ for all $\beta\in B_\thedim$,
where $e_\thedim$ is the unit element of $G_\thedim$.
\end{enumerate}

Then the \emph{twisted Cartesian product} $F\times_\tau B$ is a simplicial set
$E$ with $E_\thedim=F_\thedim\times B_\thedim$, i.e., the $\thedim$-simplices
are as in the Cartesian product $F\times B$, and the face and degeneracy
operators are also as in the Cartesian product (see Section~\ref{s:prod}),
with the sole exception of $\partial_0$, which is given by
\[
\partial_0(\phi,\beta):=(\partial_0(\phi)\tau(\beta),\partial_0 \beta),
\ \ \ \ (\phi,\beta)\in F_\thedim\times B_\thedim.
\]

A twisted Cartesian product $F\times_\tau B$ is called \emph{principal}
if $F=G$ and the considered right action of $G$ on itself is
by (right) multiplication.
\end{defn}

Thus, the only way in which $F\times_\tau B$ differs from the ordinary
Cartesian product $F\times B$ is in the $0$th face operator.
It is definitely not easy to see why this should be the right way of
representing fiber bundles simplicially, but for us, it is only important
that it works, and we will have explicit formulas available for the
twisting operator for all the specific applications. Actually,
we will use solely principal twisted Cartesian products.

Let  $F,B$ be locally polynomial-time
simplicial sets, let $G$ be a locally polynomial-time simplicial
group, and let the action of $G$ on $F$ and the twisting operator $\tau$ be
polynomial-time maps (again in a sense precisely analogous
to polynomial-time simplicial maps or chain maps); we assume that
all of these objects are parameterized by the same parameter set $\II$.
It is easy to see that then
the simplicial set $F\times_\tau B$, again parameterized
by $\II$, is locally polynomial-time.

We will need that under certain reducedness assumptions,
twisted products preserve polynomial-time homology.

\begin{prop}[Twisted product]\label{p:twistprod} Let $F$ and $B$ be
simplicial sets with polynomial-time homology, let $G$ be a
locally polynomial-time simplicial group with a
polynomial-time simplicial action on $F$, and let $\tau$ be a
polynomial-time twisting operator. Moreover,  suppose that $G$ is
\emph{$0$-reduced} (a single vertex) or that $B$ is
\emph{$1$-reduced} (a single vertex, no edges). Then $E:=F\times_\tau
B$ can be equipped with polynomial-time homology.
\end{prop}

The effective-homology analogs of this result are due to Rubio and
Sergeraert \cite[Theorem~132]{SergerGenova} when $B$ is $1$-reduced and
due to Filakovsk\'y \cite[Corollary~12]{Filakovsky-tensor} when $G$ is $0$-reduced.

\begin{proof}[Proof (sketch)] Let the polynomial-time homology
of $F$ and $B$ be given by strong equivalences $C_*(F)\steqP \EF_*$
and $C_*(B)\steqP \EB_*$, respectively.

We begin with the ordinary Cartesian product $F\times B$.
By the Eilenberg--Zilber theorem (Lemma~\ref{l:EZ} for two factors,
where we do not need to assume $0$-reducedness), there
is a reduction $(\AW,\EML,\SHI):C_*(F\times B)\reduP T_*$, where $T_*$
is the tensor product $C_*(F)\otimes C_*(B)$. Further, by
Lemma~\ref{l:tens-steq} for two factors, we have
$T_*\steqP \ET_*:=\EF_*\otimes \EB_*$. So altogether
\begin{equation}\label{e:redchain}
C_*(F\times B)\reduP T_* \steqP \ET_*.
\end{equation}

Next, by the definition of the twisted product,
the chain complex $C_*(F\times_\tau B)$ has the same chain groups as
$C_*(F\times B)$, but the differential is modified. Writing $\delta$
for the difference of the two differentials, on elements $(\phi,\beta)$
the standard basis of $C_\thedim(F\times B)$ we get
$\delta(\phi,\beta)= (\partial_0(\phi)\tau(\beta),\partial_0\beta)-(\partial_0 \phi,\partial_0\beta)$.

We recall that in any simplicial set $X$, every simplex $\sigma$
can be obtained from a unique \emph{nondegenerate} simplex $\tau$
by an application of degeneracy operators. Let us refer to the dimension
of $\tau$ as the \emph{geometric dimension} of $\sigma$.
Given a simplex $(\phi,\beta)$ of $F\times B$, its \emph{filtration degree}
is defined as the geometric  dimension of~$\beta$.

In the present proof, the filtration
degree serves as a potential function for controlling nilpotency
of the appropriate maps. First, it can be checked that the
chain homotopy $\SHI$ does not increase the filtration degree, and
a simple argument shows that $\delta$ decreases it at least by~$1$
(see, e.g., \cite[Theorem~130]{SergerGenova}, for details).
It follows that the composition $\SHI\circ \delta$
has constant nilpotency bounds, namely, $N_\thedim=\thedim+1$.
Therefore, the basic perturbation
lemma (Theorem~\ref{t:BPL}) shows that $C_*(F\times_\tau B)\reduP
T'_*$, where $T'_*$ is a perturbation of the tensor product complex~$T_*$.

Next, we would like to propagate the perturbation from $T_*$
through the next strong equivalence in (\ref{e:redchain}),
which we write more explicitly as
\[
T_*\lreduP \hat T_* \reduP \ET_*.
\]
Let $\delta^T$ be the difference of the differential in $T'_*$
and in $T_*$. By Lemma~\ref{l:EPL},
we get a perturbed version $\hat T'_*$ of the middle complex
$\hat T_*$, and the difference of its differential minus
the differential of $\hat T_*$ is $\hat\delta^T=g\delta^T f$,
for some chain maps $f,g$ from the reduction $T_*\lreduP \hat T_*$.

We now recall from the proof of Lemma~\ref{l:tens-steq}
that the chain complex $\hat T_*$ is constructed
as a tensor product of two chain complexes,
and that the chain homotopy $h$ in the reduction
$\hat T_* \reduP \ET_*$ has the form
\begin{equation}\label{e:thehat}
h =h^{(1)}\otimes \id + g^{(1)}f^{(1)}\otimes h^{(2)},
\end{equation}
for some chain
maps $f^{(1)},g^{(1)}$ and chain homotopies $h^{(1)},h^{(2)}$.

In order to apply the basic perturbation lemma to the just mentioned
reduction $\hat T_* \reduP \ET_*$, we  need to show that
$h\hat\delta^T$  has constant nilpotency bounds for every chain
homotopy~$h$ of the form (\ref{e:thehat}). This follows from the
obvious fact that such a chain homotopy never increases the filtration
degree\footnote{For a basis element $\hat a\otimes \hat b$ of the
tensor product $\hat T$, the filtration degree is defined simply as the
degree of~$\hat b$.} by more than~$1$, plus a result showing
that if $G$ is $0$-reduced or $B$ is $1$-reduced, then
$\hat\delta^T$ \emph{decreases} the filtration degree at least by~$2$.
We refer to \cite[Corollary~9 and 11]{Filakovsky-tensor} for a proof of
the latter result (also see the proof of Lemma~\ref{l:diff-stays} below,
where a very similar situation is discussed). Then a constant nilpotency bound
with $N_\thedim\le k+1$ follows, and the proposition is proved.
\end{proof}

\subsection{The bar construction}

The bar construction, originating in Eilenberg and Mac~Lane
\cite{EilenbergMacLane:GroupsHPin1-1953}, is an algebraic construction
with many uses and generalizations. For us, it provides a way
of constructing auxiliary chain complexes for certain reductions
and strong equivalences; we will thus introduce it only in the setting
of chain complexes. The definition below is somewhat complicated,
but most of the details will be irrelevant in the sequel---the important
properties will be encapsulated in a couple of lemmas below.
We essentially follow \cite[Chap.~3]{RealThesis}, with some minor
technical differences.

A \emph{differential graded algebra}
is a chain complex $A_*$ together with an associative
multiplication $A_*\otimes A_*\ra A_*$ with a unit $1_{A_*}$.
We denote the image of  $a\otimes b$ simply by $a\cdot b$.
This multiplication is assumed to be a chain map;
in particular, for $a\in A_\thedim$ and $b\in A_\otherdim$
we have $a\cdot b\in A_{\thedim+\otherdim}$.
The chain map condition on the multiplication reads
\[\diff(a\cdot b)=\diff (a)\cdot b+(-1)^{\deg a}a\cdot \diff (b)\]
(the Leibniz rule).
The unit $1_{A_*}$ is necessarily of degree 0.

We say that $A_*$ is 0-reduced if $A_0=\Z$, generated by $1_{A_*}$.
Regarding $\Z$ as a chain complex whose all chain groups are zero
except for the one in dimension 0, which is $\Z$, there is
a unique homomorphism $\varepsilon\:A_*\ra\Z$ of
differential graded algebras (i.e., a chain map preserving
the unit and the multiplication).\footnote{In detail
$\varepsilon(n\cdot 1_{A_*})=n$ and, for $a$ of positive dimension,
$\varepsilon(a)=0$.} We call $\varepsilon$ the \emph{augmentation}.
Its kernel, the \emph{augmentation ideal}, is denoted by $\overline A_*$.

Further, we denote by $\overline A_*^\uparrow$ the
shift of $\overline A_*$ upwards by one, so that we have
\[\overline A^\uparrow_0=\overline A^\uparrow_1=0\textrm{{}, and }\overline A^\uparrow_\thedim=A_{\thedim-1}, \thedim\ge 2.\]
The shifted chain complex comes with the shifted differential
$\diff^{\overline A_*^\uparrow}(a)=-\diff^{\overline A_*}(a)=
-\diff^{A_*}(a)$.

A right differential graded  $A_*$-module
 is a chain complex $M_*$ equipped with a chain map
\[M_*\otimes A_*\ra M_*\]
that satisfies the usual axioms for a module structure. Again the action
being a chain map translates into a Leibniz-type rule
for the compatibility of the multiplication and the differential. Similarly,
a left  $A_*$-module $N_*$ is equipped with an action $A_*\otimes N_*\ra N_*$.

Given $A_*,M_*,N_*$ as above, the bar construction produces
a chain complex $\Barr^{A_*}(M_*,N_*)$.
In order to define it,
we first form an auxiliary chain complex given by
\[
T_*:=\bigoplus_{n=0}^\infty M_*\otimes
 (\overline A^\uparrow_*)^{\otimes n} \otimes N_*.
\]
We denote the differential in $T_*$ by $\diff^T$ and call it the
\emph{tensorial differential}. The actual bar construction
will be given by a perturbation of this differential.

Assuming that each of the chain groups in $A_*,M_*,N_*$
has a distinguished basis,  the distinguished
bases in $T_*$ are made of elements of the form
\[
z := x\otimes a_1\otimes\cdots\otimes a_n\otimes y,
\]
where $x$ comes from a distinguished basis in $M_*$,
$y$ from one in $N_*$, and $a_1,\ldots,a_n\ne 1_A$ from those in $A_*$.
(Here we can also explain the origin of the name ``bar construction'';
in the Eilenberg--Mac~Lane founding paper, the tensor
product signs $\otimes$ in the above notation for $z$
were abbreviated to vertical bars.) The tensorial differential
$\diff^T(z)$ is given by the (iterated) formula~(\ref{e:difffmla})
from Section~\ref{s:prod}.

The degree of such a $z$ equals $\deg(z)=\deg_{\rm tens}(z)+\deg_{\rm res}(z)$,
where $\deg_{\rm tens}(z)$, the \emph{tensorial degree} of $z$,
equals $\deg(x)+\deg(y)+\sum_{i=1}^n\deg(a_i)$ (with
$\deg(a_i)$ being the degree of $a_i$ in $A_*$),
and the \emph{residual degree} $\deg_{\rm res}(z)=n$.

Now the chain complex $\Barr^{A_*}(M_*,N_*)$ has the same
chain groups as $T_*$, but the differential is modified
to $\diff^T+\delta^{\rm ext}$, where $\delta^{\rm ext}$, the \emph{external
differential}, is given by
\begin{eqnarray*}
\delta^{\rm ext}(x\otimes a_1\otimes\cdots\otimes a_n\otimes y)&:=&
(-1)^{m_0} x\cdot a_1\otimes a_2\otimes\cdots\otimes a_n\otimes y\\
&& {}+\sum_{i=1}^{n-1}(-1)^{m_i} x\otimes a_1\otimes\cdots\otimes a_i\cdot a_{i+1}\otimes\cdots\otimes a_n\otimes y\\
&& {}+(-1)^{m_n} x\otimes a_1\otimes\cdots\otimes a_{n-1}\otimes
a_n\cdot y,
\end{eqnarray*}
where $m_i=\deg(x)+\deg(a_1)+\cdots+\deg(a_i)+i$. We note that
the external differential is the only part of the definition of
$\Barr^{A_*}(M_*,N_*)$ where the algebra and module structures
play a role. This finishes the definition of the bar construction.

In our applications, the bar construction will be used with
$M_*$ equal to $\Z$. Here we endow $\Z$ with the right $A_*$-module
structure obtained from the augmentation---the unit $1_{A_*}$
acts by identity as it must and the elements from the augmentation
ideal act trivially, i.e., $a\cdot x=0$.
We also note that $\Z$ acts as a unit element
for tensor product, in the sense that $C_*\otimes \Z$
and $\Z\otimes C_*$ can be canonically identified with $C_*$
(this is obvious by considering the distinguished bases, for example).

\begin{lemma}[Polynomial-time homology for the bar construction]
\label{l:bar-poly}
Let $A_*,M_*,N_*$ be locally polynomial-time versions
of the objects above,
with all the multiplications
involved being polynomial-time maps,
and let us suppose that $A_*,M_*,N_*$
are equipped with polynomial-time homology.
Then $\Barr^{A_*}(M_*,N_*)$ can
be equipped with polynomial-time homology.
\end{lemma}

\begin{proof} First we equip $T_*$ with polynomial-time
homology; this is essentially Lemma~\ref{l:tens-steq} about tensor products
of strong equivalences. 
The factors $M_*$ and $N_*$ are not 0-reduced but this can be accommodated,
in a way similar to Cartesian products---see the remark
following Proposition~\ref{p:bigprod}.
We also note that although $T_*$ is an infinite direct sum,
the $\thedim$th chain group involves only elements with $n\le \thedim$
from this direct sum, and so Lemma~\ref{l:tens-steq} is applicable.

Next, we apply the basic perturbation lemma and Lemma~\ref{l:EPL}, in a way
very similar to the proof of Proposition~\ref{p:twistprod} on twisted
products, to propagate the perturbation of the differential in $T_*$
by the external differential $\delta^{\rm ext}$ through the strong equivalence.
The only issue is to show constant nilpotency bounds. Here one uses that
the chain homotopy involved, which is of the form (\ref{e:thehat})
but with an arbitrary number of factors, does not increase the residual degree
$\deg_{\rm res}$, while $\delta^{\rm ext}$, obviously, decreases it by~$1$.
\end{proof}

The next lemma is a key property of the bar construction, showing
that it provides, in a sense, an ``inverse'' to the operation of tensor
product with $A_*$. Indeed, the bar construction $\Barr^{A_*}(\Z,A_*)$
can be regarded as a formal analog of the power series
expression $1=\frac a{a}=a+(1-a)a+(1-a)^2a+\cdots$
for a real number $a\in(0,2)$.

\begin{lemma}\label{l:bar-inverse}
Given a locally polynomial-time $0$-reduced differential graded algebra $A_*$,
there is a reduction
\[\Barr^{A_*}(\Z,A_*)\reduP \Z
\]
(where $A_*$ is taken
as a differential graded
 $A_*$-module in the obvious way). More generally, if we consider,
in addition, a locally polynomial-time
chain complex $M_*$ and turn $A_*\otimes M_*$ into
a left $A_*$-module by defining $a\cdot(b\otimes x):=(a\cdot b)\otimes x$, then we obtain
a reduction $\Barr^{A_*}(\Z,A_*\otimes M_*)\reduP M_*$.
\end{lemma}

We note that we assume no $A_*$-module structure on $M_*$;
the left $A_*$-module structure on $A_*\otimes M_*$
comes from the multiplication in~$A_*$.

\begin{proof} In the reduction $(f,g,h)\: \Barr^{A_*}(\Z,A_*)\reduP \Z$,
$f$ and $g$ are given by the assumed identification of $A_0$ with $\Z$
(note that the $0$th chain group of $\Barr^{A_*}(\Z,A_*)$ can be
canonically identified with $A_0$); in particular,
we have $f(a_1\otimes\cdots\otimes a_n\otimes a)=0$ unless
$n=0$.

In residual degree 0 we have $f(a)=\varepsilon(a)$. Denote by
$\overline a=a-\varepsilon(a)\cdot 1_A$ the projection of $a$
onto the augmentation ideal $\overline A_*$. Then, for a basis element
$z=a_1\otimes\cdots\otimes a_n\otimes a$ of $\Barr^{A_*}(\Z,A_*)$,
we put
\[h(z):=(-1)^{\deg(a_1)+\cdots+\deg(a_n)+\deg(a)+n+1}a_1\otimes\cdots\otimes a_n\otimes\overline{a}\otimes 1_{A_*}.\]
It is simple to check that we indeed get a reduction (see \cite{Moore59}),
and polynomiality is obvious.

The more general reduction $\Barr^{A_*}(\Z,A_*\otimes M_*)\reduP M_*$
is then immediately obtained from the previous one by tensoring all the
maps with the identity on~$M_*$.
\end{proof}

\subsection{The base space (a ``twisted division'')}

Here, as in Section~\ref{s:twist}, $G$ is an Abelian simplicial
group, and we consider a twisted product, this time
a principal one: $G\times_\tau B$. However, while previously
we took $G,B,\tau$ as known, and wanted to compute $G\times_\tau B$
(so we did ``twisted multiplication''), here we assume that
$G$ \emph{and} $G\times_\tau B$ are known, and we want $B$---so
one can think of this as ``twisted division''. The bar construction
is the main tool.

\begin{prop}\label{p:tw-div}
Let $G$ be a $0$-reduced
locally polynomial-time Abelian simplicial group,
let $B$ be a locally polynomial-time
simplicial set, and let $\tau$ be a polynomial-time twisting
operator.  If both $G$ and $G\times_\tau B$ are equipped with
polynomial-time homology, then $B$ can also be
equipped with polynomial-time homology.
\end{prop}

\begin{proof}
We follow the treatment in Real~\cite{RealThesis}.
We let $A_*:=C_*(G)$ be the normalized chain complex
of $G$. The \emph{Eilenberg--MacLane product} on $A_*$
is defined using the operator $\EML\:A_*\otimes A_*\to C_*(G\times G)$
as in the proof of Lemma~\ref{l:EZ}. Writing
$\EML(a\otimes b)=\sum_{i=1}^n \alpha_i(\gamma_i,\gamma'_i)$,
$\gamma_1,\ldots,\gamma'_n\in G$, we set
\[
a\cdot b:= \sum_{i=1}^n\alpha_i \gamma_i\gamma'_i,
\]
where $\gamma_i\gamma'_i$ is computed using the group operation in~$G$.
This multiplication
is polynomial-time, and with some work it
can be checked that it makes $A_*$ into a differential graded algebra.

\heading{The untwisted case. } First we assume that the
ordinary Cartesian product $G\times B$ is given with polynomial-time
homology. Then polynomial-time homology for $B$ is obtained
in the following steps:
\begin{enumerate}
\item $C_*(G\times B)$ has polynomial-time homology by the assumption.
\item The Eilenberg--Zilber reduction $C_*(G\times B)\reduP A_*\otimes
C_*(B)$
(Lemma~\ref{l:EZ}) and the composition of strong equivalences
yield polynomial-time homology for $ A_*\otimes C_*(B)$.
\item Since $A_*$ has polynomial-time
homology as well by assumption,
Lemma~\ref{l:bar-poly} yields polynomial-time homology
for  $\Barr^{A_*}(\Z,A_*\otimes C_*(B))$.
\item Finally,  the reduction
$\Barr^{A_*}(\Z,A_*\otimes C_*(B))\reduP C_*(B)$
from Lemma~\ref{l:bar-inverse} and composition of strong equivalences
provide polynomial-time homology for~$C_*(B)$.
\end{enumerate}

\heading{The twisting. } Now we present ``twisted analogs''
of  steps 1--4 above.
\begin{enumerate}
\item[$1_\tau$.]
We assume that polynomial-time
homology is available for the \emph{twisted} Cartesian product
$G\times_\tau B$.

\item [$2_\tau$.]
As in the proof of Proposition~\ref{p:twistprod}
(twisted product), applying the basic perturbation lemma
to the Eilenberg--Zilber reduction $C_*(G\times B)\reduP
Q_*:=A_*\otimes C_*(B)$ provides a reduction $C_*(G\times_\tau B)\reduP Q'_*$,
where $Q'_*$ is obtained by perturbing the differential $\diff^Q$
of the tensor product complex $Q_*$ to another differential
$\diff^{Q'}$. Let $\delta^Q:=\diff^{Q'}-\diff^Q$ be the difference.
On $Q'_*$ the multiplication by $A_*$ from the left is defined
in the same way as on $Q_*$. Using formula (\ref{e:twistcoch}) below,
one can prove that the perturbation $\delta^Q$ is $A_*$-linear.
It means that $\diff^{Q'}$ satisfies the Leibniz rule and hence $Q'_*$ is
a left $A_*$-module.

\item[$3_\tau$.]
We have $\diff^{Q'}$ polynomial-time computable (since the basic
perturbation lemma provides an explicit formula), and hence
we obtain  polynomial-time homology for $\Barr^{A_*}(\Z,Q'_*)$
by Lemma~\ref{l:bar-poly}.
\item[$4_\tau$.]
It remains to exhibit a reduction $\Barr^{A_*}(\Z,Q'_*)\reduP C_*(B)$;
then we obtain polynomial-time homology for $B$ as in the untwisted case above.
We begin with the reduction $\Barr^{A_*}(\Z,Q_*)\reduP C_*(B)$
from Lemma~\ref{l:bar-inverse} and apply the basic perturbation
lemma to it.

We note that, by the definition of the bar construction,
$\Barr^{A_*}(\Z,Q_*)$ and $\Barr^{A_*}(\Z,Q'_*)$ have the same chain
groups, and only the differential is modified.
Let $\delta^{\rm Bar}$ be the differential of $\Barr^{A_*}(\Z,Q'_*)$
minus the one of $\Barr^{A_*}(\Z,Q_*)$. We observe that the
external differentials in these bar constructions coincide,
and the tensorial differentials differ only in one term.
Thus, writing a basis element of $\Barr^{A_*}(\Z,Q_*)$
as $z=a_1\otimes\cdots\otimes a_n\otimes(a\otimes b)$,
we have
\[
\delta^{\rm Bar}z=(-1)^{\deg(a_1)+\cdots+\deg(a_n)+\deg(a)-n}a_1\otimes\cdots\otimes a_n\otimes\delta^{Q}(a\otimes b).
\]
The rest of the proof is delegated to the next lemma, which
is essentially Prop.~3.2.3 in~\cite{RealThesis}.
\end{enumerate}
\end{proof}

\begin{lemma}\label{l:diff-stays}
If $G$ is a $0$-reduced simplicial group,
$A_*=C_*(G)$ and $Q_*=A_*\otimes C_*(B)$
are as above, $(f,g,h)\:\Barr^{A_*}(\Z,Q_*)\redu
C_*(B)$ is the reduction from Lemma~\ref{l:bar-inverse}, and
$\delta^{\rm Bar}$ is the perturbation of the differential of
$\Barr^{A_*}(\Z,Q_*)$ as above, then $h\delta^{\rm Bar}$ has constant
nilpotency bounds, and the perturbed differential in $C_*(B)$ obtained
from the application of the basic perturbation lemma
to the reduction $(f,g,h)$ actually
equals the original differential in $C_*(B)$, i.e., the resulting
perturbation is zero.
\end{lemma}

\begin{proof} There is an explicit expression known for the
perturbation $\delta^Q$, going back to Brown~\cite{BrowThm}
and Shih~\cite{Shih}. We
do not need the full explicit formula, just some of its
properties.

Namely, given $G$, $B$, and the twisting
operator $\tau$, there is a sequence of homomorphisms
$t_\thedim\:C_\thedim(B)\to
C_{\thedim-1}(G)$, such that for $a\in C_\otherdim(G)$,
$b\in C_\thedim(B)$, we have
\begin{equation}\label{e:twistcoch}
\delta^Q(a\otimes b)=
\sum_{i=0}^\thedim (-1)^{\otherdim} a\cdot t_{\thedim-i}(b_{\thedim-i})\otimes \tilde b_i,
\end{equation}
for some chains $b_0,\ldots,b_\thedim,\tilde b_0,\ldots,\tilde b_\thedim$, with $b_i,\tilde b_i\in C_i(B)$,
the multiplication in $a\cdot t_{\thedim-i}(b_{\thedim-i})$
being the Eilenberg--MacLane product introduced above.\footnote{In the
literature, $t$ is called a \emph{twisting cochain},
and $\delta^Q(a\otimes b)$ is written as a cap product $t\cap(a\otimes b)$.
Moreover, $t$ is in general not determined uniquely by $G,B,\tau$,
since the operator $\AW$ in the reduction $C_*(G\times B)\redu
C_*(G)\otimes C_*(B)$ is not unique. However, the relevant sources
use the same particular $\AW$ as we do.
}

Now $t_0=0$ since $C_{-1}(G)=0$. Moreover, one can compute
(see the proof of \cite[Corollary 11]{Filakovsky-tensor})
 that  $t(b_1)=\tau(b_1)-e_0$ for all
$1$-simplices $b_1\in B_1$. Since $G$ is $0$-reduced, it follows that $t_1=0$.

Hence the sum in (\ref{e:twistcoch}) goes only up to $i=\thedim-2$, and so
$\delta^Q$ decreases the filtration degree (given by the degree in $C_*(B)$)
at least by~$2$. The same applies to $\delta^{\rm Bar}$ when we take
the filtration on $\Barr^{A_*}(\Z,Q_*)$ given again by the degree in $C_*(B)$.
Similar to the conclusion of the proof of Proposition~\ref{p:twistprod},
we obtain constant nilpotency bound of $h\delta^{\rm Bar}$.

It remains to show that the perturbation of the differential
in $C_*(B)$ obtained by using the basic perturbation lemma
to the reduction $(f,g,h)\:\Barr^{A_*}(\Z,Q_*)\redu C_*(B)$
with the perturbation $\delta^{\rm Bar}$ is zero. As was
mentioned in connection with the basic perturbation lemma,
the considered perturbation equals $f\delta^{\rm Bar}\varphi g$,
where $\varphi=\sum_{i=0}^\infty (-1)^i(h\delta^{\rm Bar})^i$.

We will check that $f\delta^{\rm Bar}=0$. Indeed, the mapping
$f$ in the reduction from Lemma~\ref{l:bar-inverse} is obtained
from the augmentation $\varepsilon\:A_*\ra\Z$ by tensoring with
$\id_{C_*(B)}$. Thus,
if $z=a_1\otimes\cdots\otimes a_n\otimes(a\otimes b)$ is a basis
element, we have $f(z)=0$ unless $\deg(a)=0$. But the
Eilenberg--MacLane product $a\cdot t_{\thedim-i}(b_{\thedim-i})$
in the formula (\ref{e:twistcoch}) has degree at least
$\deg(t_{\thedim-i}(b_{\thedim-i}))=\thedim-i-1$.
Thus, the degree can be $0$ only for $\thedim-i-1=0$,
but in this case  $t_{\thedim-i}=t_1=0$,
and so $f\delta^{\rm Bar}=0$ as claimed.
\end{proof}

\subsection{Eilenberg--MacLane spaces}\label{s:EML1}

\heading{Preliminaries on cochains. } Before entering the realm
of Eilenberg--MacLane spaces, we recall a few notions related
to cohomology. Throughout this section, let $\pi$ be an Abelian
group.

For us, it will often be convenient to regard cochains as
homomorphisms from chain groups into $\pi$. That is, given a
chain complex $C_*$ (whose chain groups are, as always in this paper,
free Abelian groups), we define its \emph{$\thedim$th cochain group
with coefficients in $\pi$} as $C^\thedim(C_*;\pi):=\Hom(C_\thedim,\pi)$,
with pointwise addition. The \emph{coboundary operator}
$\cobo_\thedim\:C^\thedim(C_*;\pi)\to C^{\thedim+1}(C_*;\pi)$ is then
given by $(\delta_\thedim c^\thedim)(c_{\thedim+1}):=
c^\thedim(\diff_{\thedim+1}c_{\thedim+1})$ for every $\thedim$-cochain
$c^\thedim$ and every $(\thedim+1)$-chain $c_{\thedim+1}$.\footnote{Sometimes
other conventions are used for the coboundary operator in the literature;
e.g.~$(\delta_\thedim c^\thedim)(c_{\thedim+1})=(-1)^{\thedim+1}c^\thedim(\diff_{\thedim+1}c_{\thedim+1})$.
But our main sources \cite{May:SimplicialObjects-1992} and \cite{Hatcher}
use the version without signs.}
(The notation $\delta$ was earlier used for a perturbation of
a differential, but from now on, we will encounter it only
in the role of a coboundary operator.)

In particular, if $X$ is a simplicial set, the \emph{normalized cochain
complex} $C^*(X;\pi)$ is $C^*(C_*(X);\pi)$; thus, a $\thedim$-cochain can be
specified by its values on the standard basis, i.e., as a labeling
of the nondegenerate $\thedim$-simplices by elements of $\pi$---this agrees
with the usual definition in introductory textbooks.

For us, it will be important that if $X$ has infinitely many
nondegenerate $\thedim$-simplices, then a $\thedim$-cochain in
$C^\thedim(X)$ is an infinite object (unlike a $\thedim$-chain!).
Thus, in algorithms, we will need to use a black-box representation
of individual cochains---the black box supplies the value
of the cochain on a given simplex (or on a given chain, which
is computationally equivalent).

To finish our remark on cochains, we recall that if $C^*$ is a cochain
complex, with coboundary operator $\cobo=(\cobo_\thedim)_{\thedim\in\Z}$,
 then  $B^\thedim :=\im \cobo_{\thedim-1}$ is the group of
\emph{$\thedim$-coboundaries}, $Z^\thedim:=\ker \cobo_\thedim$
the group of \emph{$\thedim$-cocycles}, and $H^\thedim=H^\thedim(C^*;\pi):=
Z^\thedim/B^\thedim$ is
the \emph{$\thedim$th cohomology group}.

\heading{Eilenberg--MacLane spaces topologically. }
For an Abelian group $\pi$ and an integer $\thedim\ge 1$,
the \emph{Eilenberg--MacLane space} $K(\pi,\thedim)$ is
defined as any topological space $T$ with $\pi_\thedim(Z)\cong \pi$ and
$\pi_i(T)=0$ for all $i\ne \thedim$ (actually, $K(\pi,1)$ is also
defined for an arbitrary group $\pi$, but we will consider
solely the Abelian case).

It is known, and not too hard to prove,
that a $K(\pi,\thedim)$ exists for all $\thedim\ge 1$ and all $\pi$,
and it is also known to be unique up
to homotopy equivalence.\footnote{Provided that we restrict
to spaces that are homotopy equivalent to CW-complexes.}

The definition postulates that the
homotopy groups of an Eilenberg--MacLane space
are, in a sense, the simplest possible,
and this makes it relatively easy to understand the structure of all
maps from a given space $X$ into $K(\pi,\thedim)$.
Indeed, a basic topological result says that
\begin{equation}\label{e:mapsinK}
[X,K(\pi,\thedim)]\cong H^\thedim(X;\pi),
\end{equation}
assuming that $X$ is a ``reasonable'' space (say a CW-complex).
In words, homotopy classes
of maps $X\to K(\pi,\thedim)$ correspond to the elements of the
$\thedim$th cohomology group of $X$ with coefficients in~$\pi$
(see, e.g., \cite[Lemma~24.4]{May:SimplicialObjects-1992} for
this fact in a simplicial setting, and \cite{CKMSVW11}
for a geometric explanation).

\heading{The standard simplicial model. }
There is a standard way of representing $K(\pi,\thedim)$
as a Kan simplicial set, which actually is even a simplicial group.
We will work with this simplicial representation, and from now
on, the notation $K(\pi,\thedim)$ will be reserved for this particular
simplicial representation, to be defined next.

Let $\Delta^\otherdim$ denote the
$\otherdim$-dimensional standard simplex, regarded as a simplicial complex
(or a simplicial set; the difference is purely formal in this case).
That is, the vertex set is $\{0,1,\ldots,\otherdim\}$
and the $\thedim$-dimensional
(nondegenerate) simplices are all $(\thedim+1)$-element subsets
of $\{0,1,\ldots,\otherdim\}$.

The set of $\otherdim$-simplices of $K(\pi,\thedim)$ is given by
$$
K(\pi,\thedim)_\otherdim := Z^\thedim(\Delta^\otherdim;\pi);
$$
that is, each $\otherdim$-simplex is (represented by) a $\thedim$-dimensional
cocycle on $\Delta^\otherdim$. Thus, it can be regarded as
a labeling of the $\thedim$-dimensional faces of $\Delta^\otherdim$
by elements of the group $\pi$; moreover, the
labels must add up to $0$ on the boundary of every $(\thedim+1)$-face.

It is also easy to define the face and degeneracy operators
in $K(\pi,\thedim)$.
Given an $\otherdim$-simplex $\sigma$ of  $K(\pi,\thedim)$, represented
as a labeling of the $\thedim$-faces of $\Delta^\otherdim$,
$\partial_i\sigma$ is defined as the restriction of $\sigma$
on the $i$th $(\otherdim-1)$-face of $\Delta^\otherdim$.
(The $i$th $(\otherdim-1)$-face
of $\Delta^\otherdim$ is identified with $\Delta^{\otherdim-1}$ via
the unique order-preserving bijection of the vertex sets.)
As for the degeneracy operators, $s_i\sigma$ is the labeling of
$\thedim$-faces of $\Delta^{\otherdim+1}$ induced by
the mapping $\eta_i\:\{0,1,\ldots,\otherdim+1\}
\to\{0,1,\ldots,\otherdim\}$ given by
\[\eta_i(j)=\alterdef{j&\mbox{ for }j\le i,\\
j-1&\mbox{ for }j>i.}
\]
In particular, if a $\thedim$-face contains both $i$ and $i+1$,
then it is labeled by $0$, since its $\eta_i$-image is a degenerate
simplex.

The simplicial group operation in $K(\pi,\thedim)$ is the
addition of cocycles in $Z^\thedim(\Delta^\otherdim;\pi)$.

In the simplicial setting we have
\begin{equation}\label{e:SMapToK}
\SM(X,K(\pi,\thedim))\cong Z^\thedim(X;\pi)
\end{equation}
for every simplicial set $X$. That is, simplicial maps
$X\to K(\pi,\thedim)$ are in a bijective correspondence
with $\pi$-valued $\thedim$-cocycles on~$X$ (see below
for an explicit description
of this correspondence). Moreover,
two such simplicial maps, represented by cocycles
$z$ and $z'$, are homotopic iff $z-z'$ is a coboundary
(see, e.g., \cite[Theorem~24.4]{May:SimplicialObjects-1992}).
This immediately implies $[X,K(\pi,\thedim)]
\cong H^\thedim(X;\pi)$, which was mentioned
above in~(\ref{e:mapsinK}).

\heading{The set $E(\pi,\thedim)$. }
In addition to the simplicial Eilenberg--MacLane space $K(\pi,\thedim)$
we also need another simplicial set, denoted by $E(\pi,\thedim)$.
While the $\otherdim$-simplices of $K(\pi,\thedim)$ are
all $\thedim$-cocycles on $\Delta^\otherdim$, the $\otherdim$-simplices
of $E(\pi,\thedim)$ are all $\thedim$-cochains:
\[
E(\pi,\thedim)_\otherdim := C^\thedim(\Delta^\otherdim;\pi).
\]
The face and degeneracy operators are defined in exactly the same way
as those of $K(\pi,\thedim)$.

\heading{Converting between simplicial maps and cochains. }
We have mentioned that simplicial maps $X\to K(\pi,\thedim)$
are in one-to-one correspondence with cocycles in $Z^\thedim(X;\pi)$.
Similarly, simplicial maps $X\to E(\pi,\thedim)$ correspond
to cochains in $C^\thedim(X;\pi)$:
\[
\SM(X,E(\pi,\thedim))\cong C^\thedim(X;\pi).
\]
Let us describe this correspondence explicitly, since we will
need it in the algorithm. First we note that a $\thedim$-simplex $\tau$
of $E(\pi,\thedim)$ is a $\thedim$-cochain on $\Delta^\thedim$,
i.e., a labeling of the single $\thedim$-face of $\Delta^\thedim$
by an element of $\pi$. Let us denote this element by $\ev(\tau)$
(here ev stands for ``evaluation'').

Given a simplicial map $f\:X\to E(\pi,\thedim)$,
the corresponding cochain $\kappa\in C^\thedim(X;\pi)$
is simply given by $\kappa(\sigma)=\ev(f(\sigma))$ for
every $\sigma\in X_\thedim$ (where on the left-hand side,
$\sigma$ is taken as a generator of the chain group $C_\thedim(X)$).

Conversely, given $\kappa\in C^\thedim(X;\pi)$, we describe
the corresponding simplicial map $f$.
The value $f(\sigma)$ on an $\otherdim$-simplex $\sigma\in X_\thedim$
should be a $\thedim$-chain
on $\Delta^\otherdim$.
There is a unique simplicial map $i_\sigma\:\Delta^\otherdim\to X$
that sends the nondegenerate
$\otherdim$-simplex of $\Delta^\otherdim$ to $\sigma$
(indeed, a simplicial map has to respect the ordering of
vertices, implicit in the face and degeneracy operators).
Then $f(\sigma)$ is the cochain $i_{\sigma}^*(\kappa)$,
i.e., the labels of the $\thedim$-faces of $\sigma$ given by $\kappa$
are pulled back to~$\Delta^{\otherdim}$.
Moreover, if $\kappa$ is a cocycle, then
$f$ goes into $K(\pi,\thedim)$.

\heading{A useful fibration. }
Since an $\otherdim$-simplex $\sigma\in E(\pi,\thedim)$ is formally a
$\thedim$-cochain, we
can take its coboundary $\cobo\sigma$. This is a $(\thedim+1)$-coboundary
(and thus also cocycle), which we can interpret as an $\otherdim$-simplex
of $K(\pi,\thedim+1)$. It turns out that this induces a \emph{simplicial}
map $E(\pi,\thedim)\to K(\pi,\thedim+1)$, which is
(with the usual abuse of notation)
also denoted by~$\cobo$. This map is actually surjective, since
the relevant cohomology groups of $\Delta^\otherdim$ are all zero and
thus all cocycles are also coboundaries.

As is well known, $\cobo\: E(\pi,\thedim)\to K(\pi,\thedim+1)$ is
a fiber bundle with fiber $K(\pi,\thedim)$.

There is another simplicial description of $E(\pi,\thedim)$ as
a twisted product
\[
K(\pi,\thedim)\times_{\tau} K(\pi,\thedim+1),
\]
where $\tau$ has the following explicit form
(see \cite[\S23]{May:SimplicialObjects-1992} or \cite[Sec.~7.10.2]{SergerGenova}):

Let $z\in Z^{\thedim+1}(\Delta^\otherdim;\pi)$ be an $\otherdim$-simplex
of $K(\pi,\thedim+1)$, i.e., a labeling of the $(\thedim+1)$-faces
of $\Delta^\otherdim$ by elements of $\pi$ (satisfying the cocycle
condition). Then we want
$\tau(z)$ to be an $(\otherdim-1)$-simplex of $K(\pi,\thedim)$,
i.e., a labeling of $\thedim$-faces of $\Delta^{\otherdim-1}$.
If we write a $\thedim$-face of $\Delta^{\otherdim-1}$ as an
increasing $(\thedim+1)$-tuple $(i_0,\ldots,i_\thedim)$,
$0\le i_0<\cdots<i_\thedim\le\otherdim-1$, we set
\begin{equation}\label{e:taudef}
(\tau(z))(i_0,\ldots,i_\thedim):=z(0,i_0+1,i_1+1,\ldots,i_\thedim+1)-
z(1,i_0+1,i_1+1,\ldots,i_\thedim+1).
\end{equation}
The twisted product $K(\pi,\thedim)\times_{\tau} K(\pi,\thedim+1)$
is simplicially isomorphic to
$E(\pi,\thedim)$ as defined earlier. The isomorphism will be described,
in a slightly more general setting, in the proof of Corollary~\ref{c:pullback}
below.

\subsection{Polynomial-time homology for $K(\pi,\thedim)$ }
\label{s:EML2}

A crucial ingredient in our algorithm for computing
Postnikov systems is obtaining polynomial-time
homology for $K(\pi,\thedim)$. Here, as usual,
we assume $\thedim$ fixed, and $\pi$ is a globally
polynomial-time Abelian group (as introduced
after Definition~\ref{d:glob-ptime}); then
$K(\pi,\thedim)$ has the same parameter set as~$\pi$.
It is easily checked that
$K(\pi,\thedim)$ is a locally polynomial-time simplicial group.

\heading{The $\overline W$ construction. }
Polynomial-time homology for $K(\pi,\thedim)$ will be constructed by
induction on $\thedim$. The inductive step is based on a construction
$\overline W$ (see \cite[pages 87--88]{May:SimplicialObjects-1992})
that, given an Abelian simplicial group $G$, produces another Abelian simplicial
group $\overline WG$. The $\thedim$-simplices
have the form $\omega=(\gamma_{\thedim-1},\gamma_{\thedim-2},\ldots,\gamma_0)$,
where $\gamma_i$ is an $i$-simplex of $G$, $i=0,1,\ldots,\thedim-1$,
and the group operation in $\overline WG$ is obtained by using the operation
of $G$ componentwise. The face operators are
\begin{eqnarray*}
\partial_0\omega&:=&(\gamma_{\thedim-2},\gamma_{\thedim-3},\ldots,\gamma_0),\\
\partial_{i+1}\omega&:=&(\partial_i\gamma_{\thedim-1},\ldots,
\partial_1\gamma_{\thedim-i},\underbrace{\partial_0
\gamma_{\thedim-i-1}+\gamma_{\thedim-i-2}}_{\mbox{\footnotesize
operation~in~$G$}}, \gamma_{\thedim-i-3},\ldots,\gamma_0),\ \ \
i=0,1,\ldots,\thedim-1,
\end{eqnarray*}
and the degeneracy operators are given by
\begin{eqnarray*}
s_0\omega&:=&(e_\thedim,\gamma_{\thedim-1},\ldots,\gamma_0),\\
s_{i+1}\omega&:=& (s_i\gamma_{\thedim-1},\ldots,s_0\gamma_{\thedim-i-1},
e_{\thedim-i-1},\gamma_{\thedim-i-2},\ldots,\gamma_0),
\ \ \ i=0,1,\ldots,\thedim-1,
\end{eqnarray*}
where $e_\thedim$ is the unit element of~$G_\thedim$.

Topologically, $\overline WG$ is the
classifying space of~$G$, usually denoted by $BG$,
but we won't use this fact directly.
What we need is the following simplicial isomorphism.

\begin{lemma}\label{l:barW-iso}
 For every Abelian group $\pi$ and every $\thedim\ge 1$,
there is a simplicial isomorphism
\[
f\: K(\pi,\thedim+1)\to\overline WK(\pi,\thedim);
\]
if $\thedim$ is fixed and $\pi$ is globally polynomial-time,
then both $f$ and $f^{-1}$ are polynomial-time maps.
Consequently, polynomial-time homology for
 $\overline WK(\pi,\thedim)$ yields
polynomial-time homology for $K(\pi,\thedim+1)$.
\end{lemma}

\begin{proof}
We define an auxiliary simplicial set $WK(\pi,\thedim)$
as the twisted Cartesian product $K(\pi,\thedim)\times_\tau \overline
WK(\pi,\thedim)$,
where $\tau\: K(\pi,\thedim+1)\to K(\pi,\thedim)$ is the twisting
operator of $\delta$ introduced at the end of Section~\ref{s:EML1}.
Then, according to \cite[Theorem~23.10]{May:SimplicialObjects-1992},
there are simplicial isomorphisms
$f\: K(\pi,\thedim+1)\to \overline W K(\pi,\thedim)$
and $F\: E(\pi,\thedim)\to W K(\pi,\thedim)$ that are compatible
with respect to the projection maps
 $\delta\:E(\pi,\thedim)\ra K(\pi,\thedim+1)$ and
$WK(\pi,\thedim)\ra\overline WK(\pi,\thedim)$.
By \cite[Lemma~21.9]{May:SimplicialObjects-1992} and the formula (1)
there, the isomorphism $f$
maps $z\in K(\pi,\thedim+1)_\otherdim$ to
\[
f(z):= \Bigl(\tau(z),\tau(\partial_0z),\tau(\partial_0^2z),\ldots,
\tau(\partial_0^{\otherdim-1}z)\Bigr)\in \overline WK(\pi,\thedim)_{\otherdim}
\]
where $\tau$ is the twisting
operator as above.
Combining these statements together it follows that $f$ is an isomorphism,
and to finish the proof, we need to compute its inverse in polynomial time.

We describe an inductive algorithm for this.
First we note that $$f(z)=(\tau(z),f(\partial_0z)).$$
There is only one simplex in dimension at most $\thedim$ in both
of the considered simplicial sets, so the isomorphism
is given uniquely there. A $(\thedim+1)$-simplex
of $\overline WK(\pi,\thedim)$ has the form $\omega=(w_\thedim,0,0,\dots,0)$,
where $w_\thedim \in Z^\thedim(\Delta^\thedim;\pi)$.
Defining $z_{\thedim+1}\in K(\pi,\thedim+1)_{\thedim+1}=
Z^{\thedim+1}(\Delta^{\thedim+1};\pi)$ by
$z_{\thedim+1}(0,1,2,\dots,\thedim+1):=w_\thedim(0,1,\dots,\thedim)$,
we get $f(z_{\thedim+1})=(\tau(z_{\thedim+1}),0,\ldots,0)=\omega$,
so we have found $f^{-1}(\omega)$.

Next, we suppose that we can compute $f^{-1}$ for simplices
up to dimension $\otherdim\ge\thedim+1$, and let
$\omega=(w_{\otherdim},w_{\otherdim-1},\dots,w_0)
\in\overline WK(\pi,n)_{\otherdim+1}$.
In order to obtain $z=f^{-1}(\omega)$,
we first inductively compute $z'=f^{-1}(w_{\otherdim-1},\dots,w_0)$;
then $z'=\partial_0 z$, and by the definition of
$\partial_0$ in $K(\pi,\thedim+1)$, we get that
for $1\le i_0< i_1<\cdots<i_{\thedim+1}\le \otherdim+1$
we have
\begin{equation}\label{e:w'}
z(i_0,i_1,\dots,i_{\thedim+1})=z'(i_0-1,i_1-1,\dots, i_{\thedim+1}-1).
\end{equation}
On the other hand, for $0=i_0<i_1<\cdots<i_{\thedim+1}\le \otherdim+1$,
from the formula (\ref{e:taudef}) defining $\tau$ we obtain
\begin{eqnarray}\nonumber
\tau(z)(i_1-1,\ldots,i_{\thedim+1}-1) & =&z(0,i_1,\ldots,i_{\thedim+1})-z(1,i_1,\ldots,i_{\thedim+1}) \\
& =&z(0,i_1,\ldots,i_{\thedim+1})-z'(0,i_1-1,\ldots,i_{\thedim+1}-1).
\label{e:w''}
\end{eqnarray}
From this we can express $z(0,i_1,\ldots,i_\thedim)$ in terms of
$\tau(z)=w_\ell$ and $z'$, which are both known.
This finishes the construction of the inverse.
\end{proof}

Now we can state the main result of this section.

\begin{theorem}\label{t:EML}
Let $\thedim\ge 1$ be a fixed integer.
The standard simplicial model of the Eilenberg--MacLane
space $K(\pi,\thedim)$, where $\pi$ is a globally polynomial-time
Abelian group, can be equipped with polynomial-time
homology.
\end{theorem}

\begin{proof}
The proof proceeds by induction on $\thedim$. The base
case is $K(\pi,1)$, and it goes as follows.

\begin{enumerate}
\item Polynomial-time homology for $K(\Z,1)$ is the main
result of \cite{pKZ1}.
\item Polynomial-time homology for $K(\Z/m,1)$ is derived
from that for $K(\Z,1)$ in Lemma~\ref{l:ZtoZm} below.
\item For $\pi$ arbitrary, we use
the specified polynomial-time isomorphism $\pi\cong\AB(\mm)$ to write $K(\pi,1)\cong K(\AB(\mm),1)$. Since $\AB(\mm)$ decomposes into a direct sum of cyclic groups, we can obtain polynomial-time homology for $K(\pi,1)$ using
\[K(\pi_1\oplus\cdots\oplus\pi_s,1)\cong K(\pi_1,1)\times\cdots\times K(\pi_s,1),\]
which is easy to see from the definition
of $K(\pi,1)$, plus Proposition~\ref{p:bigprod}
(product with many factors).
\end{enumerate}

The inductive step from $K(\pi,\thedim)$
to $K(\pi,\thedim+1)$ is as in
\cite{RealThesis}, and it goes as follows.

\begin{enumerate}
\item To get polynomial-time homology for
$K(\pi,\thedim+1)$, according to Lemma~\ref{l:barW-iso}
it suffices to obtain  polynomial-time homology for $\overline WK(\pi,\thedim)$.
\item With $G=K(\pi,\thedim)$, let us consider the
twisted product $G\times_\tau
\overline WG$, where the twisting operator is given by
$\tau_\otherdim(\gamma_{\otherdim-1},\ldots,\gamma_0):=\gamma_{\otherdim-1}$
(this twisted product was denoted by $WG$ in the proof of
Lemma~\ref{l:barW-iso}).
Then there is a reduction
\[
(f,g,h): C_*(G\times_\tau \overline WG)\reduP \Z,
\]
with $f,g$ defined in the obvious way (note that both $G$
and $\overline WG$ are $0$-reduced), and with $h$ given by
$h_\otherdim(\gamma_\otherdim,(\gamma_{\otherdim-1},\ldots,\gamma_0)):=
(e_{\otherdim+1},(\gamma_\otherdim,\gamma_{\otherdim-1},\ldots,\gamma_0))$,
where
$e_{\otherdim+1}$ is the unit element of $G_{\otherdim+1}$
(see \cite[page~88]{May:SimplicialObjects-1992}).
Thus, using Proposition~\ref{p:tw-div} (twisted division)
with $B=\overline WG$, we obtain polynomial-time homology
for $\overline WG$ from that of~$G$.
\end{enumerate}

The proof of Theorem~\ref{t:EML} is finished,
except for the proof of the next lemma.
\end{proof}

\begin{lemma}\label{l:ZtoZm} Given a polynomial-time homology for
$K(\Z,1)$, one can equip
$K(\Z/m,1)$ (parameterized by the natural number
$m$ encoded in binary) with polynomial-time homology.
\end{lemma}

We note that the simplicial set $K(\Z/m,1)$ has finitely many
simplices in each dimension (the number is even bounded by a polynomial
in $m$ for every fixed dimension). Nevertheless, we cannot treat
it as a finite simplicial set, since it is parameterized
by the group $\Z/m$, whose encoding size is only $\log m$,
and so the number of simplices is exponential in this size.
Somewhat paradoxically, we will use the infinite simplicial
set $K(\Z,1)$ to get a handle on the finite (in every dimension) $K(\Z/m,1)$.

\begin{proof}
By the assumption,
the simplicial group $K(\Z,1)$ is equipped with polynomial-time homology.

We will exhibit a twisting operator $\tau$
such that the principal twisted Cartesian product
$P:=K(\Z,1)\times_\tau K(\Z/m,1)$ is simplicially isomorphic to~$K(\Z,1)$.
Let $\varphi\: P\to K(\Z,1)$ be the isomorphism;
assuming that both $\varphi$ and $\varphi^{-1}$ are polynomial-time
maps, we can thus equip $P$ with polynomial-time
homology as well. Then we obtain the desired polynomial-time
homology for $K(\Z/m,1)$ from Proposition~\ref{p:tw-div}
(twisted division).

Conceptually, the isomorphism $\varphi$ is obtained
from the short exact sequence of Abelian groups
\[\xymatrixcolsep{3pc}
\xymatrix{
0\ar[r]&\Z \ar[r]^{\times m}& \Z\ar[r]^
{{\rm mod}~m~~~}& \Z/m\ar[r]& 0
}
\]
by passing to classifying spaces. But our presentation below
does not refer to this approach and
is completely elementary.

In order to define $\varphi$ and $\tau$, it will be convenient
to use a particular representation of simplices in
$K(\Z,1)$ and in $K(\Z/m,1)$, described next.

We recall that the $\otherdim$-simplices of $K(\Z,1)$ are
1-dimensional integral cocycles on $\Delta^\otherdim$,
in other words, labelings $c$ of the edges of the complete graph
on $\{0,1,\ldots,\otherdim\}$ with integers such that, for every
triple $i<j<k$, $c(i,j)-c(i,k)+c(j,k)=0$. It is easy to see
that every such labeling is determined by a ``potential function''
$a$ on the vertex set, i.e., $c(i,j)=a(j)-a(i)$
(from the topological point of view, every cocycle $c$ is
a coboundary since $\Delta^\otherdim$ is contractible,
and $a$ is a $0$-cochain with $c=\cobo a$). Moreover, w.l.o.g.
we can assume that $a(0)=0$, and then $a$ is determined uniquely.

Then we represent the $\otherdim$-simplex $c$ by the
$\otherdim$-tuple $\alpha=(a_1,a_2,\ldots,a_\otherdim)$, where we write
$a_i$ instead of $a(i)$ for typographic reasons. The boundary operators
then work as follows:
\begin{eqnarray*}
\partial_0\alpha&=&(a_2-a_1,a_3-a_1,\ldots,a_\otherdim-a_1),\\
\partial_i\alpha&=&(a_1,a_2,\ldots,a_{i-1},a_{i+1},\ldots,a_\otherdim),\ \ \
i=1,2,\ldots,\otherdim.
\end{eqnarray*}
The degeneracy operator $s_0$ prepends $0$ to the beginning of the
sequence, and for $i\ge 1$, $s_i$ duplicates the $i$th term.
An analogous representation is used for the simplices of $K(\Z/m,1)$.

Now if $\alpha=(a_1,\ldots,a_\otherdim)\in K(\Z,1)_\otherdim$
and $\beta=(b_1,\ldots,b_\otherdim)\in K(\Z/m,1)_\otherdim$
are simplices represented
in this way, the desired simplicial isomorphism
$\varphi\: K(\Z,1)\times_\tau K(\Z/m,1)\to K(\Z,1)$ is defined by
\[
\varphi_\otherdim(\alpha,\beta):= (ma_1+\intg(b_1),\ldots,ma_\otherdim+
\intg(b_\otherdim)),
\]
where $\intg\:\Z/m\to\Z$ is the identification of
$\Z/m$ with $\{0,1,\ldots,m\}\subseteq\Z$.
It is clear that $\varphi_\otherdim$ is a bijection between the sets
of $\otherdim$-simplices, and that both $\varphi$ and $\varphi^{-1}$
are polynomial-time computable.

We recall that in the twisted product $K(\Z,1)\times_\tau K(\Z/m,1)$
we have $s_i(\alpha,\beta)=(s_i\alpha,s_i\beta)$ for all $i$,
and $\partial_i(\alpha,\beta)=(\partial_i\alpha,\partial_i\beta)$
for all $i\ge 1$. It is then straightforward to check that
the mapping $\varphi$ commutes with $s_0,\ldots,s_\otherdim$
and with $\partial_1,\ldots,\partial_\otherdim$.

The face operator $\partial_0$ is twisted, i.e.,
$\partial_0(\alpha,\beta)=(\tau(\beta)+\partial_0\alpha,\partial_0\beta)$
(here we write the group operation additively, unlike in the
general discussion of twisted products earlier). From the requirement
that $\varphi$ commute with $\partial_0$, we can compute the appropriate
twisting operator~$\tau$.

Namely, we have
\[
\partial_0\varphi_\otherdim(\alpha,\beta)=
\Bigl(m(a_2-a_1)+\intg(b_2)-\intg(b_1),\ldots, m(a_\otherdim-a_1)+
\intg(b_\otherdim)-\intg(b_1)\Bigr),
\]
while
\[
\varphi_{\otherdim-1}(\partial_0\alpha,\partial_0\beta)=
\Bigl(m(a_2-a_1)+\intg(b_2-b_1),\ldots, m(a_\otherdim-a_1)+
\intg(b_\otherdim-b_1)\Bigr)
\]
(where the subtraction in the argument of $\intg$ is in $\Z/m$, i.e.,
modulo $m$). It follows that $\tau$ has to be given by
\[
\tau_\otherdim(\beta)=\Bigl(\intg(b_2)-\intg(b_1)-\intg(b_2-b_1),\ldots,
\intg(b_\otherdim)-\intg(b_1)-\intg(b_\otherdim-b_1)\Bigr).
\]
This is obviously a polynomial-time map, and a routine check of properties
(i)--(iv) of a twisting operator in Definition~\ref{d:twistprod}
concludes the proof.
\end{proof}

\subsection{A pullback from a fibration of Eilenberg--MacLane spaces}
\label{s:pullb}

For our construction of Postnikov systems, we will need an operation
that is essentially a twisted Cartesian product, but in a somewhat
different representation.
We will have the following situation. We are given
a simplicial set $P$, plus a simplicial mapping
$f\:P\to K(\pi,\thedim+1)$, for some Abelian group $\pi$
and a fixed~$\thedim\ge 1$.

Now we define a simplicial set $Q$ as the \emph{pullback}
according to the following commutative diagram:
\[
\xymatrix{
Q \ar[r] \ar[d] & E(\pi,\thedim) \ar[d]^{\cobo}\\
P\ar[r]^-{f} & K(\pi,\thedim+1)
}
\]
This means that $Q$ is the simplicial subset of the Cartesian product
$P\times E(\pi,\thedim)$ consisting of the pairs
$(\alpha,\beta)$ of simplices $\alpha\in P_\otherdim$,
$\beta\in E(\pi,\thedim)_\otherdim$
with $f(\alpha)=\cobo(\beta)$.

As a simple consequence of Proposition~\ref{p:twistprod} (twisted product)
and of an explicit isomorphism of the pullback with a suitable twisted
product, we obtain the following.

\begin{corol}\label{c:pullback} Given $\pi,\thedim,P,f$ as above, where $\pi$ is a globally polynomial-time Abelian group, $P$ is equipped with polynomial-time homology, and $f$ is polynomial-time, all parameterized by
the same parameter set $\II$,
the pullback $Q$ can be equipped with polynomial-time homology.
\end{corol}

\begin{proof} Let $\tau$ be the twisting operator in the
twisted product $K(\pi,\thedim)\times_\tau K(\pi,\thedim+1)$
at the end of Section~\ref{s:EML1},
and let $\tau^*$ be the pullback of $\tau$ by $f$; that is,
$\tau^*(\alpha):=\tau(f(\alpha))$. Then Proposition~\ref{p:twistprod}
yields polynomial-time homology for the twisted
product $K(\pi,\thedim)\times_{\tau^*} P$. According to
\cite[Prop.~18.7]{May:SimplicialObjects-1992} (which is formulated
in a more general setting), there is a simplicial isomorphism
$\varphi\:K(\pi,\thedim)\times_{\tau^*} P\to Q$, given by
\[
\varphi(\alpha,\beta) := (\psi(f(\alpha))+\beta,\alpha),
\]
where $\psi\:K(\pi,\thedim+1)\to E(\pi,\thedim)$ is the
\emph{pseudo-section} given by
\[
\psi(z)(i_0,\ldots,i_{\thedim}):=z(0,i_0+1,\ldots,i_\thedim+1),
\]
with the same notation as in the definition of~$\tau$.
Since both $\varphi$ and its inverse are polynomial-time maps,
we obtain polynomial-time homology for~$Q$ as needed.

In addition, setting $P=K(\pi,\thedim+1)$, we have
$Q=E(\pi,\thedim)$ and we obtain the isomorphism $E(\pi,\thedim)\cong
K(\pi,\thedim)\times_\tau K(\pi,\thedim+1)$ mentioned
at the end of Section~\ref{s:EML1}.
\end{proof}




\section{Postnikov systems}\label{s:post}

Let $Y$ be a topological space, which we will assume to be given
as a simplicial set equipped with polynomial-time homology.
Moreover, we assume that $Y$ is $1$-connected.
This is needed for the proof of correctness of the algorithm;
the algorithm itself does not
make use of any certificate of $1$-connectedness,
and in particular, we do not assume
$Y$ $1$-reduced.

For our purposes, we define a \emph{(simplicial) Postnikov system} of $Y$
as the collection of simplicial sets and simplicial maps organized
into the following commutative diagram,%
\newlength{\lxxx}
$$
\xymatrix{
& & \vdots  \ar[d]^{p_3} \\
& & P_{2} \ar[d]^{p_2}\\
& & P_{1} \ar[d]^{p_1} \\
Y \ar[uurr]^-{\varphi_2}
  \ar[urr]^{\ \ \ \ \varphi_{1}}
  \ar[rr]^{\varphi_0}&&
\leftbox{P_0}{{}=*}\\
}
$$
where $P_0$ is a single point, and the following conditions hold:
\begin{enumerate}
\item[(i)] For each $\thedimm\ge 0$,
the map $\varphi_\thedimm: Y\to P_\thedimm$
induces isomorphisms  $\varphi_{\thedimm*}:\pi_i(Y)\to \pi_i(P_\thedimm)$
of homotopy groups for $0\le i\le\thedimm$,
while $\pi_i(P_\thedimm)=0$ for $i\ge \thedimm+1$.
\item[(ii)] Each $P_\thedimm$, $\thedimm\ge 1$, is the pullback
according to the following diagram (as in Section~\ref{s:pullb})
for some
 map $\kkk_{\thedimm-1}\: P_{\thedimm-1}\to  K(\pi_\thedimm(Y),\thedimm+1)$:
\[
\xymatrix{
P_\thedimm \ar[r] \ar[d]_{p_\thedimm} & E(\pi_\thedimm(Y),\thedimm) \ar[d]^{\cobo}\\
P_{\thedimm-1}\ar[r]^-{\kkk_{\thedimm-1}} & K(\pi_\thedimm(Y),\thedimm+1)
}
\]
\end{enumerate}

The simplicial sets $P_0,P_1,\ldots$ are called the \emph{stages}
of the Postnikov system, and
the mappings $\kkk_{i}$ are called
\emph{Postnikov classes}  (the terms \emph{Postnikov factors}
or \emph{Postnikov invariants} are also used in the literature).

In the simplicial Postnikov system as introduced above,
each $P_{\thedimm}$ is a simplicial subset of the Cartesian
product $ P_{\thedimm-1}\times E(\pi_\thedimm(Y),\thedimm)$,
and the map $p_\thedimm\:P_\thedimm\to P_{\thedimm-1}$
is the projection to the first component.

In the rest of this section, we will prove Theorem~\ref{t:ipost}.
First we should make the statement precise.

\begin{theorem}[Restatement of Theorem~\ref{t:ipost}]\label{t:restat14}
Let $\thedim\geq 2$ be fixed and let $(Y(I):I\in\II)$
be a simplicial set with polynomial-time homology, the main example being a
finite simplicial complex, and let us suppose that $Y$ is
$1$-connected (or simple;  see the remark following
Theorem~\ref{t:pi_k}). Then there is a polynomial-time
algorithm that, given $I\in\II$, computes, for each $i\le\thedim$,
the isomorphism type $\mm_i=\mm_i(I)$ of the homotopy group $\pi_i(Y(I))$.
Furthermore, we can construct the following objects
(i.e., write down the algorithms for the black boxes representing them,
which use the black boxes defining $Y$ as subroutines).
\begin{itemize}
\item
Simplicial sets $P_0,P_1,\ldots,P_\thedim$ with polynomial-time homology.
\item
Polynomial-time simplicial maps $\varphi_i\:Y\to P_{i}$,
$i\le\thedim$.
\item
Polynomial-time simplicial maps $\kkk_{i-1}\:P_{i-1}\to K(\pi_i,i+1)$,
$i\le\thedim$, where we use the notation $\pi_i:=\AB(\mm_i)$
for the canonical representation of the Abelian group described
by $\mm_i$ (see the text following Definition~\ref{d:glob-ptime}).
\end{itemize}
All of these objects are parameterized by $\II$.
The $P_i(I)$, the $(\varphi_i)_I$,
and the $(\kkk_{i-1})_I$ form a Postnikov system of~$Y(I)$.
\end{theorem}

%

\subsection{The algorithm}

\heading{Representing a simplicial map by an effective cocycle. }
In the Postnikov system algorithm, we will encounter the
following situation. We consider a simplicial set $(U(I):I\in\II)$
with polynomial-time
homology; let us write $\EC_*^U$ for the globally polynomial-time
chain complex used in the polynomial-time
homology, i.e., the one for which $C_*(U)\steqP \EC_*^U$.

Let us also consider a $(\thedim+1)$-cocycle $\psi^{\ef}\in
Z^{\thedim+1}(\EC_*^U;\pi)$ for some globally
polynomial-time Abelian group $\pi$, also parameterized
by $\II$; here
the superscript ``ef'' should suggest that the cocycle belongs
to the ``effective'' chain complex $\EC_*^U$ associated to~$U$.
Then $\psi^\ef$ can be represented by
a finite matrix, since it is a homomorphism from the
chain group $\EC_{\thedim+1}^U$ of finite rank into~$\pi$.

Now the strong equivalence $C_*(U)\steqP \EC_*^U$ defines, in
particular, a chain map $f\:C_*(U)\to \EC_*^U$. We define
a cocycle $\psi\in Z^{\thedim+1}(C_*(U))$ as
$\psi=f \psi^\ef$. As was discussed in Section~\ref{s:EML1},
such a $\psi$ canonically defines a simplicial map
$\hat\psi\:U\to K(\pi,\thedim+1)$.

The point we want to make here is that $\hat\psi$ can be regarded
as a polynomial-time simplicial map parameterized by
pairs $(I,\psi^\ef)$.

\heading{Re-parameterizing the Postnikov system. }
In Theorem~\ref{t:restat14}, we have the Postnikov system
parameterized by the same parameter set $\II$ as the input
simplicial set $Y$. This simplifies the formulation, but
as we have already remarked earlier, it is not very efficient
for an implementation, since it stipulates re-computing everything
from scratch every time we call one of the black boxes representing
the Postnikov system.

We are going to organize the algorithm somewhat differently.
We are going to define a new parameter set $\JJ_\thedim$,
whose elements have the form $(I,F_\thedim(I))$, where $F_\thedim$
is a polynomial-time mapping described below. The computation
of $F_\thedim(I)$ corresponds to a preprocessing, or ``construction''
of the Postnikov system. Then we will have the Postnikov system
parameterized by $\JJ_\thedim$ instead of $\II$, and this will
allow for much more effective black boxes. This point of view
is also very natural for presentation of the Postnikov
system algorithm.

What kind of data should be included in $\JJ_\thedim$ to
describe the Postnikov system? First, given $I\in\II$, we need the
homotopy groups $\pi_i(Y(I))$, $i\le\thedim$. As
in Theorem~\ref{t:restat14}, we are going to
represent  each $\pi_i(Y(I))$  by its isomorphism
type $\mm_i$, and we use the notation $\pi_i=\AB(\mm_i)$.
Thus $\mm_1,\ldots,\mm_\thedim$ are
included in~$F_\thedim(I)$.



Next, the Postnikov stage $P_\thedim$ is a simplicial subset of the
product
\[
P_\thedim\subseteq E(\pi_1,1)\times\cdots\times E(\pi_\thedim,\thedim),
\]
and for describing it, we need the Postnikov classes
$\kkk_{i-1}$, $i\le\thedim$. We are going to have $\kkk_{i-1}$
represented by a cocycle $\kappa_{i-1}^\ef\in Z^{i+1}(\EC_*^{P_{i-1}};\pi_{i})$,
in the way described above, and $\kappa_1^\ef,\ldots,\kappa_{\thedim-1}^\ef$
are also a part of $F_\thedim(I)$.

This, of course, assumes that
$P_{i-1}$ has already been equipped with polynomial-time homology;
indeed, the algorithm will proceed inductively, constructing
$P_{i-1}$
first, then $\kappa_{i-1}^\ef$ (and thus $\kkk_{i-1}$), and then $P_{i}$.
Here $P_i$ with polynomial-time homology is obtained as the
pullback as in the definition of a Postnikov system,
using Corollary~\ref{c:pullback}.

Finally, to describe the maps $\varphi_1,\ldots,\varphi_\thedim$,
we need even more data. Namely, $\varphi_\thedim$ is,
in particular, a simplicial map into $E(\pi_1,1)\times\cdots
\times E(\pi_\thedim,\thedim)$, and so we can write it
as $(\ell_1,\ldots,\ell_\thedim)$, where $\ell_i$ goes into
$E(\pi_i,i)$. Each $\ell_i$ is going to be specified
using a cochain $\lambda_i^\ef\in Z^{i}(\EC_*^Y;\pi_i)$.
The construction of $\ell_i$ from $\lambda_i^\ef$ is
described in the algorithm below;
it is roughly similar to the construction of $\kkk_i$
from $\kappa_i^\ef$, but there is a subtlety involved.

Hence the parameter $J\in\JJ_\thedim$ describing the first
$\thedim$ stages of the Postnikov system has the form
\[J=(I,\mm_1,\lambda_1^\ef,\kappa_1^\ef,\mm_2,\ldots,\kappa_{\thedim-1}^\ef,\mm_\thedim,\lambda_\thedim^\ef).\]
%
%
Of course, the $\kappa_i^\ef$ and $\lambda_i^\ef$ have to satisfy
certain consistency requirements, so that they describe a valid
Postnikov system (up to stage $\thedim$). These will be formulated
and proved later.


\heading{The Postnikov system algorithm. }
Now we describe the way of computing $F_\thedim(I)$, i.e.,
obtaining the values of $\mm_1,\lambda_1^\ef,\kappa_1^\ef,\ldots,
\kappa^\ef_{\thedim-1},\mm_\thedim,\lambda_\thedim^\ef$
from $I$ (using the black boxes
specifying $Y$, of course).

As was mentioned above, we proceed by induction.
By definition, there is nothing to compute for $\thedim=0$ and,
in order to make the induction start, we define $P_0$ to be a single point
and $\varphi_0$ to be the constant map. Next, we assume that the algorithm
for $F_{\thedim-1}$, computing the parameters, is given and we are required
to compute the components $\kappa_{\thedim-1}^\ef$, $\mm_\thedim$,
and~$\lambda_\thedim^\ef$.


\begin{enumerate}
\item\label{step:cone}
Construct the algebraic mapping cone $\MCo_*:=\MCone_*((\varphi_{\thedimm-1})_*)$, where $(\varphi_{\thedimm-1})_*\:C_*(Y)\to C_*(P_{\thedimm-1})$ is the chain map induced by $\varphi_{\thedimm-1}$, as a chain complex with polynomial-time homology, by Proposition~\ref{p:mcone}.
By the proof of that proposition,
the corresponding globally polynomial-time chain complex $\EC^M_*$ has
$\EC^M_{\thedim+1}=\EC_{\thedim}^Y\oplus\EC_{\thedim+1}^{P_{\thedim-1}}$.

\item
Compute the homology group
$H_{\thedimm+1}(\EC^M_*)$ as a globally polynomial-time Abelian group.
We let $\mm_{\thedimm}$ be its isomorphism type,
and let $\pi_\thedimm=\AB(\mm_\thedimm)$.
 We also have an explicit, polynomial-time
isomorphism $H_{\thedimm+1}(\EC^M_*)\cong\pi_\thedimm$,
as in the definition of a globally polynomial-time Abelian group.

\item\label{step:getDecomp}
Choose a decomposition of the chain group $\EC^M_{\thedimm+1}$ of the
form $\EC^M_{\thedimm+1}=\EZ_{\thedimm+1}\oplus \tEM_{\thedimm+1}$,
where $\EZ_{\thedimm+1}$ is the subgroup of all cycles,
and $\tEM_{\thedimm+1}$ is an arbitrary direct complement.

Let $\rho\: \EC^M_{\thedimm+1}\to \pi_{\thedimm}$ be given as the projection
\[\rho\:\EC^M_{\thedimm+1}=\EZ_{\thedimm+1}\oplus \tEM_{\thedimm+1}\to \EZ_{\thedimm+1}\to H_{\thedimm+1}(\EC^M_*)\stackrel{\cong}{\longrightarrow}\pi_{\thedimm}.\]
In other words, every chain $c\in\EC^M_{\thedimm+1}$ has a unique expression as $c=z+\tilde c$, $z\in \EZ_{\thedimm+1}$, $\tilde c\in \tEM_{\thedimm}$, and $\rho(c)$ is  the element of $\pi_\thedimm$
corresponding to the homology class
$[z]\in H_{\thedimm+1}(\EC^M_*)\cong\pi_\thedimm$.

\item
Using the decomposition of $\EC^M_{\thedimm+1}$ as in Step~\ref{step:cone}, we denote the restriction of $\rho$ to $\EC_{\thedimm}^Y$ by
$\lambda_{\thedim}^\ef$ and the restriction to
$\EC_{\thedimm+1}^{P_{\thedimm-1}}$ by $\kappa_{\thedim-1}^\ef$.
In effect, to give $\rho$ is the same as to give its
two components $\lambda_\thedimm^\ef$ and $\kappa_{\thedimm-1}^\ef$.

\item\label{step:nonef}
In the strong equivalence $M_*\steqP\EC^M_*$, let $f$ denote the composite chain map $M_*\to\EC^M_*$. Then we obtain a cochain $\rho f\:M_{\thedim+1}\to\pi_\thedim$. Again we have a direct sum decomposition $M_{\thedimm+1}=C_\thedim(Y)\oplus C_{\thedim+1}(P_{\thedim-1})$. We define $\lambda_\thedim\:C_\thedimm(Y)\to\pi_\thedimm$ as the restriction of $\rho f$ to the summand $C_\thedim(Y)$ and $\ell_\thedimm:Y\to E(\pi_\thedimm,\thedimm)$ as the corresponding simplicial map; it is clearly polynomial-time.
\end{enumerate}

It is easy to see that all the computations can be implemented
in polynomial time. Perhaps only the decomposition in Step~\ref{step:getDecomp}
may need some comment. The computation of $\EZ_{\thedimm+1}$ is a part
of computing the homology group $H_{\thedimm+1}(\EC^M_*)$. Then, given a basis
of $\EZ_{\thedimm+1}$, it suffices to extend it to a basis of the
free Abelian group $\EC^M_{\thedimm+1}$, which is also straightforward
using the Smith normal form.

To prove correctness, we will need to verify that
$\pi_\thedim\cong\pi_\thedim(Y)$, that $\kappa_{\thedim-1}^\ef$ is a cocycle,
that the image of the induced map
$\varphi_\thedim=(\varphi_{\thedim-1},\ell_\thedim)$ lies in $P_\thedim$,
and that it satisfies the conditions in the definition of a Postnikov system.
The proofs of all these claims are postponed to Section~\ref{s:correctness}.

\heading{Remark: non-uniqueness. }
A Postnikov system of a space $Y$ is typically not unique.
The algorithm above involves some arbitrary choices, namely,
the choice of the direct complement of $\EZ_{\thedim+1}$
in Step~\ref{step:getDecomp}, as well as the choice of the
isomorphism of $H_{\thedim+1}(\EC^M_*)$ with $\AB(\mm_i)$.
Performing these choices differently may result in a different
Postnikov system.

At the same time, in an algorithm that uses a Postnikov system,
such as the one in Corollary~\ref{c:[XY]}, we make many calls
to the black boxes representing the Postnikov system,
and we thus need that each time they refer to the same
Postnikov system, for otherwise, the algorithm
may not work correctly. This requirement is
reflected in the definition of  a parameterized simplicial set
$(X(I):I\in\II)$, where $I$ determines $X(I)$ uniquely.

One way of satisfying this requirement is to use only deterministic
algorithms (no randomization). Then, although the algorithm
makes some ``arbitrary'' choices, these choices are always made
in the same way for a given input.

Another, more conceptual and practical way, is the re-parameterization
as above: the results of all of the arbitrary choices are encoded
in $F_\thedim(I)$, and then the Postnikov stages $P_i(J)$ are
defined uniquely, and similarly for the $(\kkk_i)_J$ and $(\varphi_i)_J$.
In this case the computation of $F_\thedim(I)$ may use randomized
algorithms as well, which may be useful, e.g., for a fast computation
of the Smith normal form.





\subsection{Further properties of Eilenberg--MacLane spaces}

Here we prepare several lemmas needed in the proof of correctness
of our algorithm for computing Postnikov systems.
The proofs are routine, but we have no good reference for
these facts. Here, $\pi$ will stand for an Abelian group.

We recall that $\ev\:K(\pi,\thedimm)_\thedimm=E(\pi,\thedimm)_\thedimm \to\pi$
is the mapping assigning to each $\pi$-valued cocycle
$z\in Z^\thedimm(\Delta^\thedimm;\pi)$ its value on the
unique $\thedimm$-face of $\Delta^\thedimm$.
We can extend $\ev$ linearly to a  homomorphism
$\ev\:C_\thedimm(K(\pi,\thedimm))\to\pi$.

The first lemma is essentially just re-phrasing of the considerations
in Section~\ref{s:EML1} concerning the correspondence of simplicial
maps into $E(\pi,\thedimm)$ with cochains.

\begin{lemma}[Lemma~24.2 in \cite{May:SimplicialObjects-1992}]\label{l:ev*}
Let $f\: X\to E(\pi,\thedimm)$ be a simplicial map. Then the
cochain $\kappa\:C_\thedimm(X)\to \pi$ corresponding to it can be expressed as
$\kappa=\ev f_*$, where $f_*\:C_*(X)\to C_*(E(\pi,\thedimm))$
is the chain map induced by~$f$.
\end{lemma}

Also see \cite[Lemma~24.3]{May:SimplicialObjects-1992} for
the corresponding statement for $K(\pi,\thedim)$.

The next two lemmas deal with maps induced by $\ev$ in homology.

\begin{lemma}\label{eval}
The  homomorphism
$\ev\:C_\thedimm(K(\pi,\thedimm))\to\pi$
induces an isomorphism $H_{\thedimm}(K(\pi,\thedimm))\to\pi$.
\end{lemma}

\begin{proof} First we note that $C_\thedimm(K(\pi,\thedimm))=
Z_\thedimm(K(\pi,\thedimm))$, since $K(\pi,\thedimm)_{\thedimm-1}=\{0\}$.
Then $\ev$ is easily seen to be surjective, and so it remains to prove
that $\ker(\ev)=B_\thedimm(K(\pi,\thedimm))$.

Let us consider $z\in K(\pi,\thedimm)_{\thedimm+1}=
Z^{\thedimm}(\Delta^{\thedimm+1};\pi)$; thus, $z$
is given by the $(\thedimm+2)$-tuple $(g_0,\ldots,g_{\thedimm+1})$,
where $g_i$ is the value of $z$ on $\partial_i\Delta^{\thedimm+1}$,
and the cocycle condition reads $\sum_{i=0}^{\thedimm+1}(-1)^ig_i=0$
(in $\pi$). On the other hand, considering $z$ as a chain in
$C_{\thedimm+1}(K(\pi,\thedimm))$, we have $\diff z=
\sum_{i=0}^{\thedimm+1}(-1)^i\partial_i z$, and $\partial_i z$
is the $\thedimm$-cochain on $\Delta^\thedimm$ with value $g_i$
(if $g_i=0$, the term $\partial_i z$ is ignored in $\diff z$).
Thus $\ev(\diff z)=\sum_{i=0}^{\thedimm+1}(-1)^ig_i=0$,
and so $B_\thedimm(K(\pi,\thedimm))\subseteq \ker(\ev)$.

For the reverse inclusion, we recall that there is a one-to-one
correspondence, given by the mapping $\ev$, between
the nondegenerate  $\thedimm$-simplices of $K(\pi,\thedimm)$ and the
 nonzero elements of $\pi$. Let us write $\sigma_g$ for the
unique $k$-simplex of $K(\pi,\thedimm)$ with $\ev\sigma_g=g$.
Then a $\thedimm$-chain
$c\in C_\thedimm(K(\pi,\thedimm))$ can be written
as $c=\sum_{g\in\pi\setminus\{0\}}\alpha_g\cdot\sigma_g$,
with finitely many nonzero coefficients $\alpha_g$.
We have $\ev(c)=0$ iff $\sum_{g\in\pi\setminus\{0\}}\alpha_g g=0$
in~$\pi$.

By the above description
of generators of $B_\thedimm(K(\pi,\thedimm))$, and since $\thedimm\ge 1$,
we get that for every $g_1,g_2\in \pi$, the chain
$1\cdot\sigma_{g_1}+1\cdot\sigma_{g_2}$ is homologous to
$1\cdot\sigma_{g_1+g_2}$ (where terms involving $\sigma_0$
are to be ignored). Then by induction we get that
a general chain  $c=\sum_{g\in\pi\setminus\{0\}}\alpha_g\cdot\sigma_g$
is homologous to $1\cdot\sigma_s$, where
$s=\sum_{g\in\pi\setminus\{0\}}\alpha_g g$. In particular,
if $\ev c=0$, then $c$ is homologous to the zero chain,
and so $c\in B_\thedimm(K(\pi,\thedimm))$ as claimed.
\end{proof}

\begin{lemma}\label{l:ev-cone}
The homomorphism \[
h:=\ev+\ev\:\MCone_{\thedimm+1}(\delta_*)=
C_\thedimm (E(\pi,\thedimm))\oplus C_{\thedimm+1}(K(\pi,\thedimm+1))
\to \pi
\] sending $(\sigma,\tau)$ to
$\ev\sigma+\ev\tau$ induces an isomorphism $H_{\thedimm+1}(
\MCone_*(\delta_*))\to\pi$.
\end{lemma}

\begin{proof}
For brevity, we write $E=E(\pi,\thedimm)$ and
$K=K(\pi,\thedimm+1)$ since there are no other
Eilenberg--MacLane spaces in this proof.

In order to claim that $h$ induces a map in homology,
we verify that it vanishes on all boundaries.
Thus, let $(\sigma', \tau')\in\MCone_{\thedimm+2}(\delta_*)$
be a generator,  $\sigma'\in E_{\thedimm+1}$,
$\tau'\in K_{\thedimm+2}$. According to the formula
(\ref{e:conediff}) in Section~\ref{s:mcon} we have
$\diff^{\MCone_*}(\sigma',\tau')=
(-\diff^E\sigma',\delta_*(\sigma')+\diff^K\tau')$.
Since $\tau'$ is a cocycle, we have $\ev(\diff^K\tau')=0$,
as we saw in the proof of Lemma~\ref{eval}.
Moreover, it is easily checked that
$\ev(\diff^E\sigma')=\ev(\delta_*(\sigma'))$.
It follows that $h$ indeed vanishes on boundaries
and induces a homomorphism
$h_*\:H_{\thedimm+1}(\MCone_*(\delta_*))\to\pi$.

Now we consider the canonical inclusion
$C_*(K)\to \MCone_*(\delta_*)$, which is a chain map,
and thus it induces a map in homology, as in the following
diagram (here we use that $C_{\thedimm+1}(K)=Z_{\thedimm+1}(K)$
and $C_{\thedimm}(E)=Z_\thedimm(E)$):
\[\xymatrix{
& C_{\thedimm+1}(K)\ar[r]^{i~~~~~} \ar[d] & \leftbox{C_{\thedimm}(E)\oplus
C_{\thedimm+1}(K)}{=\MCone_{\thedimm+1} (\delta_*)} \ar[d] \\
\pi&H_{\thedimm+1}(K) \ar[l]^{\cong~~~~}_{\ev_*~~~} \ar[r]_-\cong^{i_*~~~~} &
H_{\thedimm+1}(\MCone_*(\delta_*))
}\]
Here $\ev_*$ on the left in the bottom row is the
isomorphism induced by $\ev$ as in Lemma~\ref{eval}.

The map $i_*$ is an isomorphism by the long exact homology sequence
of the pair $(\MCone(\delta_*),C_*(K))$,
 because the quotient $\MCone(\delta_*)/C_*(K)\cong C_*(E)^\uparrow$
is the shift of the chain complex of a contractible simplicial set $E$
(see e.g.~\cite[Proposition~21.5, Theorem~23.10]{May:SimplicialObjects-1992}),
and thus
 all homology groups of this quotient vanish except for the one in dimension~1.

Finally, it suffices to verify that $h_*i_*=\ev_*$,
 but this is clear, since the composition on the left maps
$[\tau]\xmapsto{~i_*~}[(0,\tau)]\xmapsto{~h_*~}\ev\tau$.
\end{proof}

\subsection{Correctness of the algorithm}\label{s:correctness}

Here we provide a proof of correctness for the algorithm above.
It uses more or less standard methods, but we do not know of
an accessible presentation in the literature.
In this part, we are going to use somewhat more advanced topological notions
without defining them; we refer to standard textbooks, such as~\cite{Hatcher}.

We assume the correctness of the algorithm for $\thedim-1$.
For brevity
we write $\varphi=\varphi_{\thedimm-1}$, $P=P_{\thedimm-1}$,
$K=K(\pi_{\thedimm},\thedimm+1)$, and $E=E(\pi_{\thedimm},\thedimm)$.

\heading{Checking $\pi_\thedimm\cong\pi_\thedimm(Y)$. }
We recall that the algorithm sets up an isomorphism $H_{\thedimm+1}(\EC^M_*)\cong\pi_\thedimm$;
thus we need to verify that $H_{\thedimm+1}(\EC^M_*)\cong\pi_{\thedimm}(Y)$.
Let $\MCyl\varphi$ be the mapping cylinder
of $\varphi\:Y\to P$, i.e.,
 the simplicial set $(Y\times \Delta^1 \cup P)/\sim$,
 where $\sim$ is the equivalence identifying $(y,0)$ with $f(y)$,
$y\in Y$. Let us also identify  $Y$ with $Y\times\{1\}$,
the ``top copy'' of $Y$ in~$\MCyl\varphi$.

Using the Eilenberg--Zilber reduction, it is easy to check
that the chain complex of the pair $(\MCyl\varphi,Y)$
has a reduction to $\MCo_*=\MCone_*(\varphi_*)$.
Hence
\[H_{\thedimm+1}(\MCyl\varphi,Y)\cong H_{\thedimm+1}(\MCo_*)\cong H_{\thedimm+1}(\EC^M_*).\]
Using the fact that $\MCyl\varphi$ is homotopy equivalent to $P$,
and the assumption $\pi_i(P)=0$ for $i\ge k$,
the long exact sequence  of homotopy groups for the pair
$(\MCyl\varphi,Y)$ yields that this pair is $\thedimm$-connected
and $\pi_{\thedimm}(Y)\cong \pi_{\thedimm+1}(\MCyl\varphi,Y)$.
Due to the $\thedimm$-connectedness of $(\MCyl\varphi,Y)$,  the
Hurewicz isomorphism yields $\pi_{\thedimm+1}(\MCyl\varphi,Y)
\cong H_{\thedimm+1}(\MCyl\varphi,Y)$.
Putting all these isomorphisms together we obtain
$\pi_{\thedimm}(Y)\cong\pi_{\thedimm}$,
as desired.

\heading{The cochain $\kappa_{\thedimm-1}^\ef$ is a cocycle. }
We recall that $\kappa_{\thedimm-1}^\ef$ is the composition
\[\kappa_{\thedim-1}^\ef\:\EC_{\thedimm+1}^{P_{\thedimm-1}}\hookrightarrow \EC_{\thedimm}^{Y}\oplus \EC_{\thedimm+1}^{P_{\thedimm-1}}=\EC^M_{\thedim+1}\xrightarrow{\ \rho\ }\pi_\thedim\]
The inclusion, being a chain map, preserves boundaries, and $\rho$, by definition, vanishes on them. Thus the composite $\kappa_{\thedim-1}^\ef$ also vanishes on boundaries and is indeed a cocycle.

\heading{The map $\varphi_{\thedimm}$ takes values in $P_{\thedimm}$. }
First we will need a description of the cocycle $\kappa_{\thedimm-1}$ similar to that of $\lambda_\thedimm$. Namely, the remark following the proof of Proposition~\ref{p:mcone} says that $\kappa_{\thedimm-1}$ can be also obtained as a restriction of $\rho f$ from Step~\ref{step:nonef} to $C_{\thedimm+1}(P_{\thedimm-1})$.\footnote{On the other hand, $\lambda_\thedimm$
in general cannot be computed solely from $\lambda_\thedimm^\ef$ and the effective homology of $Y$. A notable exception to this is when $C_*(Y)=\EC_*^Y$, as happens e.g.~for finite simplicial complexes. In this case we have $\lambda_\thedimm=\lambda_\thedimm^\ef$.} Thus, denoting the inclusions of the two summands by $i\:C_{\thedim+1}(P_{\thedim-1})\to M_{\thedim+1}$ and $j\:C_\thedim(Y)\to M_{\thedim+1}$, we can write $\kappa_{\thedimm-1}=\rho f i$ and $\lambda_\thedimm=\rho f j$.



Now, we will verify that the map $\varphi_{\thedimm}=(\varphi,\ell_\thedimm)\:Y\to P\times E$ has image in the pullback~$P_{\thedimm}$, which amounts to showing that
$\kkk_{\thedimm-1}\varphi=\cobo\ell_\thedimm$.
This will follow easily from the following equality of cochains in
$C^{\thedimm+1}(Y;\pi_\thedimm)$:
\begin{equation}\label{e:cocha-level}
\kappa_{\thedim-1}\varphi_*=\lambda_\thedim\diff^Y,
\end{equation}
where $\varphi_*\:C_{\thedimm+1}(Y)\to C_{\thedimm+1}(P)$
is the chain map induced by $\varphi$ and $\diff^Y$
is the differential in $C_*(Y)$.
We have
$\kappa_{\thedimm-1}\varphi_*-\lambda_\thedimm\diff^Y=\rho f(i\varphi_*-j\diff^Y)$.
As above, $\rho f$ maps boundaries
in $\MCo_*$ to $0$, so it suffices to show that the images
of $i\varphi_*-j\diff^Y$ are boundaries---but by the formula
for the differential in the algebraic mapping cone,
we have that for every $\sigma\in Y_{\thedimm+1}$,
$(i\varphi_*-j\diff^Y)(\sigma)=\diff^{\MCo_*}(\sigma,0)$
is indeed a boundary.

Using Lemma~\ref{l:ev*}, we find that
$\kappa_{\thedim-1}\varphi_*=(\ev (\kkk_{\thedimm-1})_*)\varphi_*=
\ev(\kkk_{\thedimm-1}\varphi)_*$. It is also easy to verify from the definitions
that $\ev(\delta\ell_\thedimm)_*= \lambda_\thedimm\diff^Y$, and
so  the equality (\ref{e:cocha-level}) of cochains
yields the desired equality $\kkk_{\thedimm-1}\varphi=\cobo\ell_\thedimm$
of simplicial maps.

\heading{The maps induced by $\varphi_\thedimm$ in homotopy. }
 Considering the long exact
sequence of homotopy groups
of the fibration $K(\pi_{\thedimm},\thedimm)\to P_{\thedimm}
\to P$ and using the assumption $\pi_i(P)=0$ for $i\ge\thedimm$,
it is straightforward to check that $\pi_i(P_\thedimm)=0$
for $i\ge\thedimm+1$, and that the maps $\pi_i(Y)\to\pi_i(P_\thedimm)$
induced by $\varphi_\thedimm$ are isomorphisms for $i\le \thedimm-1$.
For establishing  condition (i) in the definition of
a Postnikov system, it remains to verify that
$(\varphi_{\thedimm})_*\:\pi_\thedimm(Y)\to\pi_\thedimm(P_\thedimm)$
is an isomorphism as well.

To this end, we begin with the diagram
\[\xymatrix{
Y \ar[d]^\varphi \ar[r]^{\varphi_\thedimm} & P_{\thedimm} \ar[r] \ar[d]^{p_{\thedimm}}
& E \ar[d]^\delta \\
P \ar@{=}[r] & P \ar[r] & K,
}
 \]
where the right square is the pullback diagram defining $P_\thedimm$.
Next, we replace each of the spaces in the bottom row with
the mapping cylinder of the respective vertical map, so that
the vertical maps become inclusions (of the domain in the cylinder);
the horizontal maps of the cylinders are then induced in a canonical way.
\begin{equation}\label{e:two_squares}
\xy *!C\xybox{\xymatrix{
Y \ar[d] \ar[r]^{\varphi_\thedimm}
 & P_\thedimm \ar[r] \ar[d] & E \ar[d]\\
\MCyl \varphi \ar[r] & \MCyl p_\thedimm \ar[r] & \MCyl \cobo
}}\endxy
\end{equation}

\begin{lemma}\label{l:cylindrak}
The map $\pi_{\thedimm+1}(\MCyl \varphi ,Y)\to\pi_{\thedimm+1}
(\MCyl  p_{\thedimm} ,P_{\thedimm})$
induced by the left square of the last diagram is an isomorphism.
\end{lemma}

We first finish the proof of correctness of the algorithm assuming
the lemma. We consider the long exact sequences coming from the pairs
$(\MCyl\varphi,Y)$ and $(\MCyl p_{\thedimm},P_{\thedimm})$:
\[\xymatrix@C=14pt{
0=\pi_{\thedimm+1}(\MCyl\varphi)\ar[r]
\ar[d]^{\cong}
& \pi_{\thedimm+1}(\MCyl\varphi,Y) \ar[d]^{\cong} \ar[r]
& \pi_{\thedimm}(Y) \ar[d]^{\varphi_{\thedimm*}} \ar[r]
& \pi_{\thedimm}(\MCyl\varphi)=0 \ar[d]^{\cong}\\
0=\pi_{\thedimm+1}(\MCyl p_{\thedimm})
\ar[r] & \pi_{\thedimm+1}(\MCyl p_{\thedimm},P_{\thedimm}) \ar[r]
& \pi_{\thedimm}(P_{\thedimm}) \ar[r] & \pi_{\thedimm}(\MCyl p_{\thedimm})=0
 }\]
The second vertical isomorphism is proved in the lemma and the other two
follow from $\pi_i(P)=0$ for $i\ge \thedimm$, since both of the
cylinders deform onto the base $P$. Then the five-lemma implies
that $\varphi_{\thedimm*}$ is an isomorphism on $\pi_{\thedimm}$,
which completes the proof of condition (i) from
the definition of a Postnikov system. All that remains is
to prove the lemma.

\begin{proof}[Proof of Lemma~\ref{l:cylindrak}]
We will show that both the right square and the composite square induce
an isomorphism in the relative homotopy groups of
the vertical pairs in dimension $\thedimm+1$.
We start with the composite square.

 Since both $(\MCyl\varphi,Y)$ and $(\MCyl\delta,E)$
are $\thedimm$-connected, it suffices to prove that the square induces an
isomorphism on the $(\thedimm+1)$-st homology group. We use that the
chain complexes of these pairs are isomorphic to the respective (reduced)
algebraic mapping cones. We find that the chain map $C_*(\MCyl\varphi,Y)\to
C_*(\MCyl\delta,E)$ is actually $\ell_{\thedim*}\oplus \kkk_{(\thedimm-1)*}$;
this can be seen using the diagram
\[
\xymatrix{
Y \ar[r]^{\ell_\thedim} \ar[d]_\varphi& E\ar[d]^{\cobo}\\
P\ar[r]^{\kkk_{\thedimm-1}} & K.
}
\]
Then we consider the diagram
\[\xymatrix{
C_{\thedimm}(Y)\oplus C_{\thedimm+1}(P)
\ar[rr]^-{\ell_{\thedim*}\oplus \kkk_{(\thedimm-1)*}}
\ar[rd]_-{\lambda_{\thedim}+\kappa_{\thedim-1}} & & C_{\thedimm}(E)\oplus C_{\thedimm+1}(K),
\ar[ld]^-{\ev+ \ev} \\
& \pi_\thedimm
}\]
which commutes in view of Lemma~\ref{l:ev*}.
The left map $\lambda_\thedimm+\kappa_{\thedimm-1}$ equals
$\rho f$, and since both $\rho$ and $f$
induce isomorphisms in homology, so does $\lambda_\thedimm+\kappa_{\thedimm-1}$.
The map $\ev+\ev$ induces an isomorphism in homology
by Lemma~\ref{l:ev-cone}.
Therefore the same is true for the horizontal map, and hence the composite
square in the diagram \eqref{e:two_squares} induces an isomorphism in the $(k+1)$st
homotopy groups of the vertical pairs, as claimed.

It remains to study the right square.
 Before we passed to mapping cylinders, the original square was a pullback.
The original vertical maps are fibrations, and consequently, the induced map
on fibers (which are both $K(\pi_{\thedimm},\thedimm)$) is an isomorphism.
Next, there is an isomorphism
$\pi_{\thedimm+1}(\MCyl p_{\thedimm},P_{\thedimm})\cong\pi_\thedimm
(\fib p_{\thedimm})$,
and a similar one for $\delta$.
From their description below it will be apparent that this isomorphism
is natural so that the square
\[\xymatrix{
\pi_{\thedimm+1}(\MCyl p_{\thedimm},P_{\thedimm}) \ar[r] \ar[d]_-\cong & \pi_{\thedimm+1}
(\MCyl \delta,E)
\ar[d]^-\cong \\
\pi_{\thedimm}(\fib p_{\thedimm}) \ar[r]_-\cong
& \pi_{\thedimm}(\fib \delta).
}\]
commutes.
 We will thus be able to conclude that
$\pi_{\thedimm+1}(\MCyl p_{\thedimm},P_{\thedimm})\to
\pi_{\thedimm+1}(\MCyl \delta,E)$ is indeed an isomorphism as required.

The required map $\pi_{\thedimm+1}(\MCyl p_{\thedimm},P_{\thedimm})\to
\pi_{\thedimm}(\fib p_{\thedimm})$ is defined via representatives.
To this end, we represent an element
of $\pi_{\thedimm+1}(\MCyl p_{\thedimm},P_{\thedimm})$ by a map
 $f\:I^{\thedimm+1}\ra\MCyl p_{\thedimm}$
that sends the face $I^{\thedimm}$ (where the last coordinate is zero)
to $P_{\thedimm}$ and the union of the remaining faces, which we denote
by $J^{\thedimm}$, to the basepoint (here $I^{\thedimm+1}$ denotes the
unit cube). Now composing $f$ with the projection
$\mathop\mathrm{pr}\:\MCyl p_{\thedimm}\ra P$ we obtain
$g=\mathop\mathrm{pr}\circ f\:I^{\thedimm+1}\ra P$, which we lift along
$p_{\thedimm}$ to $\widetilde g\:I^{\thedimm+1}\ra P_{\thedimm}$.
One may prescribe the values on all the faces except for one.
Here we decide that $\widetilde g$ agrees with $f$ on the only interesting
face $I^{\thedimm}$ and that it is constant onto the basepoint on
the neighboring faces. Finally, the restriction to the remaining face
(opposite to $I^{\thedimm}$) gives us a map
$\widetilde g_1\:I^{\thedimm}\ra\fib p_{\thedimm}$,
and this is the representative of the image of $[f]$ under the desired map
$\pi_{\thedimm+1}(\MCyl p_{\thedimm},P_{\thedimm})
\to\pi_{\thedimm}(\fib p_{\thedimm})$.

It remains to show that this map
is indeed an isomorphism.
For this, we consider the following diagram, whose top row is the long exact
sequence of the pair $(\MCyl p_{\thedimm},P_{\thedimm})$, and whose bottom row
is associated with the fibration $p_{\thedimm}$.
\[\xymatrix{
\cdots \ar[r] & \pi_{\thedimm+1}P_{\thedimm} \ar[r] \ar[d]^\id & \pi_{\thedimm+1}\MCyl p_{\thedimm} \ar[r] \ar[d]^\cong \POS[];[dr]**{}?(.5)*{\textrm{(A)}} & \pi_{\thedimm+1}(\MCyl p_{\thedimm},P_{\thedimm}) \ar[r] \ar[d] \POS[];[dr]**{}?(.5)*{\textrm{(B)}} & \pi_{\thedimm}P_{\thedimm} \ar[r] \ar[d]^\id & \pi_{\thedimm}\MCyl p_{\thedimm} \ar[r] \ar[d]^\cong & \cdots \\
\cdots \ar[r] & \pi_{\thedimm+1}P_{\thedimm} \ar[r] & \pi_{\thedimm+1}P \ar[r] & \pi_{\thedimm}\fib p_{\thedimm} \ar[r] & \pi_{\thedimm}P_{\thedimm} \ar[r] & \pi_{\thedimm}P \ar[r] & \cdots
}\]
The isomorphism will follow from the five-lemma once we show that the
squares (A) and (B) commute up to a sign.
The square (A) anticommutes because the path through the bottom left
corner consists of lifting $g$ as above but with the restriction to
$J^{\thedimm}$ being constant onto the basepoint.
One can obtain this by first flipping $I^{\thedimm+1}$ along the last
coordinate and then lifting as above. The flipping amounts to multiplication
by $-1$ on $\pi_{\thedimm+1}(P)$. The square denoted by (B) commutes by
an easy inspection: the map $\widetilde g_1$ is homotopic inside
$P_{\thedimm}$ with $f|_{I^{\thedimm}}$ (the image in the top right corner
of that square), the required homotopy being~$\widetilde g$.
\end{proof}

\section{The extension problem}\label{s:ext}

\begin{proof}[Proof of Theorem~\ref{t:ext}.]
Here we prove the result about testing extendability of a map using tools from~\cite{CKMSVW11}. We are given simplicial sets $A\subseteq X$ and $Y$
and a simplicial map $f\:A\to Y$, where $X$ is finite,
 $\dim X\le 2\thedim-1$, and $Y$ is $(\thedim-1)$-connected.

First, by \cite[Theorem~7.6.22]{Spanier}, a continuous extension of $f$ to
$X$ exists, under these assumptions,
if and only if the composition $\varphi_{2\thedim-2}f\:A\to P_{2\thedim-2}$
admits a continuous extension to $X$,
where $\varphi_{2\thedim-2}\:Y\to P_{2\thedim-2}$ is the map in
the Postnikov system of $Y$. By the homotopy extension property,
this happens precisely when there exists a map $X\to P_{2\thedim-2}$,
whose restriction to $A$ is homotopic to $\varphi_{2\thedim-2}f$.
In terms of homotopy classes of maps, this is if and only
$[\varphi_{2\thedim-2}f]$ lies in the image of the restriction map
$\rho\:[X,P_{2\thedim-2}]\to[A,P_{2\thedim-2}]$.

The algorithm in Corollary~\ref{c:[XY]} for computing
$[X,Y]$ actually computes $[X,P_{2\thedim-2}]$. The isomorphism
$[X,Y]\cong [X,P_{2\thedim-2}]$ holds only for $\dim X\le2\thedim-2$,
but the computation of $[X,P_{2\thedim-2}]$ works correctly
for $X$ of arbitrary dimension.
Thus, in the setting of Theorem~\ref{t:ext}, we can compute
the Abelian group $[X,P_{2\thedim-2}]$ represented by generators,
which are specified
as simplicial maps\footnote{Actually, a more compact
\emph{cochain representation} is used in \cite{CKMSVW11},
but for our purposes, we can think of explicit simplicial maps.}
$X\to P_{2\thedim-2}$, and relations (it is \emph{fully effective}
in the terminology of \cite{CKMSVW11}).

For the simplicial subset $A\subseteq X$, we similarly compute
$[A,P_{2\thedim-2}]$.
As we recall from \cite{CKMSVW11}, the  group operation
in $[X,P_{2\thedim-2}]$ is
induced by an operation $\boxplus$ on $\SM(X,P_{2\thedim-2})$, which
is defined simplexwise (i.e., $(f\boxplus g)(\sigma)=
f(\sigma)\boxplus g(\sigma)$). This easily implies that the
restriction map $\rho$ is a group homomorphism.

Given an element (homotopy class) $[g]\in [X,P_{2\thedim-2}]$,
represented by a simplicial map $g$, we consider the
restriction $g|_A$ as a representative of an element of $[A,P_{2\thedim-2}]$,
and we can express it using the generators of $[A,P_{2\thedim-2}]$.
Thus, $\rho$ is polynomial-time computable, and we can compute the
image $\im\rho$ as a subgroup of $[A,P_{2\thedim-2}]$ (by computing
the images of the generators of $[X,P_{2\thedim-2}]$ and the
subgroup generated by them).

Then, given a simplicial map $f\:A\to Y$, we compute the corresponding
element $[\varphi_{2\thedim-2}f]\in [A,P_{2\thedim-2}]$ and test (in
polynomial time) whether it lies in the image of $\rho$.
This is the desired algorithm for testing the extendability of~$f$.

In case $\dim X\leq 2\thedim-2$ we have $[X,Y]\cong[X,P_{2\thedim-2}]$
and $[A,Y]\cong[A,P_{2\thedim-2}]$. Thus, if $x\in \im\rho$,
we can compute the preimage $\rho^{-1}(x)$ as a coset in
$[X,P_{2\thedim-2}]$ (since we have $\rho$ represented by a matrix),
and this coset is isomorphic to $[X,Y]_f$ as needed.
This concludes the proof.
\end{proof}

\subsection*{Acknowledgments}

We would like to thank Francis Sergeraert for many useful discussion
and extensive advice; in particular, the current proof of
Lemma~\ref{l:ZtoZm} is mostly due to
him---not speaking of his leading role in the
development of effective homology, which constitutes the foundation
of our algorithmic methods.
We also thank Marek Filakovsk\'y for useful discussions, and finally, to an anonymous referee for many valuable suggestions and comments.

\bibliographystyle{plain}
\bibliography{Postnikov,simplicHomotopy}

\begin{thebibliography}{10}

\bibitem{AlvarezArmarioFrauReal}
V.~\'Alvarez, J.A. Armario, M.D. Frau, and P.~Real.
\newblock {Algebra structures on the comparison of the reduced bar construction
  and the reduced $W$-construction}.
\newblock {\em Commun. Algebra}, 37(10):3643--3665, 2009.

\bibitem{Anick-homotopyhard}
D.~J. Anick.
\newblock {The computation of rational homotopy groups is {\#}$\wp$-hard}.
\newblock {Computers in geometry and topology, Proc. Conf., Chicago/Ill. 1986,
  Lect. Notes Pure Appl. Math. 114, 1--56}, 1989.

\bibitem{BarnesLambe}
D.~W. Barnes and L.~A. Lambe.
\newblock A fixed point approach to homological perturbation theory.
\newblock {\em Proc. AMS}, 112:881--892, 1991.

\bibitem{BrowThm}
E.~H. {Brown, Jr.}
\newblock {Twisted tensor products. I.}
\newblock {\em Ann. Math. (2)}, 69:223--246, 1959.

\bibitem{Brown}
E.~H. B{rown (jun.)}.
\newblock {Finite computability of Postnikov complexes}.
\newblock {\em Ann. Math. (2)}, 65:1--20, 1957.

\bibitem{CKMSVW11}
M.~{\v{C}}adek, M.~Kr\v{c}\'al, J.~Matou\v{s}ek, F.~Sergeraert,
  L.~Vok\v{r}\'{\i}nek, and U.~Wagner.
\newblock Computing all maps into a sphere.
\newblock {\em J. ACM}, 2014.
\newblock To appear. Preprint in arXiv:1105.6257. Extended abstract in
  \emph{Proc. ACM--SIAM Symposium on Discrete Algorithms} (SODA 2012).

\bibitem{ext-hard}
M.~{\v{C}}adek, M.~Kr\v{c}\'al, J.~Matou\v{s}ek, L.~Vok\v{r}\'{\i}nek, and
  U.~Wagner.
\newblock Extendability of continuous maps is undecidable.
\newblock {\em Discr. Comput. Geom.}, 51(1):24--66, 2014.
\newblock Preprint arXiv:1302.2370.

\bibitem{aslep}
M.~{\v{C}}adek, M.~Kr\v{c}\'al, and L.~Vok\v{r}\'{\i}nek.
\newblock Algorithmic solvability of the lifting-extension problem.
\newblock Preprint, arXiv:1307.6444, 2013.

\bibitem{Curtis:SimplicialHomotopyTheory-1971}
E.~B. Curtis.
\newblock Simplicial homotopy theory.
\newblock {\em Advances in Math.}, 6:107--209, 1971.

\bibitem{Dehn:comput-fund-group-2-mani}
M.~Dehn.
\newblock Transformation der kurven auf zweiseitigen fl{\"a}chen.
\newblock {\em Mathematische Annalen}, 72(3):413--421, 1912.

\bibitem{EilenbergMacLane:GroupsHPin1-1953}
S.~Eilenberg and S.~Mac~Lane.
\newblock On the groups of {$H(\Pi,n)$}. {I}.
\newblock {\em Ann. of Math. (2)}, 58:55--106, 1953.

\bibitem{EilenbergMacLane:GroupsHPin2-1954}
S.~Eilenberg and S.~Mac~Lane.
\newblock On the groups {$H(\Pi,n)$}. {II}. {M}ethods of computation.
\newblock {\em Ann. of Math. (2)}, 60:49--139, 1954.

\bibitem{ellis2008HAP}
G.~Ellis.
\newblock Homological algebra programming.
\newblock In {\em Computational group theory and the theory of groups}, volume
  470 of {\em Contemp. Math.}, pages 63--74. Amer. Math. Soc., Providence, RI,
  2008.

\bibitem{Filakovsky-tensor}
M.~Filakovsk{\'y}.
\newblock Effective chain complexes for twisted products.
\newblock Preprint, arXiv: 1209.1240, 2012.

\bibitem{VokriFil-homotopic}
M.~Filakovsk\'y and L.~Vok\v{r}\'{\i}nek.
\newblock Are two given maps homotopic? {A}n algorithmic viewpoint, 2013.
\newblock Preprint, arXiv:1312.2337.

\bibitem{FranekKrcal:RobustSatisfiability}
P.~Franek and M.~Kr\v{c}\'al.
\newblock Robust satisfiability of systems of equations.
\newblock In {\em Proc.\ Ann. ACM-SIAM Symp. on Discrete Algorithms (SODA)},
  2014.

\bibitem{Friedm08}
G.~{Friedman}.
\newblock An elementary illustrated introduction to simplicial sets.
\newblock {\em Rocky Mountain J. Math.}, 42(2):353--423, 2012.

\bibitem{GoerssJardine}
P.~G. Goerss and J.~F. Jardine.
\newblock {\em Simplicial homotopy theory}.
\newblock Birkh{\"a}user, Basel, 1999.

\bibitem{GonRea05}
R.~Gonzalez-Diaz and P.~Real.
\newblock Simplification techniques for maps in simplicial topology.
\newblock {\em J. Symb. Comput.}, 40:1208--1224, October 2005.

\bibitem{Hatcher}
A.~Hatcher.
\newblock {\em Algebraic {T}opology}.
\newblock Cambridge University Press, Cambridge, 2001.
\newblock Electronic version available at
  \url{http://math.cornell.edu/hatcher#AT1}.

\bibitem{KannanBachem}
R.~Kannan and A.~Bachem.
\newblock Polynomial algorithms for computing the {S}mith and {H}ermite normal
  forms of an integer matrix.
\newblock {\em SIAM J. Computing}, 8:499--507, 1981.

\bibitem{Kochman}
S.~O. Kochman.
\newblock {\em {Stable homotopy groups of spheres. A computer-assisted
  approach.}}
\newblock Lecture Notes in Mathematics 1423. Springer-Verlag, Berlin etc.,
  1990.

\bibitem{kozl-surv}
D.~N. Kozlov.
\newblock Chromatic numbers, morphism complexes, and {Stiefel--Whitney}
  characteristic classes.
\newblock In {\em Geometric Combinatorics (E. Miller, V. Reiner, and B.
  Sturmfels, editors)}, pages 249--315. Amer. Math. Soc., Providence, RI, 2007.
\newblock Also arXiv:math/0505563.

\bibitem{Krcal-thesis}
M.~Kr\v{c}\'al.
\newblock {\em Computational Homotopy Theory}.
\newblock PhD thesis, Department of Applied Mathematics, Faculty of Mathematics
  and Physics, Charles University, Prague, 2013.

\bibitem{pKZ1}
M.~Kr\v{c}\'al, J.~Matou\v{s}ek, and F.~Sergeraert.
\newblock Polynomial-time homology for simplicial {Eilenberg--MacLane} spaces.
\newblock {\em J. Foundat. of Comput. Mathematics}, 13:935--963, 2013.
\newblock Preprint, arXiv:1201.6222.

\bibitem{Manning:mani-with-word-problem}
J.~{Manning}.
\newblock {Algorithmic detection and description of hyperbolic structures on
  closed 3-manifolds with solvable word problem}.
\newblock {\em {Geometry and Topology}}, 6:1--26, 2002.

\bibitem{Matousek:BorsukUlam-2003}
J.~Matou{\v{s}}ek.
\newblock {\em Using the {B}orsuk-{U}lam theorem (revised 2nd printing)}.
\newblock Universitext. Springer-Verlag, Berlin, 2007.

\bibitem{MatousekTancerWagner:HardnessEmbeddings-2011}
J.~Matou{\v{s}}ek, M.~Tancer, and U.~Wagner.
\newblock {Hardness of embedding simplicial complexes in $\R^d$}.
\newblock {\em J. Eur. Math. Soc.}, 13(2):259--295, 2011.

\bibitem{Mat-homotopyW1}
J.~Matou\v{s}ek.
\newblock Computing higher homotopy groups is {$W[1]$}-hard.
\newblock {\em Fundamenta Informaticae}, 2014.
\newblock To appear.

\bibitem{May:SimplicialObjects-1992}
J.~P. May.
\newblock {\em Simplicial objects in algebraic topology}.
\newblock Chicago Lectures in Mathematics. University of Chicago Press,
  Chicago, IL, 1992.
\newblock Reprint of the 1967 original; the page numbers do not quite agree
  with the 1967 edition.

\bibitem{Moore59}
J.~C. Moore.
\newblock Homological algebra and the cohomology of classifying spaces (in
  {F}rench).
\newblock {\em S\'eminaire H.~Cartan 1959/60}, exp. 7, 1959.

\bibitem{Munkres}
J.~R. Munkres.
\newblock {\em Elements of Algebraic Topology}.
\newblock Addison-Wesley, Reading, MA, 1984.

\bibitem{Novikov:UndecidabilityWordProblem-1955}
P.~S. Novikov.
\newblock Ob algoritmi\v cesko\u\i\ nerazre\v simosti problemy to\v zdestva
  slov v teorii grupp ({O}n the algorithmic unsolvability of the word problem
  in group theory).
\newblock {\em Trudy Mat. inst. im. Steklova}, 44:1--143, 1955.

\bibitem{Ravenel}
D.~C. Ravenel.
\newblock {\em Complex Cobordism and Stable Homotopy Groups of Spheres (2nd
  ed.)}.
\newblock Amer. Math. Soc., 2004.

\bibitem{RealThesis}
P.~Real.
\newblock Algorithms for computing effective homology of classifying spaces (in
  {S}panish).
\newblock PhD. Thesis, Facultad de Mathem\'aticas, Univ. de Sevilla, 1993.
\newblock Available online at
  \url{http://fondosdigitales.us.es/media/thesis/1426/C_043-139.pdf}.

\bibitem{Real96}
P.~Real.
\newblock An algorithm computing homotopy groups.
\newblock {\em Mathematics and Computers in Simulation}, 42:461---465, 1996.

\bibitem{real-assoc}
P.~Real.
\newblock {Homological perturbation theory and associativity}.
\newblock {\em Homology Homotopy Appl.}, 2:51--88, 2000.

\bibitem{RomeroEllisRubio}
A.~Romero, G.~Ellis, and J.~Rubio.
\newblock Interoperating between computer algebra systems: computing homology
  of groups with {K}enzo and {GAP}.
\newblock In {\em Proc. ISAAC, ACM, New York}, pages 303--310, 2009.
\newblock Available on-line at
  \url{http://hamilton.nuigalway.ie/preprints/sigproc-sp.rev1.pdf}.

\bibitem{RomeroRubioSergeraert}
A.~Romero, J.~Rubio, and F.~Sergeraert.
\newblock {Computing spectral sequences}.
\newblock {\em J. Symb. Comput.}, 41(10):1059--1079, 2006.

\bibitem{SergRomEffHmtp}
A.~Romero and F.~Sergeraert.
\newblock Effective homotopy of fibrations.
\newblock {\em Applicable Algebra in Engineering, Communication and Computing},
  23(1-2):85--100, 2012.

\bibitem{RouneSaensdeCabezon}
B.~H. Roune and E.~{S{\'a}enz de Cabez{\'o}n}.
\newblock Complexity and algorithms for {E}uler characteristic of simplicial
  complexes.
\newblock Preprint arXiv:1112.4523, \url{http://arxiv.org/pdf/1112.4523v1},
  2011.

\bibitem{Rubio-thesis}
J.~Rubio.
\newblock Effective homology of iterated loop spaces (in {F}rench).
\newblock PhD. Thesis, Univ. Grenoble~I, 1991.

\bibitem{RubioSergeraert:ConstructiveAlgebraicTopology-2002}
J.~Rubio and F.~Sergeraert.
\newblock Constructive algebraic topology.
\newblock {\em Bull. Sci. Math.}, 126(5):389--412, 2002.

\bibitem{SergRub-homtypes}
J.~Rubio and F.~Sergeraert.
\newblock Algebraic models for homotopy types.
\newblock {\em Homology, Homotopy and Applications}, 17:139--160, 2005.

\bibitem{SergerGenova}
J.~Rubio and F.~Sergeraert.
\newblock Constructive homological algebra and applications.
\newblock Preprint, arXiv:1208.3816, 2012.
\newblock Written in 2006 for a MAP Summer School at the University of Genova.

\bibitem{RS-funcform}
J.~Rubio and F.~Sergeraert.
\newblock Homology of iterated loop spaces.
\newblock Manuscript in preparation, 2012.

\bibitem{Schoen-effectivetop}
R.~Sch{\"o}n.
\newblock {Effective algebraic topology}.
\newblock {\em Mem. Am. Math. Soc.}, 451:63 p., 1991.

\bibitem{Sergeraert:ComputabilityProblemAlgebraicTopology-1994}
F.~Sergeraert.
\newblock The computability problem in algebraic topology.
\newblock {\em Adv. Math.}, 104(1):1--29, 1994.

\bibitem{SergerTrieste}
F.~Sergeraert.
\newblock Introduction to combinatorial homotopy theory.
\newblock Available at
  \url{http://www-fourier.ujf-grenoble.fr/~sergerar/Papers/}, 2008.

\bibitem{Shih}
Weishu Shih.
\newblock {Homologie des espaces fibres}.
\newblock {\em Publ. Math. de l'IHES}, 13:93--176, 1962.

\bibitem{skopenkov-survey}
A.~B. Skopenkov.
\newblock Embedding and knotting of manifolds in {E}uclidean spaces.
\newblock In {\em Surveys in contemporary mathematics}, volume 347 of {\em
  London Math. Soc. Lecture Note Ser.}, pages 248--342. Cambridge Univ. Press,
  Cambridge, 2008.

\bibitem{smith-mstructures}
J.~R. Smith.
\newblock {m}-{S}tructures determine integral homotopy type.
\newblock Preprint, arXiv:math/9809151v1, 1998.

\bibitem{Spanier}
E.~H. Spanier.
\newblock {\em Algebraic topology}.
\newblock McGraw Hill, 1966.

\bibitem{Steenrod:CohomologyOperationsObstructionsExtendingContinuousFunctions%
-1972}
N.~E. Steenrod.
\newblock Cohomology operations, and obstructions to extending continuous
  functions.
\newblock {\em Advances in Math.}, 8:371--416, 1972.

\bibitem{Storjohann:NearOptimalAlgorithmsSmithNormalForm-1996}
A.~Storjohann.
\newblock Near optimal algorithms for computing {Smith} normal forms of integer
  matrices.
\newblock In {\em International Symposium on Symbolic and Algebraic
  Computation}, pages 267--274, 1996.

\bibitem{vokrinek:odd-spheres}
L.~{Vok{\v r}{\'{\i}}nek}.
\newblock {Decidability of the extension problem for maps into odd-dimensional
  spheres}.
\newblock Preprint, arXiv:1312.2474, January 2014.

\bibitem{Vokrinek-algoheaps}
L.~Vok\v{r}\'{\i}nek.
\newblock Computing the abelian heap of unpointed stable homotopy classes of
  maps, 2013.
\newblock Preprint, arXiv:1312.2474.

\bibitem{Vokrinek}
L.~Vok\v{r}\'{\i}nek.
\newblock Constructing homotopy equivalences of chain complexes of free
  {${\mathbb{Z}}G$}-modules.
\newblock Preprint, arXiv:1304.6771, 2013.

\end{thebibliography}

\end{document}